\documentclass[a4paper,onecolumn,11pt,accepted=2024-03-06]{quantumarticle}
\pdfoutput=1
\usepackage[utf8]{inputenc}
\usepackage[english]{babel}
\usepackage[T1]{fontenc}
\usepackage{amsmath}
\usepackage{amsthm}
\usepackage{mathtools}
\usepackage[makeroom]{cancel}
\usepackage{amssymb}
\usepackage{hyperref}

\usepackage{tikz}
\usepackage{lipsum}
\usepackage{quantikz}
\usepackage{float}
\usepackage{subcaption}
\usepackage[numbers]{natbib}

\DeclarePairedDelimiter\ceil{\lceil}{\rceil}
\DeclarePairedDelimiter\floor{\lfloor}{\rfloor}
\newtheorem{theorem}{Theorem}
\newtheorem{lemma}{Lemma}
\newtheorem{proposition}{Proposition}
\newtheorem{corollary}{Corollary}
\newtheorem{example}{Example}
\newtheorem{conjecture}{Conjecture}
\newtheorem{definition}{Definition}

\newcommand{\thicktilde}[1]{\mathbf{\tilde{\text{$#1$}}}}

\begin{document}

\title{Quantum advantage in temporally flat measurement-based quantum computation}

\author{Michael de Oliveira}
\email{michael.oliveira@inl.int}
\affiliation{University of Minho, Department of Computer Science, Braga, Portugal}
\affiliation{INESC TEC, Braga, Portugal}
\affiliation{International Iberian Nanotechnology Laboratory (INL) Av. Mestre Jose Veiga, 4715-330, Braga, Portugal}
\author{Luís Soares Barbosa}
\affiliation{University of Minho, Department of Computer Science, Braga, Portugal}
\affiliation{INESC TEC, Braga, Portugal}
\affiliation{International Iberian Nanotechnology Laboratory (INL) Av. Mestre Jose Veiga, 4715-330, Braga, Portugal}
\author{Ernesto F. Galvão}
\affiliation{International Iberian Nanotechnology Laboratory (INL) Av. Mestre Jose Veiga, 4715-330, Braga, Portugal}
\affiliation{Instituto de Física, Universidade Federal Fluminense Av. Gal. Milton Tavares de Souza s/n, Niterói, RJ, 24210-340, Brazil}
\maketitle

\begin{abstract}
 Several classes of quantum circuits have been shown to provide a quantum computational advantage under certain assumptions. The study of ever more restricted classes of quantum circuits capable of quantum advantage is motivated by possible simplifications in experimental demonstrations. In this paper we study the efficiency of measurement-based quantum computation with a completely flat temporal ordering of measurements. We propose new constructions for the deterministic computation of arbitrary Boolean functions, drawing on correlations present in multi-qubit Greenberger, Horne, and Zeilinger (GHZ) states. We characterize the necessary measurement complexity using the Clifford hierarchy, and also generally decrease the number of qubits needed with respect to previous constructions. In particular, we identify a family of Boolean functions for which deterministic evaluation using non-adaptive MBQC is possible, featuring quantum advantage in width and number of gates with respect to classical circuits.
\end{abstract}

\newpage 
\tableofcontents

\section{Introduction}
One of the primary motivations for studying quantum information and computation, along with the search for practical advantage, is to clarify 
the longstanding question of what generates the classical-quantum separation in information processing power \cite{DevonAaronson21}. Phenomena such as non-locality and contextuality have been identified as possible sources of quantum computational advantage \cite{Jozsa03,Howard2014,Bermejo17}, among others \cite{Galvao05,man12,Grover98}. 
It appears that no single phenomenon can be associated with all forms of quantum information processing advantage, so it is important to identify and study quantum advantage in different models and regimes.

In this pursuit, the measurement-based quantum computation (MBQC) model, first presented by Raussendorf and Briegel in \cite{Raussendorf01} as a sequence of adaptively selected single-qubit measurements on a highly entangled quantum state, 
is a natural setting to study quantum-to-classical separations. In particular, it 
allows for demonstrations that specific structures of non-local correlations between qubits are necessary for universal quantum computation \cite{UniversalRes06}. Nevertheless, a temporal structure for the measurements is simultaneously imposed to 
achieve universality \cite{AndersBrown2009,Danos06,Browne_2007}. 
This means statically selected measurements are not computationally expressive enough to 
implement arbitrary quantum algorithms, even with access to highly correlated resource states. A judicious choice of side classical computation and control is essential for universality in the MBQC model. 

The importance of time structure in measurement-based quantum computation is not fully understood. For instance, it is conjectured that classical computers cannot efficiently simulate instantaneous quantum polynomial (IQP) circuits \cite{Bremner16}, 
despite the fact that these circuits of commuting gates have no temporal structure, and can be implemented with no adaptivity in MBQC
\cite{Matty14}. 
These results have motivated the study of temporally flat computation for demonstrations of quantum advantage \cite{Bremner2016,Novo2021quantumadvantage,Hayashi_2019}. Despite having an output distribution that is hard to simulate classically, no practical application for IQP circuits has been found. Even for quantum circuits with temporal order, constant improvements in classical simulation techniques contribute to enlarging the classes of circuits that are classically efficiently simulable.

This paper studies non-adaptive MBQC computations which have no temporal structure for the measurement operators. In this setting, there exist
circuits with quantum advantage, defined as 
those that cannot be efficiently simulated classically
\cite{Vega18,Miller17}. Here we will be interested in investigating how this model can be applied to decision problems. Previously, in 
references \cite{AndersBrown2009,Hoban2011a,Mori2019,Frembs22,Mackeprang22} this model was studied with the goal of evaluating Boolean functions. In the next section we briefly describe the main results we obtain here, relating to the use of non-adaptive MBQC for the evaluation of Boolean functions.

\subsection{Our results}

For temporally unstructured MBQC computations, we propose new constructions that synthesize quantum circuits for the deterministic evaluation of Boolean functions. These constructions reduce the number of qubits required in the GHZ resource states used, especially for the case of symmetric Boolean functions. We also improve the circuit synthesis process, removing an exponential scaling of previous constructions on the degree of the Boolean function \cite{Mori2019,Frembs22}. Regarding the complexity of the 
single-qubit measurements, we also characterize the maximum level of the Clifford hierarchy required for deterministic Boolean function evaluation.

\begin{theorem}\label{clif1} (Informal)
Any Boolean function $f$ can be evaluated deterministically in the non-adaptive MBQC model 
using measurement operators of the $\mathsf{deg}(f)$-level of the Clifford hierarchy.
\end{theorem}

\noindent This result describes 
a measure of complexity for the evaluation of Boolean functions in the model while defining the type of measurement operators necessary to maximally violate a multipartite Bell inequality 
with dichotomic observables and outcomes \cite{Hoban2011a}. Furthermore, this strengthens the idea that the degree of a Boolean function that can be computed with certain given quantum resources introduces a hierarchy for quantum correlations, connecting measures of computational complexity with measures of non-classicality beyond the binary characterization, as suggested in \cite{Frembs22}.

We translate the abstract description of the quantum circuits obtained from the constructions 
into specific circuits based on a fixed gate set. Also, we characterize how the solutions produced by the constructions are related to various circuit complexity measures. This enabled us to prove a quantum over classical separation on the circuit level for the computations of a specific class of Boolean functions with degree two. 

\begin{theorem}\label{main}
(Informal) Any classical circuit with unary and binary Boolean operators with single fan out computes symmetric Boolean functions $f$ with $\mathsf{deg}(f)=2$ with $\Theta(n*\log_2(n))$ gates and circuit width. In contrast, realizations of the non-adaptive MBQC model compute these functions with $\Theta(n)$ gates and circuit width.
\end{theorem}

 The proof compares these quantum circuits with the length of the classical Boolean formulas, which describe classical circuits without a memory. The comparison is indeed very strong in the sense that the quantum circuit also does not use any memory (see Figure \ref{fig:separation} for a 
 pictorial representation of this subclass of circuits). Interestingly, this result also has an interpretation in terms of Bell inequalities. In particular, the evaluation of the specific family of Boolean functions for which we identify quantum advantage corresponds to maximal violations of generalized Svetlichny inequalities \cite{Collins02}. Therefore, a maximal violation of these inequalities starting from a specific size implies a corresponding circuit separation \footnote{Recall that the separation in Theorem \ref{main} is proven to hold asymptotically. Therefore, there is a minimum problem size above which all the quantum circuits are more efficient than the best classical circuits.}. More precisely, any quantum state that can be prepared and measured with linearly bounded quantum circuits that violate these inequalities maximally implies a corresponding circuit separation.

\begin{figure}[H]
    \centering
    \includegraphics[scale=0.40]{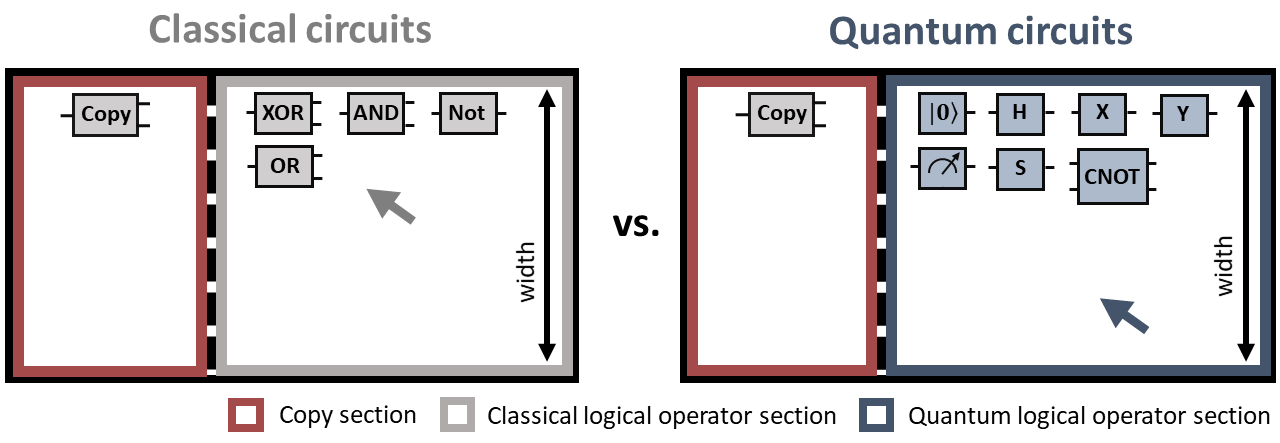}
    \caption{Illustration of the circuit classes for which we prove a separation in Theorem \ref{main}. For the classical class, the input can be copied only in the copy section, which precedes the logical operator section. For the quantum circuits, the same occurs. This prevents the circuits from employing the copy operation at any point in time, which would replicate the effect of memory in the circuits. }
    \label{fig:separation}
\end{figure}

Finally, we extend the previous analysis to higher degree symmetric Boolean functions, and conjecture that these have no advantage with respect to classical circuits.  

\begin{conjecture}\label{lowerdeghigh}
(Informal) Symmetric Boolean functions $f$ with $\mathsf{deg}(f)\geq 3$ evaluated within the non-adaptive MBQC model do not entail circuits with better scalings than classical circuits with unary and binary Boolean operators with single fan out. 
\end{conjecture}

\noindent This conjecture is based on several independent results. The first is the exponential lower bound proven in \cite{Hoban2011a} for the general $\mathsf{AND}$ function which is symmetric and has an instance for any possible degree. Then, on the same level, we conjecture and provide numerical support that symmetric Boolean functions require quantum states whose size scales at a greater rate than the corresponding number of classical bits required in the optimal classical circuits for the same computations. Furthermore, we show that with the Clifford$+$T gate set the necessary measurement operators cannot be synthesized exactly, and always need to be approximated. In contrast, there is no difficulty in computing the same function deterministically with classical circuits. This illustrates some of the restrictions resulting from the imposed flat temporal order.

\subsection{Related work}

\noindent \textbf{Non-local games.\ }
Computations within the non-adaptive MBQC model were proven to have a one-to-one correspondence to multi-party Bell inequalities with dichotomic observables and outcomes  \cite{Hoban2011a}. Therefore, the classical and the quantum bounds determined in previous works for this type of Bell inequalities immediately translate to computational efficiencies of the respective functions in the non-adaptive MBQC model \cite{BrunnerNicolas2014}. For the quantum bounds, we obtain the maximal efficiency that can be obtained from quantum resources to compute the respective Boolean functions, and equivalently, the same happens for classical bounds. From this relation, we can compare our work with a large spectrum of results and techniques on the optimal strategies for non-local games \cite{BrunnerNicolas2014,kravvcenko2013quantum,Slofstra11}.

For instance, in \cite{ambainis2012advantage,ambainis2013provable}, the authors consider arbitrary symmetric non-local games, which correspond to the symmetric Boolean functions we study here. 
Their results complement our own in the probabilistic setting. We determined the size of the quantum resource to obtain deterministic computations. At the same time, they considered the efficiency obtained if a quantum resource with a size equal to the input string is evaluated, with the additional comparison to a classical resource of the same size. Their main results show a mean advantage of $\sqrt{\log(n)}$ between the use of classical or quantum resources. Moreover, the magnitude of this separation appears to be similar to the separation demonstrated in Theorem \ref{main} due to the $\log(n)$ term. Still, it is not comparable to our results, given that they consider the classical limit based on measurement-based computations with classical resources, while we allow for arbitrary classical circuits without memory. 

Furthermore, in the study of communication complexity, the non-local boxes usually employed can be interpreted as non-adaptive MBQC computations \cite{Reznik08}. This observation relates this model to results obtained in the study of non-local boxes. 
The study of this relationship is two-fold. First, in references \cite{Pawlowski09,Popescu94,Barret05}, the hypothesis of axiomatizing quantum mechanics with information-theoretical principles intends to limit the quantum correlations allowed with these principles. Therefore, the resulting limits to the non-local boxes' efficiencies also relate to the maximal efficiency of the corresponding non-adaptive MBQC computations \cite{razborov2003quantum,zhang2009communication,botteronnonlocal}. On the other hand, non-local games defined on a non-binary set of observables and outcomes can be interpreted as a generalization of the non-adaptive MBQC model to qudits , together with the respective classical linear side processors \cite{bae2018generalized,Frembs_2018}. 

\vspace{0.4cm}

\noindent \textbf{Quantum circuits.} A breakthrough result 
by Bravyi, Gosset, and K\"{o}nig \cite{Bravyi188}  shows that a class of relation problems can be solved in constant depth with a quantum circuit, while requiring logarithmic depth in a classical computer \cite{Bravyi188}. This result has been extended in various manners \cite{Grier20,Libor22}. In \cite{Gall19}, the analysis is extended from worst-case difficulty to average-case difficulty. In \cite{Courdron18}, the probability with which a classical circuit correctly computes the relation problem was exponentially reduced with the size of the input string. Dealing with the same separation, the works of \cite{Bravyi2020,Atsuya21} have made it more robust against noise.

The first connection of our work to these results is that the proof relies on non-local games that are fundamental to the separation. These are naturally described with the non-adaptive MBQC model, which could be used to extend the relation problems with different non-local games, as done in \cite{Courdron18}. Nevertheless, a direct connection to these circuit separations can be made with the work of \cite{Watts19}. In this paper, the authors present the parity halving problem $\textnormal{PHP}:\{0,1\}^n \rightarrow \{0,1\}^n$, defined for all $x \in \{0,1\}^n$ with even Hamming weight and $y\in \{0,1\}^{n+1}$ as
\begin{equation}
 |x|/2 \equiv |y|\ \mod\ 2 \ ,
\end{equation}
\noindent with $|x|$ and $|y|$ describing the Hamming weight of the respective bitstrings.
The authors prove that this relation problem can be computed in $\mathsf{QNC}^0$ while not being computable with any polynomial size circuit of the $\mathsf{AC}^0$ class. Furthermore, this is precisely a sub-problem of the relation problem built in the quantum circuit that we use to prove a quantum advantage for a decision problem. In particular, the Hamming weight of the input string and the output string from the measurement layer performed on the GHZ state describes the full parity halving problem $\textnormal{FPHP}:\{0,1\}^n\rightarrow\{0,1\}^{n+1}$, defined for all $x\in \{0,1\}^n$ and $|y|\in \{0,1\}^n$ as  
\begin{equation}
\begin{cases}
    |x|/2 \equiv |y|\ mod\ 2,\ \ |x|\mod\ 2=0 \\
    |x-1|/2 \equiv |y|\ mod\ 2,\ \ |x|\mod\ 2=1
\end{cases}\ .
\end{equation}
\noindent We translate this into a decision problem by taking the parity of the output string $|y|$ of this relation (Figure \ref{fig:relationdecision}). Interestingly, using a sub-circuit with a quantum advantage respective to the local relation problem does not imply an advantage for the equivalent decision problem. Only the inverse implication would be unconditionally true. 

Another related work concerning shallow quantum circuits is that of \cite{parham2022power}. Here constant-depth quantum circuits that solve a sampling problem related to the class of symmetric Boolean functions we have explored are proposed. In addition, an advantage to the best constant depth circuits has been identified for these problems. 
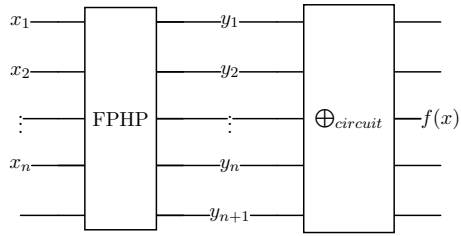
\begin{figure}[H]
\begin{center}
\begin{tikzpicture}\node[scale=0.7]{
\begin{quantikz}
x_1    & &\gate[5, nwires=5]{\textnormal{FPHP}}        & \qw  & y_1     & & \gate[wires=5][1.7cm]{\bigoplus_{circuit}}& \\
x_2    & &                      & \qw  &  y_2    & &  & \\
\vdots & &                      & \qw & \vdots  & &  &\qw f(x) \\
x_n    & & \qw                  & \qw  & y_{n}   & &  &  \\
      & &                      & \qw & y_{n+1} & &  &
\end{quantikz}};
\end{tikzpicture}
\end{center}
    \caption{Circuit for evaluating the symmetric Boolean functions with degree two built from the FPHP relation by the composition of a classical circuit that computes the parity.}
    \label{fig:relationdecision}
\end{figure}
From another perspective, the quantum advantage for the decision problem we identified is described in terms of the number of gates and the width of the circuits. Various related results have been obtained for circuit width. For instance, in \cite{Maslov21}, the authors prove a quantum advantage using the same symmetric Boolean function of degree two that we adress here. The advantage regarding the equivalent classical model is proven by fixing a width size for the computational space and showing that these computations can be computed with greater efficiency in an actual quantum device, IBM's quantum computer. Therefore, our result can be interpreted as the computation of the same function without any width restriction but with a temporal restriction on the time ordering for the quantum circuit (all the measurements have to commute). 

Finally, quantum width advantages for a class of symmetric Boolean functions in the setting of branching programs were demonstrated in \cite{ABLAYEV2005}. In particular,  these functions separate optimally quantum and classical branching programs and are generalizations of the symmetric Boolean function we used in the proof of Theorem \ref{main}.

\subsection{Document organization}
The remainder of the document is organized as follows. Section \ref{background} provides the background for the non-adaptive measurement-based quantum computation model. Section \ref{constructions} introduces our new constructions to select the measurement instructions for Boolean function evaluation. Subsequently, Section \ref{tbound} analyses theoretical resource bounds for the model, providing the technical proof of Theorem \ref{clif1} and support for Conjecture \ref{lowerdeghigh}. Section \ref{circuits} translates the model into realizable computations in the form of circuits, providing the technical proof of Theorem \ref{main} and additional support for Conjecture \ref{lowerdeghigh}. Finally, section \ref{conclusion} concludes with a brief discussion of the results.

\section{Background}\label{background}

\subsection{The non-adaptive measurement-based quantum computation model}

The definition of the non-adaptive MBQC (NMQC$_\oplus$) model builds upon that of the MBQC model \cite{Hoban2011a}. In general terms, the MBQC model is a universal quantum computing model which performs computations by a sequence of measurements on a highly entangled quantum state \cite{Raussendorf01, UniversalRes06}. The measurement operators applied during the computation must be selected adaptively, depending on previous measurement outcomes. The \textnormal{NMQC}$_\oplus$ model differs from the general MBQC in that it does not allow for measurement adaptivity. Therefore, it cannot select subsequent measurement operators assuming knowledge about the outcome of previous measurements. Without the capacity to adapt, the order of single-qubit measurements performed on the resource state is not defined, and the temporal structure becomes flat \cite{Bremner16, Bremner09}.

Furthermore, in the usual MBQC model \cite{Raussendorf01, UniversalRes06}, each measurement is chosen using a limited classical side computer capable only of linear computations. Thus, in \textnormal{NMQC}$_\oplus$ this linear side computer can immediately determine all the measurement bases required for the computation based on the input string. Then, the selected measurement operators are applied to the resource, producing a set of outcome values. These outcomes are fed again into a linear side computer, which delivers the final output of the entire computational process. A detailed description of the  model can be given as a procedure consisting of the following three main stages:
\paragraph{Stage 1.}
A classical input string $x \in \{0,1\}^n$ is pre-processed by a linear side computer generating a binary output $s_i \in \{0,1\}$ for each qubit $i$ of the resource state, with $k$ qubits. Without loss of generality, this computation for each qubit is given by the general linear expression,
\begin{equation}\label{linear}
 s_i=L_i(x)= a_{i_0} \oplus \bigoplus_{j \in S_i} \big(a_{i_j}\oplus x_j \big)
\end{equation}
\noindent  for all $i \in \{0, 1,..., k-1\}$, with $a_{i_j} \in  \{0,1\}$, and $\ S_i \subseteq [n]$.
\paragraph{Stage 2.}
To each bit-value of $s_i\in \{0,1\}$ we associate a single-qubit measurement operator on the equator of the Bloch sphere as, 
\begin{equation}\label{measurements}
   M_{i}(s_i)=\cos(\theta_i+\phi_i*s_i) \sigma_x + \sin(\theta_i+\phi_i*s_i) \sigma_y 
\end{equation}
\noindent  for all $i \in \{0,1,...,k-1\}$, with $\sigma_x$ and $\sigma_y$ standing for the Pauli $X$ and $Y$ operator, respectively, and angles $\theta_i$ and $\phi_i$ defining the specific measurement operators. In the measurement of, 
\begin{equation}
 M_0(s_0)\otimes M_{1}(s_{1}) \otimes ...  \otimes M_{i}(s_i) \otimes ... \otimes M_{k-1}(s_{k-1})  \ ,
\end{equation}
\noindent on resource state $\ket{\Psi_k}$, each one of these measurements yields a binary output $m_i\in \{-1,1\}$, with an associated probability for all $i \in \{0,1,...,k-1\}$. The result of this process is a string $x_m$ of symbols $\pm 1$ with length $k$ composed of eigenvalues of the measurement operators $\otimes_{i=0}^{k-1} M_{i}(s_i)$,
\begin{equation}
 x_m= m_0 , m_1 , \hdots ,m_i,\hdots , m_{k-1}\ \ \text{with}\ m_i \in \{-1,1\}\ .
\end{equation}
\paragraph{Stage 3.}
 In the final stage, the auxiliary linear side processor is employed again to compute the binary output based on the previously obtained string $x_m$, 
\begin{equation}
  L_f(x_m)= a_0 \oplus \bigoplus_{i=1}^{k} \big(a_i\oplus g(m_i)\big ) 
\end{equation}
\noindent where $g(x)=\dfrac{1-x}{2}$ acts as a renormalization function to map the $\pm 1$ outcomes to the more usual bit values $0$ and $1$. The complete process is illustrated in Figure \ref{fig:nmqc} (a), summarizing the different stages of the \textnormal{NMQC}$_\oplus$ model.

\begin{figure}[H]
\centering
\begin{subfigure}[b]{0.59\textwidth}
\centering
\includegraphics[scale=0.32]{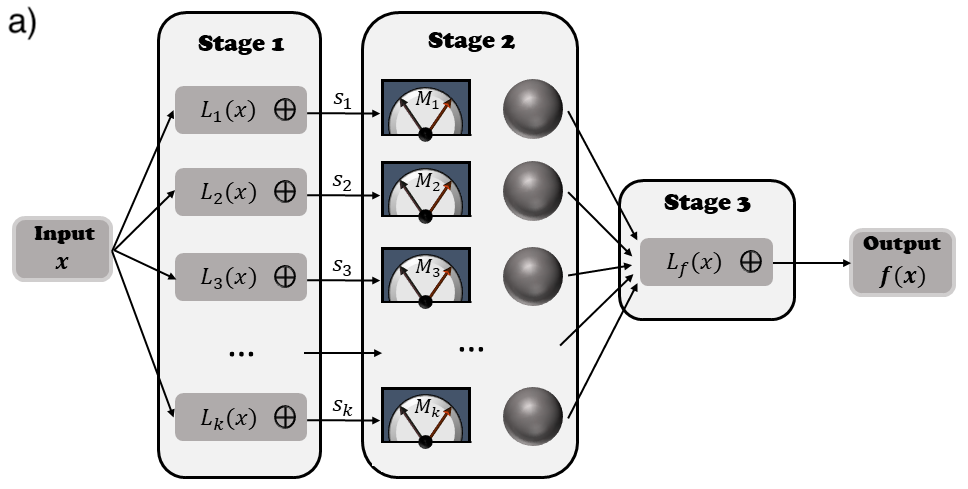}
\label{fig:model}
\end{subfigure}
\begin{subfigure}[b]{0.4\textwidth}
\centering
\includegraphics[scale=0.28]{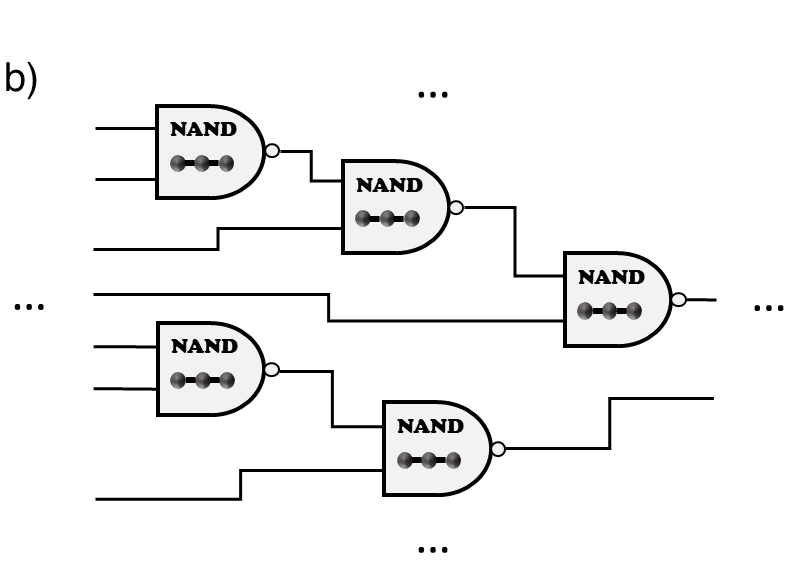}
\label{fig:nand}
\end{subfigure}
\caption{(a) The complete \textnormal{NMQC}$_\oplus$ procedure. The input string $x \in \{0,1\}^n$ is shared with each linear function of the pre-processing stage. Then the binary values $s_i$ are forwarded to the measurement process, determining a specific measurement operator on all qubits. In the end, the measurement results are fed into a linear function to compute the result of the \textnormal{NMQC}$_\oplus$ computation. (b) Representation of a digital circuit composed of NAND gates that are implemented using individual \textnormal{NMQC}$_\oplus$ computations.}
        \label{fig:nmqc}
\end{figure}

\subsection{Measurement assignments} \label{universality}

The lack of adaptivity prevents \textnormal{NMQC}$_\oplus$ from being a universal model for quantum computation, as MBQC is. Nevertheless, Anders and Browne demonstrated that it could evaluate the NAND gate deterministically with a 3-qubit GHZ state and measurements of its stabilizer group. This particular gate is special because it describes a complete gate set for a universal classical computer. Consequently, any Boolean function can be computed by composing layers of  \textnormal{NMQC}$_\oplus$ evaluations (Figure  \ref{fig:nmqc} (b)).  However, by composing different \textnormal{NMQC}$_\oplus$ computations, the whole computational process becomes adaptive. This implication can be easily understood as the result of a specific NAND gate will define the measurement operators of the NAND gates to its outputs are fed into. Therefore, we will focus mainly on the computation capabilities of the \textnormal{NMQC}$_\oplus$ model without 
any compositions. 

Subsequent to the work of Anders and Browne, which introduced the main ideas, the first detailed description of the \textnormal{NMQC}$_\oplus$ model was presented by Hoban \textit{et al.} in \cite{Hoban2011a}. This work examines the set of instructions necessary for each stage of the model, the measurement assignments, in order to deterministically evaluate Boolean functions. Formally,

\begin{definition} \textnormal{ \textbf{(Measurement assignment).}}
A measurement assignment for a \textnormal{NMQC}$_\oplus$ computation is composed by a set of linear functions $L_i$, a set of dichotomic measurement operators $M_i(s_i)$, for all $i\in \{0,1,...,k-1\}$, a final linear function $L_{f}$ and a resource state $\Psi_k$, which takes the form of a $k$-qubit quantum state. 
\end{definition}

In the same work, the authors proved that there is a measurement assignment for any Boolean function, $f:\{0,1\}^n\rightarrow\{0,1\}$, which deterministically evaluates the function in the \textnormal{NMQC}$_\oplus$ model. These measurement assignments require general GHZ states as a resource state \cite{GHZ1989,Cruz19}. Moreover, such highly entangled quantum states were proven to be the most efficient resource state for these computations, concerning the number of required qubits \cite{Wolf2001,Hoban2011a}, as stated in the following theorem.

\begin{theorem}\label{hobans} \cite{Hoban2011a} In \textnormal{NMQC}$_\oplus$ with quantum resources, any Boolean function $f(x)$, on an arbitrary n-bit string $x$, can be computed with unit probability. This can be achieved with the resource state $\frac{1}{\sqrt{2}}\big(\ket{0}^{\otimes(2^n-1)} +\ket{1}^{\otimes(2^n-1)} \big) $ i.e. a
$(2^n-1)$ qubit generalised GHZ state.
\end{theorem}

 \noindent This result proves that classical universality can be 
 achieved within the \textnormal{NMQC}$_\oplus$ model, without the need to compose various layers of \textnormal{NMQC}$_\oplus$ computations. Furthermore, it motivates the following research question

\begin{center}
 \textit{"What are the measurement assignments for an arbitrary Boolean function?"}
 \end{center}
 
\noindent Hoban \textit{et al.} state that their proof for Theorem \ref{hobans} also provides a constructive method to discover valid measurement assignments. We argue that this process is computationally too expensive, as it searches over all possible combinations of linear functions for the pre-processing stage and over a continuum for all the angles of the measurement stage. This assertion will be supported by the complexity analysis of the constructions proposed in \cite{Mori2019, Frembs22}, and which are much more computationally efficient than the original proposal in \cite{Hoban2011a}.

In the remainder of this section, we will review two different constructions that can be used to derive specific measurement assignments for a Boolean function. The two methods were introduced in \cite{Mori2019}, while the second construction was reintroduced independently in \cite{Frembs22}.

\paragraph{Construction 1: The Fourier representation ($\mathcal{FR}$) (\cite{Mori2019}) }

The first construction merges two ideas to derive a method that yields a measurement assignment for any Boolean function, $f:\{0,1\}^n\rightarrow\{0,1\}$. It starts by verifying that the expected value of a measurement operator composed of single-qubit measurements on a general GHZ state, 
 $\Psi_{GHZ}$, can be described with a cosine function,
\begin{align}
    \bra{\Psi_{GHZ}^k}\otimes_{i=1}^k M_i(s_i) \ket{\Psi_{GHZ}^k}= \cos \big(   \sum_{i=1}^k\theta_i+   \phi_i s_i \big)
    = \cos \big(  \underbrace{ \sum_{i=1}^k\theta_i+   \phi_i L_i(x)}_{\mathsf{poly_f(x)}} \big)  \ ,
\end{align}
\noindent whose assignment in a multi-linear polynomial, $\mathsf{poly_f}:\{0,1\}^n\rightarrow \pi(\mathbb{R})$ \cite{Wolf2001} \footnote{We use the notation $\mathsf{poly_f}$ throughout the text to represent either polynomial of the form $\mathsf{poly_f}:\{-1,1\}^n\rightarrow \mathbb{C}$ or $\mathsf{poly_f}:\{0,1\}^n\rightarrow \mathbb{C}$. Both of these mathematical objects serve as representations for a Boolean function of the form $f:\{0,1\}^n\rightarrow \{0,1\}$, and the type of $\mathsf{poly_f}$ used will depend on the convenience given the context where it is used.}. 

\begin{definition}\label{multipoly}
   \textnormal{ \textbf{(Multi-linear polynomial).}} A multi-linear polynomial is a function $\mathsf{poly}_f:\{1,-1\}^n\rightarrow \mathbb{C}$ that can be written as,
    \begin{equation}
        \mathsf{poly}_f(x)= \sum_{S\subseteq [n]} c_S^\ast \prod_{i\in S} x_i
    \end{equation}
    \noindent where each $c_S^\ast$ is a complex number.
\end{definition}

This polynomial has real coefficients $\theta_i$ and $\phi_i$, which relate to the parameters of the measurement operators from Equation \eqref{measurements}, and the binary variables $s_i$ related to the binary values from Equation \eqref{linear} that select the measurement bases. Thus, finding measurement assignments that compute a specific Boolean function is equivalent to finding the corresponding multi-linear polynomials. In particular, a valid polynomial has to map each input $x \in \{0,1\}^n$ to a multiple of $\pi$ that is congruent modulo two to the output of the Boolean function. Then, the total measurement operator maps to the eigenvalue corresponding to the output of the Boolean function, for all $f(x):\{0,1\}^n\rightarrow\{0,1\}$ and $x\in \{0,1\}^n$,
\begin{align}\label{cosmod}
    \bra{\Psi_{GHZ}^k}\otimes_{i=1}^k M_i(L_i(x)) \ket{\Psi_{GHZ}^k}=\cos \big(  \underbrace{\sum_{i=1}^k  \theta_i+ \phi_i L_i(x)}_{\mathsf{poly_f(x)}} \big)= (-1)^{f(x)}\ .
\end{align}

The existence of such a multi-linear polynomial is guaranteed for every semi-Boolean function ($f_{\ast}:\{0,1\}^n\rightarrow \mathbb{R}$) \cite{Odonnel2021}. A specific instance can be computed through the Walsh-Hadamard transform $\mathcal{F}$, which is the Fourier transform on the Boolean group $(\mathbb{Z}/2\mathbb{Z})^n$ \footnote{A description for the Walsh-Hadamard transform is given in Appendix \ref{fourier}.}. This transformation yields the Fourier coefficients ($\widehat{f}(S)$) given an ordered vector with all the possible input strings evaluated by the Boolean function,
\begin{equation}
   \frac{1}{2^n}  \mathcal{F}^T\left( \begin{matrix}
    f(x_0) \\
    f(x_1) \\
    \vdots \\
    f(x_{2^n-1}) \\
    f(x_{2^n}) \\
    \end{matrix}\right) = \begin{bmatrix}
    \widehat{f}(S_0) \\
    \widehat{f}(S_1) \\
    \vdots \\
    \widehat{f}(S_{2^n-1}) \\
    \widehat{f}(S_{2^n}) \\
    \end{bmatrix} \ .
\end{equation}

\noindent These are the coefficients of the corresponding multi-linear polynomial, $c_S^\ast=\widehat{f}(S)$. Therefore, for all $f:\{0,1\}^n\rightarrow\{0,1\}$ and $x\in \{0,1\}^n$,
\begin{equation}\label{polyeq}
  \mathsf{poly_f} (x) =    \pi*\Big(\sum_{S\subseteq [n]}{\widehat{f}(S)*g^{-1}\Big(  \bigoplus_{i\in S} x_i  \Big )} \Big)=\pi*f(x) \ ,
\end{equation}
\noindent with $g^{-1}(x)=1-2x$.

 The multi-linear polynomial $\mathsf{poly_f}$ obtained through this process can be used to construct a measurement assignment that computes the Boolean function deterministically in the \textnormal{NMQC}$_\oplus$ model. The resource state will be a general GHZ state of $k$ qubits,
\begin{equation}
\ket{\Psi_{GHZ}^k}= \dfrac{1}{\sqrt{2}} \big(\ket{0}^{\otimes k} + \ket{1}^{\otimes k} \big)\ ,
\end{equation}
\noindent with $k$ standing for the sparsity of the multi-linear polynomial $\mathsf{poly_f}$, defined as follows. 

\begin{definition}\textnormal{ \textbf{(Sparsity).}}
The sparsity of a multi-linear polynomial  $\mathsf{poly_f}:\{0,1\}^n\rightarrow \mathbb{R}$ is designated as the number of non-zero Fourier coefficients,
\begin{equation}
\mathsf{sparsity}(\mathsf{poly_f})=\mathsf{sparsity}\big( \mathcal{F}(f)\big)= \sum_{S\subseteq [n]\setminus [0]} \delta \big(\widehat{f}(S)\big), \ \ \delta(x)=
\left\{
\begin{array}{ll}
     1 & x\neq 0 \\
     0 & otherwise \\
\end{array} 
\right.\ ,
\end{equation}
which results from applying the Fourier transform $\mathcal{F}$ to the semi-Boolean function, $f_{\ast}:\{0,1\}^n\rightarrow \mathbb{R}$, such that,  for all $x \in \{0,1\}$, $f_{\ast}(x)=\mathsf{poly}_{f_{\ast}}(x)$.
\end{definition}

In the pre-processing stage (Stage 1), the linear Boolean functions which select the measurement operators, are all the parity bases ($\bigoplus_{l\subseteq S} x_l$), corresponding to non-zero Fourier coefficient  ($\widehat{f}(S_i)\neq 0$) of $\mathsf{poly_f}$,   
\begin{equation}\label{linears}
    s_i=L_i(x)=\bigoplus_{l\subseteq S_i} x_l\ ,
\end{equation}

\noindent  with $i\in \{0,1,...,k-1\}$. The angles of the measurement operators, from Equation \eqref{measurements} in Stage 2, are built from the non-zero Fourier coefficients that constructed the multi-linear polynomial resorting to Equation \eqref{polyeq}, such that, for all $i \in \{0,1,...,k-1\}$,
\begin{equation}\label{angles}
  \theta_i=  \dfrac{\pi}{k}*\sum_{S\subseteq [n]}\widehat{f}(S) \ .
\end{equation}
\noindent Notice that each $\theta_i$  always has the same value for all measurement operators. Only the values of $\phi_i$ change, and these angles are defined by the corresponding non-zero coefficients of the polynomial,
\begin{equation}\label{fmeasure}
 \phi_i=-2\pi* \widehat{f}(S_i)\ ,
\end{equation}

\noindent for all $i \in \{0,1,...,k-1\}$. Therefore, each non-zero Fourier coefficient $\widehat{f}(S_i)$ links a specific parity basis $L_i(x)=\bigoplus_{l\in S_i} x_l$ to the corresponding measurement operator $M_i(s_i)$, which will be applied to one of the qubits of the resource state.

Finally, the linear Boolean function of the post-processing stage, Stage 3, will be the linear sum of all the measurement outcomes ($m_i$), 
\begin{equation}
     L_f(x_m)=\bigoplus_{i=1}^{k}  g(m_i)\ , 
\end{equation}
\noindent for all Boolean functions, with a measurement assignment determined directly by the equivalent multi-linear polynomial.

\begin{example}\label{andexe}
For the $\mathsf{AND}$ function with input strings of size two, $\mathsf{AND}(x_1,x_2)=x_1*x_2$, the discrete Fourier transform applied to the corresponding truth table computes the Fourier coefficients,
\begin{equation}
   \frac{1}{4}\begin{bmatrix}
    1 & -1 &\ \ 1&\ \ 1\\
    1 &\ \ 1  &\ \ 1& -1\\
    1 &  -1& -1& -1\\
    1 &\ \ 1 & -1 & \ \ 1 \\
\end{bmatrix}* \begin{bmatrix}
    0 \\
    0 \\
    0 \\
    1 \\
    \end{bmatrix} = \begin{bmatrix}
 \  \ 1/4 \\
     -1/4\\
     -1/4 \\
 \  \ 1/4 \\
    \end{bmatrix} \ .
\end{equation}
\noindent
These are used to build the corresponding multi-linear polynomial resorting to Equation \eqref{polyeq},
\begin{equation}
    \mathsf{poly_{\mathsf{AND}}}(x)= \dfrac{1}{2}x_1+\dfrac{1}{2}x_2-\dfrac{1}{2} x_1\oplus x_2 \ .
\end{equation}
\noindent Then, the coefficients of the polynomial are used to create the measurement assignment. In this case, the polynomial translates to a 3-qubit GHZ state, $\ket{\Psi_{GHZ}^3}= 1/\sqrt{2} \big(\ket{0}^{\otimes 3} + \ket{1}^{\otimes 3} \big)$, the linear functions 
\begin{equation}
     L_1(x)=x_1,\  L_2(x)=x_2,\ and\ L_3(x)=x_1\oplus x_2 \ ,
\end{equation}
\noindent and the corresponding measurement operators
\begin{equation}
    M_1(s_1)=(\neg s_1)\sigma_x+ s_1\sigma_y,\ M_2(s_2)=(\neg s_2)\sigma_x+ s_2\sigma_y\ and\ M_3(s_3)=(\neg s_3)\sigma_x-s_3\sigma_y\ .
\end{equation} Moreover, the $\mathsf{AND}$ function evaluated on all possible inputs, for $x\in \{0,1\}^2$, in the \textnormal{NMQC}$_\oplus$ model is equal to,
\begin{align*}
    &\sigma_x \otimes \sigma_x \otimes\ \ \sigma_x \ket{\Psi_{GHZ}^3} = +\ket{\Psi_{GHZ}^3},\ \sigma_x \otimes \sigma_y \otimes -\sigma_y \ket{\Psi_{GHZ}^3} = +\ket{\Psi_{GHZ}^3}\\ 
    &\sigma_y \otimes \sigma_x \otimes -\sigma_y \ket{\Psi_{GHZ}^3} = +\ket{\Psi_{GHZ}^3},\ \sigma_y \otimes \sigma_y \otimes\ \  \sigma_x \ket{\Psi_{GHZ}^3} = -\ket{\Psi_{GHZ}^3}\ .
\end{align*}

\noindent Yielding the correct correspondence between the resulting eigenvalues and the input strings, such that the final state is equal to $(-1)^{\mathsf{AND}(x_1,x_2)}\ket{\Psi_{GHZ}^3}$. 
\end{example}

This first construction, taken from \cite{Mori2019} was presented with a detailed translation from a valid multi-linear polynomial, for a specific Boolean function, to a correct and deterministic measurement assignment. The construction is characterized by a process that finds a multi-linear polynomial, which in this case is the Walsh-Hadamard transform. The subsequent constructions differ precisely in the process of finding an equivalent multi-linear polynomial for the function. Thus, it is this process alone that distinguishes between them, as the conversion from the multi-linear polynomial to the measurement assignment will be always the same in every case. Therefore, this process will not be repeated for the next constructions.

\paragraph{Construction 2: The extended Fourier representation ($\mathcal{EF}$) (\cite{Mori2019,Frembs22})}

The second construction presented in \cite{Mori2019}, equivalent to the construction introduced in \cite{Frembs22}, applies the discrete Fourier transform using a different representation of the Boolean function, the Algebraic Normal Form (ANF).

\begin{definition}\textnormal{ \textbf{(Algebraic Normal Form).}}
Any Boolean function $f:\{0,1\}^n\rightarrow\{0,1\}$ can be expressed uniquely as,
\begin{equation}
      f(x)= \bigoplus_{S \subseteq [n]} c_S \prod_{i \in S} x_i
\end{equation}
\noindent with $c_S \in \{0,1\}$.
\end{definition}

Given the ANF of a Boolean function, this construction applies the  Walsh-Hadamard transform to each monomial ($\prod_{i \in S} x_i$). The resulting multi-linear polynomials are summed to assemble a single multi-linear polynomial representing the desired Boolean function. This process constructs a correct multi-linear polynomial and the corresponding correct measurement assignment due to the additivity property of the Walsh-Hadamard transform,
\begin{align}
      \mathsf{poly_f}(x) &= \pi* \frac{1}{2^n}  \mathcal{F}^T\bigg(\bigoplus_{S \subseteq [n]} c_S \prod_{i \in S} x_i \bigg)
 \\& \equiv \pi *\sum_{S \subseteq [n]} c_S * \frac{1}{2^n}  \mathcal{F}^T\bigg(\prod_{i \in S} x_i \bigg) \ .
\end{align}

In summary, we described two constructions that produce correct measurement assignments for any Boolean function. The measurement assignments obtained from them may differ. In particular, the sparsity of the multi-linear polynomials is the main difference and one of the most significant properties related to the efficiency of the \textnormal{NMQC}$_\oplus$ model. Generally, the $\mathcal{EF}$ construction produces sparser polynomials than the $\mathcal{FR}$ construction. This separation can be exponentially large for some families of Boolean functions \footnote{We designate as families of Boolean functions a set of Boolean functions with a precise description and a single instance for each input string size. The general $\mathsf{AND}$ is an example of a family of Boolean functions.} \cite{Mori2019,chistopolskaya2018parity} . Thus, asymptotically the $\mathcal{EF}$ construction is identified as the construction that produces the most efficient measurement assignments.  

\section{Evaluation of Boolean functions in the \texorpdfstring{$\textnormal{NMQC}_\oplus$}{NQMC} model}\label{constructions}

This section presents new constructions for the deterministic evaluation of Boolean function in the \textnormal{NMQC}$\oplus$ model. First, we introduce a new construction for symmetric Boolean functions (SBF). Then, we demonstrate how this new method can be integrated into another construction, yielding more efficient \textnormal{NMQC}$\oplus$ computations for arbitrary Boolean functions.

\subsection{Symmetric Boolean functions}\label{symmetric}

In Boolean functions analysis \cite{Odonnel2021}, the set of symmetric Boolean functions i.e., those which are invariant under permutations of the input strings, was studied for various properties and applications in computer science \cite{Canteaut05,Stockmeyer76,Arnold63,Braeken05}. As a result of this property, the ANF representation has the following form, for all $x\in \{0,1\}^n$,
\begin{equation}\label{symmetricrep}
   f(x) = c_0 \oplus c_1*C^1(x)
    \oplus c_2*C^2(x)   \oplus c_3*C^3(x) \oplus ... \oplus c_n*C^n(x)  = \bigoplus_{k=0}^n c_k*C^k(x),
\end{equation}
\noindent where each $c_k \in \{0,1\}$ and the $C^k(x)$ terms represent the complete symmetric functions (CSF), defined as
\begin{definition}\textnormal{ \textbf{(Complete symmetric function).}}
The complete symmetric function (CSF) of dimension $k$, $C^k(x):\{0,1\}^n\rightarrow \{0,1\}$  is defined as the sum, modulo 2, of all terms with degree equal to $k$, i.e.
\begin{equation}
   C^k(x)= \bigoplus_{i_1=1}^{n-k+1} x_{i_1} \Big(   \bigoplus_{i_2=i_1+1}^{n-k+2} x_{i_2}  \big(\ \  ...\ \  (  \bigoplus_{i_N=i_{k-1}+1}^{n} x_{i_N} )\big) \Big) \ .
\end{equation}
\end{definition}

\noindent These functions are the building blocks used to generate the entire set of symmetric Boolean functions \cite{BUHRMAN200221}, according to Equation \eqref{symmetricrep}. Therefore, a construction that renders multi-linear polynomials for the CSF functions will be universal for the set of symmetric Boolean functions.

\paragraph{Construction 3: CSF Representation ($\mathcal{CSF}$) } We propose a new construction that builds a multi-linear polynomial for an arbitrary symmetric Boolean function. This construction takes advantage of the decomposition established in the following lemma. 

\begin{lemma}\label{decomposition}
All complete symmetric functions $C^k$, for an arbitrary dimension $k$, can be decomposed as the product of complete symmetric functions with a dimension equal to power of two, i.e. for all $x\in \{0,1\}^n$
\begin{equation}\label{decomC}
C^k(x)=  \bigwedge_{r \in R_k} C^{2^r} (x),
\end{equation}
\noindent with $R_k \subseteq [\ceil*{\log(k)}]$, such that $\sum_{r\in R_k} 2^{r}=k$.
\end{lemma}
\begin{proof}
See Appendix \ref{lemma2proof}.
\end{proof}
\noindent This decomposition demonstrates that a subset of CSFs is sufficient to build the entire set of symmetric Boolean functions using the logical operators $\mathsf{AND}$ ($\wedge$) and $\mathsf{XOR}$ ($\oplus$). 

\begin{proposition}\label{decomp}
For any symmetric Boolean function $f$, there is an equivalent multi-linear polynomial $\mathsf{poly}_{f}$ that can be built solely by composing multi-linear polynomials from the subset of complete symmetric functions which have a degree equal to a power of two, such that,  for all $x\in \{-1,1\}^n$
\vspace{-0.1cm}
\begin{equation}\label{eqprop1}
   \mathsf{poly}_{f}(x) =\sum_{i=0}^n c_i \bigg (\prod_{r \in R_i} 	 \text{\large $\mathcal{GC}$} \Big(C^{2^r} (x)\Big) \bigg ),
\end{equation}
 with $R_i \subseteq [\ceil*{\log(i)}]$, such that $\sum_{r\in R_i} 2^{r}=i$ and $\text{\large $\mathcal{GC}$}$ any general construction that generates multi-linear polynomial for CSFs. 
\end{proposition}
\begin{proof}

The representation of SBF used in Equation \eqref{symmetricrep}  can be adapted to depend uniquely on the CSF set with a degree equal to a power of two. Recall that this translation is possible by Lemma \ref{decomposition}. Then, any general construction $\text{\large $\mathcal{GC}$}$ produces valid multi-linear polynomials for these two representations,  
\begin{align}
  \mathsf{poly}_{f}(x)&= \text{\large $\mathcal{GC}$} \Big( \bigoplus_{i=0}^n c_i*C^i \Big)= \text{\large $\mathcal{GC}$} \Big(\bigoplus_{i=0}^n c_i  \Big (\bigwedge_{r \in R_i} C^{2^r} (x) \Big ) \Big)  , 
\end{align}

\noindent  with $R_i \subseteq [\ceil*{\log(i)}]$, such that $\sum_{r\in R_i} 2^{r}=i$. Then, additivity can be used to push the application of the general construction from the entire expression to the CSFs, given that the linear sum $\oplus$ and the Real sum $+$ coincide for Boolean functions and their multi-linear polynomials,
\begin{align}
    \text{\large $\mathcal{GC}$} \Big(\bigoplus_{i=0}^n c_i  \Big (\bigwedge_{r \in R_i} C^{2^r} (x) \Big ) \Big)\equiv \sum_{i=0}^n c_i * \text{\large $\mathcal{GC}$} \Big (\bigwedge_{r \in R_i} C^{2^r} (x) \Big )   ,
\end{align}

\noindent Subsequently, we will resort to a lemma from \cite{Mori2019} that defines the translation from a conjunction operation on the function to the equivalent operation on multi-linear polynomials. In particular, this lemma shows that, for a set of Boolean function $f_1,f_2,...,f_k$ with common input variables $x_1,x_2,...,x_n$, their composition by logical $\mathsf{AND}\ (\wedge)$ is equivalent to the product of the corresponding multi-linear polynomials,
\begin{equation}
    \bigwedge_{j=1}^{k} f_j(x) = \prod_{j=1}^k poly_{f_j}(x)\ \ .
\end{equation}

\noindent Consequently, a valid multi-linear polynomial can be built from the polynomials with the CSFs of dimension equal to a power of two alone,
\begin{align}
    \mathsf{poly}_{f}(x) \equiv \sum_{i=0}^n c_i * \text{\large $\mathcal{GC}$} \Big (\bigwedge_{r \in R_i} C^r (x) \Big )  \equiv \sum_{i=0}^n c_i \Big (\prod_{r \in R_i} \text{\large $\mathcal{GC}$}\Big(C^r (x)\Big) \Big ).
\end{align}
Finally, note that the construction used to build these polynomials has only to be general in respect to any CSF with a degree equal to a power of two. 
\end{proof}

\noindent The same decomposition holds for the equivalent multi-linear polynomials as proven by Proposition \ref{decomp} using the multiplication and addition over the Reals. Thus, it describes a valid method to generate a multi-linear polynomial for an arbitrary symmetric Boolean function, defining our new construction.

This construction introduces modularity into the process of determining valid multi-linear polynomials. Consequently, to build these mathematical objects for symmetric functions, one can use any valid construction or known multi-linear polynomials, which are equivalent to the CSF in the decomposition. Therefore, the sparsest representations could be combined to attain the most efficient \textnormal{NMQC}$_\oplus$ computations. For instance, if one uses the $\mathcal{EF}$ construction from section \ref{universality} to compute multi-linear polynomials for symmetric Boolean functions, the same polynomials would result from the $\mathcal{CSF}$ construction. However, suppose one uses the multi-linear polynomial presented in \cite{Hoban2011a} for the $C^2$ function whenever it appears in the decomposition. In that case, the final multi-linear polynomials will always be sparser than any polynomials determined by either the $\mathcal{FR}$ or the $\mathcal{EF}$ constructions. This construction motivates the search for sparser multi-linear polynomials for this specific set of CSF, with a degree equal to a power of two.

\paragraph{CSF polynomials}
The $\mathcal{CSF}$ construction stresses the importance of determining efficient representations for the subset of CSF on the entire set of symmetric Boolean functions. Therefore, we identified a general form for the multi-linear polynomials corresponding to the CSFs, which have a degree equal to a power of two. We also proved that these are correct for any input size considering the CSFs with a degree lower or equal to 64, as proved in the lemma below. Therefore, to the authors' knowledge, these are the most compact multi-linear polynomial known for those CSFs.  

\begin{lemma}\label{polinomials}
 For $k=2,4,8,16,32$ and $64$, the general form of multi-linear polynomials represented by Equation \eqref{expr12} correctly computes the corresponding complete symmetric function $C^k:\{0,1\}^n\rightarrow \{0,1\}$ for all $x \in \{0,1\}^n$, using a general GHZ state as a resource, with the \textnormal{NMQC}$_\oplus$ model. 
 \begin{align}\label{expr12}
 \begin{split}
     \mathsf{poly}_{C^{k}}(x)=\frac{\pi}{2^{k-1}}\Bigg (\sum_{j=1}^{k/2} \binom{n-k/2-j}{k/2-j} (-1)^j \bigg( \sum_{S_i\subseteq [n],|S_i|=j} \bigoplus_{i \in S_i} x_i    - \sum_{S_i \subseteq [n] , |S_i|=n-j+1}  \bigoplus_{i \in S_i} x_i  \ \bigg)  \Bigg)
\end{split}
 \end{align}
 \noindent 
 \end{lemma}
\begin{proof}
See in Appendix \ref{polynomials}.
\end{proof}

\begin{example} 
The complete symmetric function of degree $4$ is defined for all $x\in\{0,1\}^n$ as
\begin{equation}
    C^4(x)=\bigoplus_{i_1=1}^{n-3} x_{i_1} \bigg(   \bigoplus_{i_2=i_1+1}^{n-2} x_{i_2}  \bigg( \bigoplus_{i_3=i_2+1}^{n-1} x_{i_3} \bigg (  \bigoplus_{i_4=i_3+1}^{n} x_{i_4}\bigg ) \bigg ) \bigg)\ ,
\end{equation}

\noindent and the corresponding multi-linear polynomials have the following form 

\begin{equation}
    \mathsf{poly}_{C^4}(x)= \frac{\pi}{8} \Big( \big(n-3\big)*\Big( \sum_{i=1}^n x_i - \big (\bigoplus_{j=1}^n x_j \big)\Big) + \sum_{i=1}^n \Big(\big(\bigoplus_{j=1}^n x_j\big)\oplus x_i\Big)- \Big( \sum_{i=1}^n \sum_{j=i}^n x_i\oplus x_j\Big)  \Big) \end{equation}
\noindent for all $x\in\{0,1\}^n$. Moreover, we compute their sparsity by summing the number of Fourier coefficients from different terms. Specifically, the single-bit terms contribute $n$ coefficients, as given by $\sum_{i=1}^n x_i$. The parities that involve all bits except one, expressed as $\sum_{i=1}^n \Big(\big(\bigoplus_{j=1}^n x_j\big)\oplus x_i\Big)$, contribute another $n$ coefficients. Additionally, there are $\binom{n}{2}$ coefficients from the pairwise linear sums represented by $\sum_{i=1}^n \sum_{j=i+1}^n x_i\oplus x_j$. Finally, there is a single term from the parity of all the bits, given by $\bigoplus_{j=1}^n x_j$.
\begin{equation}
    \mathsf{sparsity}\big(\mathsf{poly}_{C^4}(x)\big) =\binom{n}{2} +  2* \binom{n}{1} +1 = \frac{n^2}{2}+\frac{3n}{2}+1 = \mathcal{O}(n^2)\ ,
\end{equation}
\noindent Therefore, this reduces the number of qubits required  for the \textnormal{NMQC}$\oplus$  computation, compared to the previously result in \cite{Frembs22} that scales asymptotically with $\Theta(n^3)$ . 
\end{example}

The combination of the $\mathcal{CSF}$ construction with the CSF polynomials translates to a new set of sparse multi-linear polynomials, which are equivalent to symmetric Boolean functions. Consequently, these polynomials reduce the number of qubits required for the deterministic evaluation of symmetric Boolean functions in the \textnormal{NMQC}$_\oplus$ model. As a concrete example, in Appendix \ref{examplecsf} we show that with this construction, the $C^5$ function taking input strings of size $|n|=6$ demands a 43-qubit GHZ state, whereas the previous best solution resorts to a 63-qubit GHZ state \cite{Frembs22}. Additionally, in subsection \ref{asym} we will use the construction outlined in Proposition \ref{decomp}, in conjunction with the multi-linear polynomials from Lemma \ref{polinomials}, to establish new upper bounds on the number of qubits needed for evaluating symmetric Boolean functions.

\subsection{General Boolean functions}

Could we some up with a procedure to find measurement assignments for general Boolean functions, and not just for the symmetric ones? Ideally, this procedure should find the sparsest multi-linear polynomials while requiring a small computational effort. These conditions guarantee that the classical pre-processing is not excessively time-consuming and that the \textnormal{NMQC}$_\oplus$ computations are as efficient as possible. In this subsection, we propose a new procedure that reduces the cost of determining the measurement assignments, compared to the $\mathcal{EF}$ construction from section \ref{universality}. This reduction is achieved without compromising the efficiency of the measurement assignments, given that the generated multi-linear polynomials have equal or even smaller sparsities. 

\paragraph{Construction 4: The Krawtchouk representation ($\mathcal{KR}$)}
This solution uses the ANF of the Boolean functions to translate each monomial ($\prod_{i \in S} x_i$), to a multi-linear polynomial. However, in contrast to previous techniques, this translation resorts to the symmetric properties of these monomials. First, one can notice that all these terms are symmetric Boolean functions if evaluated independently. Therefore, it is possible to translate each of them to their simplified value vector \cite{Canteaut05} using a function $\nu:(\{0,1\}^n\rightarrow\{0,1\})\rightarrow \{0,1\}^{n+1}$, defined as
\begin{align}
  \nu  \bigg(\prod_{i \in S} x_i \bigg) =& \delta_{|S|}\bigg(\sum_{i\in  S}x_i\bigg) \\ =&\big(\underbrace{0}_{0},\underbrace{0}_{1},\hdots,\underbrace{0}_{|S|-1},\underbrace{1}_{|S|}\big)\ ,
\end{align}
 \noindent with $|S|\subseteq [n]$ and $\delta_{|S|}$ the Kronecker delta function at $|S|$.

\begin{definition}\textnormal{ \textbf{(Simplified value vector).}} The simplified value vector of a symmetric Boolean function ($v_f:\mathbb{N}\rightarrow \{0,1\}$) is defined as a function mapping the Hamming weights of the input string to the corresponding Boolean values
\begin{equation}
    v_f(|x|)=f(x),
\end{equation}
\noindent and recall that $|x|=\sum_{i=1}^{n} x_i$ stands for the Hamming weight of $x$. Thus,  the bit-string $v_f$ of length $n+1$ unambiguously determines the symmetric function,
\begin{equation}
    v_f=\{\underbrace{v_0}_{|x|=0},\ \underbrace{v_1}_{|x|=1},\ \underbrace{v_2}_{|x|=2},\ ...\ ,\ \underbrace{v_n}_{|x|=n}\}\ .
\end{equation}
\end{definition}

\noindent Given this compact representation for each monomial, one can translate such vectors to a multi-linear polynomial using the Krawtchouk transform ($\mathcal{K}:\mathbb{Z}^n\rightarrow\mathbb{R}^n$), described in Appendix \ref{Krawtchouk}. This transform computes the $n+1$ different Fourier coefficients of an arbitrary symmetric semi-Boolean function $f_{\ast}:\{0,1\}^n\rightarrow \mathbb{R}$, such that, 
\begin{equation}
   \frac{1}{2^n} \mathcal{K}^T\left( \begin{matrix}
    v_f(0) \\
    v_f(1) \\
    \vdots \\
    v_f(n) \\
    \end{matrix}\right) = \begin{bmatrix}
    \widehat{f}(|S|=0) \\
    \widehat{f}(|S|=1) \\
    \vdots \\
    \widehat{f}(|S|=n) \\
    \end{bmatrix} \ .
\end{equation}

\noindent In a subsequent step, the multi-linear polynomial equivalent to the monomials, which are symmetric Boolean functions, will be constructed using the computed Fourier coefficients as follows,  
    \begin{equation}\label{sympoly}
         \mathsf{poly}_{f}(x) =    \pi*\Big(\sum_{i=0}^n\widehat{f}(|S|=i)*\sum_{S_j\subseteq [n],|S_j|=i}g^{-1}\Big(  \bigoplus_{l\subseteq S_j} x_l  \Big ) \Big) \ . 
\end{equation}

\noindent In the end, each of the resulting polynomials is summed, or subtracted, from the others to obtain the final multi-linear polynomial that represents the queried function, 
\begin{align}
   \mathsf{poly_f}(x) &= \pi *\sum_{S \subseteq [n]} (-1)^{t_S}*c_S *  \frac{1}{2^n} \mathcal{K}^T\bigg(\nu\bigg (\prod_{i \in S} x_i \bigg ) \bigg) \\
&= \pm \mathsf{poly_1}(x)  \pm \mathsf{poly_2}(x)\pm ... \pm \mathsf{poly_n}(x) \equiv f(x)\ .
\end{align}
\noindent with $t_S\in \{0,1\}$.

Additionally, it will be necessary to prove that valid measurement assignments are generated for the intended Boolean functions. To do so, we show that the multi-linear polynomials obtained for each monomial using Equation \eqref{sympoly} are correct, and these can be summed or subtracted to build the final multi-linear polynomial. These two proofs are sufficient for our purpose because these are the only differences we have introduced with respect to the $\mathcal{EF}$ construction.

\begin{lemma}\label{rfbound}
     Any symmetric semi-Boolean function, $f_{\ast}:\{0,1\}^n \rightarrow \mathbb{R}$, has at most $n+1$ different Fourier coefficients that can be computed using the Krawtchouk transformation of dimension $n+1$, with $\mathcal{O}(n^5)$ operations.
\end{lemma}
\begin{proof}
The proof of this lemma is given in Appendix \ref{FouriertoKrawt}.
\end{proof}

\noindent Thus, Lemma \ref{rfbound} proves that the multi-linear polynomials created for each monomial are correct. In addition, it provides proof that Construction 4 can translate each monomial in the ANF exponentially more efficiently to a multi-linear polynomial compared to the use of the Fourier transform for the same process, as employed in the previous constructions.

Finally, the fact that one can use the computed polynomials, multiplied or not by a factor of $-1$, relates to the simple observation that they are equivalent to the modular two operations, $\mathsf{poly}_{f_{\ast}}(x)\ mod\ 2= -\mathsf{poly}_{f_{\ast}}(x)\ mod\ 2$. Therefore, being positive or negative does not change the parity of the angle, which determines the result of the \textnormal{NMQC}$_\oplus$ computation (Equation \eqref{cosmod}). Hence, one can use this new degree of freedom to reduce the number of qubits necessary to evaluate the selected Boolean function. These reductions can be achieved by attributing signs to the polynomials, $(-1)^{t_S}$, such that sets containing equal parity bases have their coefficients summing to zero. Consequently, this would remove the necessity to use these specific qubits that would otherwise be part of the resource and required to be measured. A simple demonstration of this new possibility is provided in Example \ref{rfexample} where we allocated the signs such that one qubit could be removed from that specific \textnormal{NMQC}$_\oplus$ computation. Furthermore, one could now define heuristics that attempt to cancel out as many coefficients as possible. We leave, however, the search for such heuristics for future work.

\begin{example}\label{rfexample}
For the function $h:\{0,1\}^3\rightarrow\{0,1\}$ with input strings of size three,
\begin{equation}
    h(x_1,x_2,x_3)=x_1*x_2 \oplus x_2*x_3\ ,
\end{equation}the Krawtchouk transform applied to the simplified value vector to each monomial computes the Fourier coefficients,
\begin{equation}
    \frac{1}{4}\begin{bmatrix}
    1 & \ \ 2 &\ \ 1\\
    1 &\ \ 0  &-1\\
    1 &  -2& \ \ 1\\
\end{bmatrix}*\underbrace{\begin{bmatrix}
    0 \\
    0 \\
    1 \\
    \end{bmatrix}}_{\nu(x_1*x_2)} = \begin{bmatrix}
    1/4 \\
     -1/4\\
    1/4 \\
    \end{bmatrix} \ , \  \frac{1}{4}\begin{bmatrix}
    1 & \ \ 2 &\ \ 1\\
    1 &\ \ 0  &-1\\
    1 &  -2& \ \ 1\\
\end{bmatrix}*\underbrace{ \begin{bmatrix}
    0 \\
    0 \\
    1 \\
    \end{bmatrix}}_{\nu(x_2*x_3)} = \begin{bmatrix}
    1/4 \\
     -1/4\\
    1/4 \\
    \end{bmatrix} \ .
\end{equation}

\noindent  These are used to build the respective multi-linear polynomials using Equation \eqref{sympoly},
\begin{equation}
    \mathsf{poly_{x_1x_2}}(x)=\pm \Big(\dfrac{1}{2}x_1+\dfrac{1}{2}x_2-\dfrac{1}{2}x_1\oplus x_2\Big )\ and\    \mathsf{poly_{x_2x_3}}(x)=\pm \Big (\dfrac{1}{2}x_2+\dfrac{1}{2}x_3-\dfrac{1}{2} x_2\oplus x_3\Big ) \ .
\end{equation}

\noindent Then, the resulting polynomials are combined to obtain the multi-linear polynomial equivalent to $h(x)$,
\begin{align}
   \label{firstpoly} \mathsf{poly_{h}}(x)&=+ \Big(\dfrac{1}{2}x_1+\dfrac{1}{2}x_2-\dfrac{1}{2}x_1\oplus x_2\Big ) - \Big (\dfrac{1}{2}x_2+\dfrac{1}{2}x_3-\dfrac{1}{2} x_2\oplus x_3\Big ) \\ \label{secondpoly} &= \dfrac{1}{2}x_1+\cancel{\dfrac{1}{2}x_2}-\cancel{\dfrac{1}{2}x_2}-\dfrac{1}{2}x_3 -\dfrac{1}{2} x_1\oplus x_2 +\dfrac{1}{2} x_2\oplus x_3 \\  &= \dfrac{1}{2}x_1-\dfrac{1}{2}x_3 -\dfrac{1}{2} x_1\oplus x_2 +\dfrac{1}{2} x_2\oplus x_3
    \ .
\end{align}

\noindent  This polynomial is then used to create the corresponding measurement assignment. In this case, it translates to a 4-qubit GHZ state, $\ket{\Psi_{GHZ}^4}= 1/\sqrt{2} \big(\ket{0}^{\otimes 4} + \ket{1}^{\otimes 4} \big)$, the linear functions
\begin{equation}
     L_1(x)=x_1,\  L_2(x)=x_3,\  L_3(x)=x_1\oplus x_2,\ and\ L_4(x)=x_2\oplus x_3\ ,
\end{equation}
and the corresponding measurement operators
\begin{align}
    & M_1(s_1)=(\neg s_1)\sigma_x+ s_1\sigma_y,\ M_2(s_2)=(\neg s_2)\sigma_x-s_2\sigma_y,\ \\
    & M_3(s_3)=(\neg s_3)\sigma_x- s_3\sigma_y,\ and\ M_4(s_4)=(\neg s_4)\sigma_x+ s_4\sigma_y\ .
\end{align} 
\noindent Moreover, the function evaluated on all possible inputs, for $x\in \{0,1\}^3$, in the \textnormal{NMQC}$_\oplus$ model is equal to 
$(-1)^{h(x_1,x_2,x_3)}\ket{\Psi_{GHZ}^5}$ as intended. Notice also that assigned the positive sign to one of the polynomials and the negative sign to the other in Equation \eqref{firstpoly}. This specific attribution removed a term ($x_2$ in Equation \eqref{secondpoly}) that would otherwise demand an extra qubit of the resource state.
\end{example}

Concerning the computational cost of determining the measurement assignments, this construction avoids resorting to the Fourier transform, which computes a number of Fourier coefficients that is exponential in $n$, the size of the input string. Furthermore, as proven in Lemma \ref{rfbound}, there are only $n+1$ different Fourier coefficients that characterize symmetric Boolean functions. Thus, it is only necessary to select $n+1$ lines of the Fourier matrix to determine these coefficients. That is precisely one of the reductions applied to obtain the Krawtchouk transform and explains one of the computational advantages of this solution. Additionally, the transformation does not use these $n+1$ lines of the Fourier matrix to determine the coefficients. Alternatively, each line of the Krawtchouk transformation computes the demanded Fourier coefficients based on the simplified value vector of the SBF. Therefore, it reduces the determination of each Fourier coefficient to a counting problem that can be solved in polynomial time. This observation entails another exponential advantage with respect to the approach bared on the Fourier transform.

In order to quantify the advantage behind the use of the Krawtchouk transform, we need to characterize the cost of the approach  that uses the Fourier transform. This cost is defined by the degree of the Boolean function.

\begin{definition} \textnormal{ \textbf{(Degree).}}
The degree of a Boolean function $f$ is equal to the cardinality of the largest monomial in its ANF,
\begin{equation}
    \mathsf{deg}(f)= max\{|S|\  |   \ c_S\neq 0\}
\end{equation}

\end{definition}
 
\noindent In particular, one needs a Fourier matrix of size $2^\mathsf{deg(f)} \times 2^\mathsf{deg(f)}$ to translate the largest monomial into a multi-linear polynomial. This process enforces the existence of at least $2^\mathsf{2deg(f)}$ operations to determine the coefficient of this singular term, imposing a run-time complexity lower bound of $\Omega (2^\mathsf{2deg(f)})$ to this method.  In addition, one must consider that it will be necessary to compute the multi-linear polynomials of all the other monomials ($m=\sum_{s\in [n]}c_S$) and then sum them all to form the resulting polynomial. Accounting for these additional operations, we roughly estimate the computational cost of the $\mathcal{EF}$ construction to be $\mathcal{O}\big(poly(m)* 2^\mathsf{2deg(f)}\big)$\footnote{The term $poly(m)$ represent a function of the type $f(x)=x^c$ with $c\in \mathbb{N}$.}. In contrast, for the Krawtchouk transform, the computational effort equivalent to determining the Fourier coefficients of the largest monomial is defined by applying a matrix of size $(\mathsf{deg}(f)+1) \times (\mathsf{deg}(f)+1)$. Furthermore, the entries of this matrix are determined in polynomial time, as proven by Lemma \ref{rfbound}. Thus, the process used to determine all the necessary Fourier coefficients decreases from an exponential to a polynomial magnitude.

Although we have found a significant reduction in one of the subprocesses necessary to determine the measurement assignments, there is still an unavoidable exponential scaling with respect to the degree of the Boolean function.

\begin{proposition}\label{polyconst}
For any Boolean function $f:\{0,1\}^n\rightarrow \{0,1\}$,  a measurement assignment for the \textnormal{NMQC}$_\oplus$ model can be determined with a computational complexity that scales exponentially with the degree of the Boolean function and polynomially with the number of monomials, $m=\sum_{s\in [n]}c_S$, in its ANF representation, i.e. $\mathcal{O}\Big(m*\Big((\mathsf{deg}(f)+1)^5+2^{\mathsf{deg}(f)}\Big )\Big)$.
\end{proposition}

\begin{proof}
The Krawtchouk transform gives the number of operations required to determine all the necessary Fourier coefficients. This matrix has $(\mathsf{deg}(f)+1)^2$ entries, which can be computed with at most $(\mathsf{deg}(f)+1)^3$ operations. Therefore, in total, there are at most $m*(\mathsf{deg}(f)+1)^5$ operations, as proven by Lemma \ref{rfbound} to determine all the necessary Fourier coefficients. These coefficients build the multi-linear polynomials equivalent to the monomials $m$ in the ANF of the Boolean function, using Equation \eqref{sympoly}. Subsequently, there is an additional cost due to the need to sum the coefficients related to the multi-linear polynomials resulting from the various monomials. This specific process assembles the final multi-linear polynomial and has at most $m*2^{\mathsf{deg}(f)}$ operations. The exponential term $2^{\mathsf{deg}(f)}$ results from the minimal sparsity for the largest monomial in the Boolean function. That this is the minimal sparsity is proven by a theorem in  \cite{Mori2019} showing that the sparsity of any Boolean function $f:\{0,1\}^n\rightarrow \{0,1\}$ is at least $ 2^{\mathsf{deg}(f)}$. Combining the two terms, we arrive at the cost value described in the proposition.
\end{proof}

The bound proven in Proposition \ref{polyconst} shows that the $\mathcal{KR}$ construction is almost as efficient as it could be to determine the measurement assignments that result from this technique. This statement is justified by demonstrating that the computational effort of determining the multi-linear polynomials regarding their sparsity is at most polynomial. Therefore, exponential scalings only exist because the object we are constructing is of exponential size. In contrast, the Fourier transform introduces an additional exponential effort to determine each one of the Fourier coefficients. Thus, determining valid measurement assignments with this transformation implies a computational effort of exponential size with respect to the sparsity of the multi-linear polynomials. Furthermore, this result highlights the importance of finding and employing constructions that produce sparse multi-linear polynomials. These reduces the number of qubits for the \textnormal{NMQC}$_\oplus$ computation and, in principle, facilitates the process of determining their measurement assignments.

\vspace{0.5cm}
\noindent \textbf{Parameterized construction 1: The symmetric complements  ($\mathcal{SC}(\zeta)$)} This construction merges all the contributions we have presented so far. In particular, it incorporates the advantage of the $\mathcal{CSF}$ construction to determine sparse-multi-linear polynomials for symmetric Boolean functions, with the computational reductions obtained with the $\mathcal{KR}$ construction. 

The construction starts from the ANF of a Boolean function $f:\{0,1\}^n\rightarrow\{0,1\}$ and creates two additional Boolean functions based on a previously selected symmetrization function $\zeta$.

\begin{definition}\textnormal{ \textbf{(Symmetrization function).}}
A symmetrization function is a function $\zeta:\big(\{0,1\}^n\rightarrow\{0,1\}\big) \rightarrow \big ( \{0,1\}^n\rightarrow\{0,1\} \big)$ that receives an arbitrary Boolean function, $f:\{0,1\}^n\rightarrow \{0,1\}$, and maps it to a symmetric Boolean function, $f_{\mathsf{sym}}:\{0,1\}^n\rightarrow \{0,1\}$, such that $\zeta(f(x))=f_{\mathsf{sym}}(x)$, for all $x$.
\end{definition} 

\noindent The first Boolean function we need to define is the symmetrized function of $f$ respective to $\zeta$. This function is computed by applying the symmetrization function $\zeta$ to $f$ in order to obtain $f_{\mathsf{sym}}$. Then, the complement of $f$ with respect to its symmetrized function $f_{\mathsf{sym}}$, designated as $\thicktilde{f}$, has to be computed.

\begin{definition}\textnormal{ \textbf{($\thicktilde{f}(x)$).}}
The complement of $f$ with respect to the symmetrization function $\zeta$ is defined as, 
\begin{equation}
    \thicktilde{f}(x)= f(x) \oplus \zeta(f(x)) = f(x) \oplus f_{\mathsf{sym}}(x)\ .
\end{equation} 
\end{definition}

\noindent The construction translates all three Boolean functions $f(x),\ \zeta(f(x)),$ and $f(x)\oplus\zeta(f(x))$ to multi-linear polynomials. The $\zeta(f)$ function is translated using the $\mathcal{CSF}$ construction, given that it is a symmetric Boolean function, while for $f$ and $\thicktilde{f}$ the $\mathcal{KR}$ construction is employed to generate the following polynomials 

\begin{equation}
    f(x),\ \thicktilde{f}(x)\  \xmapsto{\ \mathcal{KR}\ }\  \mathsf{poly}_{f_1}(x),\  \mathsf{poly}_{\thicktilde{f}}(x);\ \ \ \   \zeta(f)(x)\  \xmapsto{\ \mathcal{CSF}\ }\  \mathsf{poly}_{\zeta(f)}(x)\ .
\end{equation}

\noindent From the resulting multi-linear polynomials, two valid solutions can be built.  The first is $\mathsf{poly}_{f_1}$ which derives directly from $f$ using the $\mathcal{KR}$ construction. The second solution results from the linear sum of $\mathsf{poly}_{\thicktilde{f}}(x)$ and $ \mathsf{poly}_{\zeta(f)}(x)$, provided that $\thicktilde{f}(x)\oplus \zeta(f)(x) = f(x)$. In the end, the sparsest of the two multi-linear polynomials will be used to generate the measurement assignment for the \textnormal{NMQC}$_\oplus$ computation,
\begin{equation} \mathsf{poly}_{f}= min \big \{\mathsf{sparsity}\big(\mathsf{poly}_{f_1}\big ),\ \mathsf{sparsity}\big(\mathsf{poly}_{\thicktilde{f}}(x)+\mathsf{poly}_{\zeta(f)}(x)\big )\big \} .
\end{equation}

In order to demonstrate how advantageous this construction could be, we will select a specific symmetrization function $\zeta_C$. This symmetrization function takes the Boolean function $f$ and introduces all the missing terms such that if there is a term of degree $k$ in its ANF, then the corresponding CSF $C^k$ will occur in its symmetrized version. 

\begin{definition}
\textnormal{ \textbf{($\zeta_C(x)$).}}
The symmetrization function $\zeta_C(x):\big(\{0,1\}^n\rightarrow\{0,1\}\big) \rightarrow \big ( \{0,1\}^n\rightarrow\{0,1\} \big)$ is defined as
\begin{equation}
    \zeta_C(f(x))= \bigoplus_{k\in T} c_k *C^k(x),\  with\ T=\{|S|\ |\ c_S \neq 0\}.
\end{equation}
\end{definition} 

\noindent We can, then, use this function $f$ to generate a multi-linear polynomial for an arbitrary  Boolean function. In particular, we shall consider a class of Boolean functions for which there is an asymptotical advantage of this new constriction over all the previously described. 

\begin{definition}\textnormal{ \textbf{($[aC^k]$).}}
The class of almost complete symmetric functions $[aC^k]$ is composed by families of Boolean functions $f$ such that, 
\begin{equation}
    f(x) \in [aC^k] \implies f(x) = C^k(x) \bigoplus_{i=0}^t \Bigg (  \prod_{j \in S_i} x_j  \Bigg )
\end{equation}
\noindent with $|S_i|=k$ for all $i\in [0,t]$, $x\in \{0,1\}^n$ and $k,t\in \mathbb{N}$.
\end{definition}

\noindent Function in this class miss only a fixed set of terms is missing to be precisely equal to a complete symmetric function $C^k$.

\begin{lemma}\label{acbound}
   All families of Boolean functions in the $[aC^k]$ class have multi-linear polynomials with sparsities that scale at most with the same asymptotical rate as its symmetrized function $C^k$. 
\end{lemma}
\begin{proof}
 Suppose the $\mathcal{SC}$ construction is used to compute the multi-linear polynomials for a family of Boolean functions $f$ contained in the class of $[aC^k]$. In this case, these will result from combining two other multi-linear polynomials: the multi-linear polynomials corresponding to the symmetrized Boolean functions $\zeta_C(f)$ and  multi-linear polynomials corresponding to the complement of $f$ with regard to $\zeta_C(f)$, $\thicktilde{f}$. These particular functions will be by the definitions of $[aC^k]$ and $\zeta_C$ equal to some CSF $C^k$ for the former, and equal to some constant number terms ($\bigoplus_{i=0}^t (  \prod_{j \in S_i} x_j )$) for latter. Hence, the multi-linear polynomials resulting from the $\zeta_C(f)$ have sparsities equal to $C^k$, and the polynomials corresponding to the complement $\thicktilde{f}$ have a constant sparsity. Therefore, assymptotically, only the multi-linear polynomials corresponding to the symmetrized functions contribute to their sparsity.  
\end{proof}

The first observation that follows from Lemma \ref{acbound} is that for specific elements in the $[aC^k]$ class, one can take direct advantage of the sparse CSF polynomials described in Section \ref{symmetric} to lessen their upper bound. In Example, \ref{acexample}, a specific family of Boolean functions was defined for which the best qubit scaling went from $\Theta(n^{64})$ to $\Theta(n^{32})$ simply with the use of the $\mathcal{SC}$ construction with the $\zeta_C$ function and a CSF polynomial. Furthermore, similar reductions are feasible for the whole $[aC^k]$ class. In contrast, the $\mathcal{EF}$ construction and similar techniques were utterly blind to the possibility of inheriting the sparsity of SBFs.

\begin{example}\label{acexample}
Consider function $aC^{64}(x)= C^{64}(x)\oplus x_1x_2....x_{64}$, which is built using a CSF function and a constant term added to it. If one uses the $\mathcal{EF}$ construction, the resulting multi-linear polynomials have the following form, 
\begin{equation}
    poly_{aC^{64}(x)}= \Bigg(\sum_{i=0}^{64} c_i  \bigg ( \sum_{S \subseteq [n], |S|=i} \bigoplus_{j\in S} x_j \bigg ) \Bigg )- poly_{x_1x_2...x_{64}}(x)
\end{equation}
\noindent with a sparsity of $\mathsf{sparsity}(poly_{aC^{64}(x)})=\sum_{i=0}^{64}\binom{n}{i}-1$, that scales asymptotically with $\Theta(n^{64})$. Suppose the $\mathcal{SC}$ construction is used on the same function. In this case, it builds different multi-linear polynomials using the symmetrized function $\zeta_C(aC^{64})$, which is exactly equal to $C^{64}$ and the complement $\thicktilde{aC^{64}}$, which is equal to the term $x_1x_2...x_{64}$. Therefore the resulting multi-linear polynomials will have the following form

\begin{equation}
    poly_{aC^{64}(x)}= \Bigg(\sum_{i=0}^{32} c_i  \bigg ( \sum_{S \subseteq [n], |S|=i} \bigoplus_{j\in S} x_j - \sum_{S \subseteq [n] , |S|=n-i+1}  \bigoplus_{j \in S} x_j  \bigg ) \Bigg )+ poly_{x_1x_2...x_{64}}(x)\ ,
\end{equation}

\noindent which has a sparsity of $\mathsf{sparsity}(poly_{aC^{64}(x)})= 2\sum_{i=1}^{31} \binom{n}{i}+\binom{n}{32}+\sum_{j=33}^{64} \binom{64}{j}$, that scales asymptotically with $\Theta(n^{32})$, and which determines the number of qubits necessary to evaluate the Boolean function deterministically in the \textnormal{NMQC}$_\oplus$ model. 
\end{example}

In summary, the construction proposed entails the possibility of exploring SBF to generate more efficient \textnormal{NMQC}$_\oplus$ computations, not only for symmetric but also for general Boolean functions. From a practical point of view, we introduced a specific symmetrization function $\zeta_C$ to demonstrate the impact the construction could have on a whole class of Boolean families $[aC^k]$ by reducing the computational requirements of all their elements regarding an evaluation in the \textnormal{NMQC}$_\oplus$ model. However, there is plenty of space for improvement in selecting and creating other symmetrization functions. One could define more elaborated functions that symmetrize the arbritary Boolean functions so that classes larger than $[aC^k]$ are reduced. Consequently, further study of this construction is in order to illustrate its potential advantages.

\section{Bounding the resources for the \texorpdfstring{$\textnormal{NMQC}_\oplus$}{NQMC} model}\label{tbound}

This section discusses bounds on the resources needed for the deterministic evaluation of Boolean functions in the \textnormal{NMQC}$_\oplus$ model. For that,
we will analyze the Clifford hierarchy of the measurement operators, and the sparsity of the resource state.
Both measures follow naturally from the definition of the model, which performs computations via measurements applied to a resource state with auxiliary linear computations. In particular, these measures will be examined in light of the new constructions proposed in the previous section, producing mainly new upper bounds for the \textnormal{NMQC}$_\oplus$ model.

\subsection{Clifford hierarchy of the measurement operators}

The measurement operators required for deterministic computations in the \textnormal{NMQC}$_\oplus$ model are described by Equation \eqref{measurements} as single qubit rotation around the Bloch sphere. In particular, they can be decomposed into a unitary operation $U_i$ and a measurement on the computational basis,
 
\begin{center}
\begin{quantikz}
    \meter{\text{$M_i$}}  & = &  & \gate{U_i} & \meter{} & \text{.}
\end{quantikz}
\end{center}

\noindent Therefore, the corresponding unitaries $U_i$ capture the complexity of the measurement operations imposed by the Boolean function in the \textnormal{NMQC}$_\oplus$ model \footnote{Interestingly, determining the complexity of unitary operations has been an important research topic in quantum computing, with extensive applications from circuit synthesis to quantum error correcting codes \cite{Matthew13, Bullock06, Juha04, Zeng08}.}. Intuitively, operations are considered simple if they require a small number of gates and low-depth circuits to be executed. In contrast, other more complex unitaries require more computational resources \cite{Gary20,Zeng08,Nadish20}. In order to obtain a better characterization, Gottesman and Chuang proposed the Clifford hierarchy to define specific levels of complexity for unitary operations \cite{GottesChuang99}.

\begin{definition}
\textnormal{ \textbf{(Clifford hierarchy).}}
The Clifford hierarchy is defined recursively as
\begin{equation}
    Cl_k:=\{\ U\ |\ UPU^\dagger \in Cl_{k-1},\ \forall P \in \mathcal{P}\}
\end{equation}
\noindent for $\mathcal{P}$ the Pauli group. 
\end{definition}

\noindent This specific hierarchy has been connected to several fundamental results. For instance, the second level of the hierarchy, $Cl_2$, defines the Clifford group, which was proven to be efficiently simulable by classical computers, a result known as the Gottesman-Knill theorem \cite{Gottesman98}. It is also known that any element of the third level, $Cl_3$, together with the Clifford group generates a universal gate set for quantum computation. For instance, the gate set Clifford$+$T is often used as a universal gate set, denoting by T the unitary from the third level of the hierarchy \cite{Moska13, Zeng08}.

For the \textnormal{NMQC}$_\oplus$ model, it was proven that the level of the measurement operators in the Clifford hierarchy directly bounds the complexity of the Boolean functions that can be evaluated deterministically. More precisely, in \cite{Frembs22}, the authors prove that any measurement assignment with measurement operators 
belonging to at most the $k$th level of Clifford hierarchy cannot deterministically compute Boolean functions of degree greater than $k$. This result defines a lower bound on the hierarchy level required for deterministic computations. 

Reciprocal to this finding, we determine an upper bound for the highest level of the Clifford hierarchy required for a particular Boolean function. For that, we define the granularity of semi-Boolean functions, given that this quantity determines the type of measurement operators used in a measurement assignment (Equation \eqref{fmeasure} and Equation \eqref{measurements}).

\begin{definition}
\textnormal{ \textbf{(Granularity).}}
The granularity of a semi-Boolean function $f_{\ast}:\{0,1\}^n \rightarrow \mathbb{R}$ is given by
\begin{equation}
    \mathsf{gran}(f)= max_{S\subseteq [n]}  \ \mathsf{gran}\big(\widehat{f}(S)\big)\ 
\end{equation}
with the granularity of each Fourier coeficient or any rational number being equal to the smallest $k$ such that $\widehat{f}(S)*2^k\in \mathbb{N}.$
\end{definition}

\noindent Then, we prove that the $\mathcal{KR}$ construction always produces measurement assignments based on semi-Boolean functions that have a granularity bounded by the degree of the Boolean function. 

\begin{lemma}\label{coef}
All Fourier coefficients of a multi-linear polynomial resulting from the $\mathcal{KR}$ construction applied to a Boolean function $f$ have their granularity $\mathsf{gran}\big(\widehat{f}(S)\big)$ bounded by $\mathsf{deg}(f)$.
\end{lemma}
\begin{proof}
The Fourier transform applied to each monomial ($\prod_{k \in S_j} x_k$) in the ANF of the Boolean function generates a multi-linear polynomial $\mathsf{poly}_i$ with Fourier  coefficients $\widehat{f}(S_{j,1})$ of the type $\frac{l_{S_{j,k}}}{2^i}$, with $S_{j,k}\subseteq S_j$, $k\in [2^i]$, and  $i=\mathsf{deg}\big( \prod_{k \in S_j} x_k\big )$ by definition of the Fourier transformation \cite{Odonnel2021}. Therefore, the same is true for the Krawtchouk transform, which computes the same Fourier coefficients.

Then, all such individual polynomials are summed or subtracted based on the coefficients $c_i \in \{0,1\}$ to generate the final multi-linear polynomial, 
\begin{align}
&\underbrace{(-1)^{c_{1}}\bigg( \overbrace{\frac{l_{S_{1,1}}}{2^{i_1}}}^{\widehat{f}(S_{1,1})} \bigoplus_{k\in S_{1,1}} x_k + \frac{l_{S_{1,2}}}{2^{i_1}}\hdots + \hdots + \frac{l_{S_1,i_1}}{2^{i_1}}\hdots \bigg)}_{\mathsf{poly}_1}+ \\
&\underbrace{(-1)^{c_{2}} \bigg( \frac{l_{S_{2,1}}}{2^{i_2}}\hdots  + \frac{l_{S_{2,2}}}{2^{i_2}}\hdots + \hdots + \frac{l_{S_2,i_2}}{2^{i_2}}\hdots \bigg)}_{\mathsf{poly}_2} + \hdots = \mathsf{poly}_f \ .
\end{align}
\noindent Therefore the resulting Fourier coefficients ($\widehat{f}(S)$) for each 
of the parity bases ($\bigoplus_{k\subseteq S} x_k$) in the final multi-linear polynomial are equal to the sum of all non-zero Fourier coefficients that have the same parity base in the multi-linear polynomials resulting from each monomial, $\frac{l_S}{t}= \sum_{S_{j,k}=S} \frac{(-1)^{c_{j}}*l_{S_{j,k}}}{2^{i_{j}}}$. Consequently, the denominator $t$ of the resulting value is equal to the least common multiple between all denominators, given that we are summing and subtracting rational numbers. Additionally, all the denominators are of the form $\frac{1}{2^i}$, for which the $l.c.m(\frac{1}{2^{i_1}},\frac{1}{2^{i_2}},...,\frac{1}{2^{i_n}})=max(\frac{1}{2^{i_1}},\frac{1}{2^{i_2}},...,\frac{1}{2^{i_n}})$. This proves that, for any Fourier coefficient in the final multi-linear polynomial, the corresponding denominator is at most equal to the largest denominator present in the sum. Thus, the largest denominator is at most equal to the denominators of the monomials defining the degree of the Boolean function, which implies that the resulting Fourier coefficients have a denominator that is at most equal to $\frac{1}{2^{\mathsf{deg}(f)}}$, proving the lemma. 
\end{proof}

In addition, we relate the level of the Clifford hierarchy of measurement operators with the angles' granularity using the semi-Clifford hierarchy.

\begin{definition}\textnormal{ \textbf{(Semi-Clifford hierarchy).}}
A gate $U$ is in the semi-Clifford hierarchy of level-k if $U=A_1DA_2$, such that $A_1,A_2 \in Cl_2^N$ are Clifford gates and $D\in Cl_k^N$ is the diagonal. The set of k-th level Clifford gates is denoted by $SCl_k^N$.
\end{definition}

\begin{lemma}\label{operators}
All measurement operators of type $ M=\cos \big(\phi \big)\sigma_x +\sin \big(\phi \big) \sigma_y$ are in the level-$\mathsf{gran}\big(\phi\big)+1$ of the Clifford hierarchy. 
\end{lemma}

\begin{proof}
This type of operator can be characterized by the semi-Clifford hierarchy, given that it can be decomposed  as, 
\begin{align}
    I*M_i*\sigma_x &= I \big(\cos \big( \phi \big)\sigma_x +\sin \big( \phi \big) \sigma_y \big) \sigma_x \\
    &= \cos \big( \phi \big)I -i\sin \big( \phi \big) \sigma_z) \\
    &= R_Z\big( \phi \big) \ ,
\end{align}
\noindent noting that $R_Z\Big(\frac{\pi l}{2^j}\Big)$ is a diagonal unitary gate of level-$2^{j+1}$ in the Clifford hierarchy \cite{Gary20}, while $I,\sigma_x \in Cl_2$. We conclude that gates of the form $M$ are  in the $\mathsf{gran}\big(\phi_i\big)$-level of the semi-Clifford hierarchy. Additionally, in \cite{Zeng08} it was proven that $Cl_k^1=SCl_k^1$, such that all $M\in SCl_{2^{i+1}}^1= Cl_{2^{i+1}}^1$. Therefore we have proven that a measurement operator of the type of $M$ is in the $\mathsf{gran}\big(\phi\big)$-level of the Clifford hierarchy. 
\end{proof}

Finally, it is possible to define an upper bound on the level of the Clifford hierarchy of the measurement operators necessary for deterministic evaluations of Boolean functions.

\begin{theorem}\label{maxcliff}
For any Boolean function $f$, the $\mathcal{KR}$ construction generates deterministic \textnormal{NMQC}$_\oplus$ evaluations that use measurement operators belonging to at most the $\mathsf{deg}(f)$-level of the Clifford hierarchy 
\end{theorem}

\begin{proof}
With Lemma \ref{coef} it is possible to bound the granularity of the Fourier coefficients determined by the Krawtchouk transform for an arbitrary Boolean function. These coefficients then define the angle of the measurement operator based on Equation \eqref{fmeasure}. Therefore, it is equally possible to bound the granularity of all the angles used in the measurement assigned by the degree of the Boolean function. In addition, the granularity of $\theta_i$ from Equation \eqref{measurements} can be made equal to zero with a linear shift of the Boolean function (a free operation). Consequently, the granularity $\mathsf{gran}(\theta_i+\phi_i)$ of the total angle defining the measurement operator can be bounded by $  \mathsf{deg}(f)-1$. Now, it suffices to use Lemma \ref{operators} to bound the level of the Clifford hierarchy of the selected measurement operators. This finishes the proof given that the Krawtchouk transform always generates measurement assignments that evaluate the intended Boolean function deterministically. 
\end{proof}

This result demonstrates that minimizing the Clifford hierarchy of the measurement operators is always possible and that the degree of the Boolean function alone defines the minimum and maximum level required. For the upper bound, we used the granularity of the Fourier spectrum, which abstractly was the same concept used in \cite{Frembs22} for the lower bound. Furthermore, in Boolean analysis, the granularity of the Fourier spectrum is also used to prove lower bounds \cite{chistopolskaya2018parity}. However, in those cases, only Boolean functions are considered. In contrast, for the \textnormal{NMQC}$_\oplus$ model, due to the symmetry of the cosine function, there is a large class of equivalent semi-Boolean functions that compute the intended Boolean functions correctly. Therefore, these results can not be translated to the \textnormal{NMQC}$_\oplus$ model. Nevertheless, it was with the concept of the minimum possible granularity of all the possible semi-Boolean functions that the author of \cite{Mori2019} proved a lower bound for the sparsity of arbitrary Boolean functions. In the end, we highlight the importance of these measures to bound the necessary resources for the deterministic evaluation of Boolean functions in \textnormal{NMQC}$_\oplus$.

\subsection{Lower bound conjecture for symmetric Boolean functions}

The previous section introduced several methods to determine measurement assignments that deterministically compute a given Boolean function in the \textnormal{NMQC}$_\oplus$ model. Moreover, all these solutions consider the general GHZ state as the resource, which was proven to be the optimal resource \cite{Hoban2011a, Wolf2001}. Nevertheless, no construction guarantees that the sparsest polynomials are created. This indicates that determining the minimum number of qubits necessary to evaluate a Boolean function is a nontrivial problem. Additionally, it is possible to show that each Boolean function in the model is equivalent to a Bell inequality \cite{Hoban2011a}. Therefore, one could suspect that the difficulty of the former problem is associated to the known hardness of the latter one \cite{Brunner12}\footnote{Determining the maximal violation of a Bell inequality obtained from  measurements on a quantum state was proven to be NP-hard even in the tripartite case \cite{Kempe11}}. Nonetheless,  finding lower bounds for specific subsets of Boolean functions was proven possible in \cite{Hoban2011a}, and it would be exciting to extend these results further. 

Finding the smallest possible resource is generally a complex problem because there are an infinite number of valid measurement assignments for a particular Boolean function. In particular, all these solutions exist because the phase introduced after the measurement process, which defines the result of the computation, is invariant with respect to a $2\pi$ rotation (see Equation \eqref{cosmod}). Therefore, various multi-linear polynomials of type $\mathsf{poly}:\{0,1\}^n\rightarrow \pi\mathbb{Z}$ are equivalent in the model when all their outputs correspondingly map to the same equivalence class of the quotient set $\pi(\mathbb{Z}/2\mathbb{Z})$. Consequently, the minimum number of qubits necessary to evaluate a Boolean function in the \textnormal{NMQC}$_\oplus$ model can be defined by the sparsest element in this set.

\begin{definition}\label{minspar}\textnormal{ \textbf{(Minimal sparsity).}}
 The minimum number of qubits necessary for the deterministic evaluation of a Boolean function $f:\{0,1\}^n\rightarrow\{0,1\}$ in the \textnormal{NMQC}$_\oplus$ model can be defined as the sparsest multi-linear polynomial equivalent to the function, i.e.
 \begin{equation}
     min \Big  \{ \mathsf{sparsity} \big(\mathcal{F}(g)\big) \Big |\ g : \{0,1\}^n \rightarrow \mathbb{Z}\ s.\ t.\  \forall_{x\in \{0,1\}^n}\ f(x)\equiv g(x)\mod 2 \Big \}\ .
 \end{equation}
\end{definition}

\noindent Fortunately, while the number of multi-linear polynomials equivalent to the Boolean functions is infinite, the effective search space is finite \cite{Odonnel2021}. Nevertheless, the size of the search space is double exponential with the size of the input string \cite{Frembs22}. Thus, an exhaustive search for the sparsest multi-linear polynomials is not feasible.

Despite this tremendously ample search space, there are some restrictions that can make the problem considerably simpler for the set of SBF. More precisely, we will impose the constraint that only symmetric measurement assignments are allowed for their evaluation. 

\begin{definition}\textnormal{ \textbf{(Symmetric measurement assignment).}} A symmetric measurement assignment for an \textnormal{NMQC}$_\oplus$ computation is composed of a set of linear functions $L$, which is limited to having either all parity bases of a specific size or none,
\begin{equation}
L=\Big\{\bigoplus_{l\in S} x_l\  \Big |\ S \subseteq [n],\ |S|\in T \Big \},\ with\ T\subseteq [n]\ ,
\end{equation}

 \noindent and the corresponding dichotomic measurement operators, which have to be equal for all parity bases of the same size, i.e. $|S_i|=|S_j| \implies M_i(x) = M_j(x)\ .$
 
\end{definition}

\noindent Furthermore,  all symmetric measurement assignments for an \textnormal{NMQC}$_\oplus$ computation  are equivalent to a symmetric multi-linear polynomial ($\mathsf{poly}_{\mathsf{sym}}:\{0,1\}^n\rightarrow\mathbb{R}$) of the following form:
  \begin{equation}
\mathsf{poly}_{\mathsf{sym}}(x)=\pi*\bigg(\sum_{i\in T}\widehat{f}(|S|=i)*\sum_{S_j\subseteq [n],|S_j|=i}g^{-1}\bigg(  \bigoplus_{l\subseteq S_j} x_l  \bigg ) \bigg),\ with \ T\subseteq [n]\  \ . 
 \end{equation}

\noindent Therefore, with this restriction, the multi-linear polynomials that describe SBF have at most $n+1$ different Fourier coefficients concerning the size of the input string  $n$. Unfortunately, the set of solutions is still of an exponential size to perform a search. Nevertheless, knowing the particular polynomials with the smallest sparsities is not necessary to determine 
lower bounds. It is sufficient to independently learn the sparsities of the polynomials with the smallest sparsities. Hence, one could define a decision problem that only questions if there is a multi-linear polynomial equivalent to the SBF with a particular sparsity. Fortunately, the combinations of this restriction with the $\mathcal{KR}$ transform allow us to define an algorithm able of answering this decision problem in polynomial time.   

\begin{lemma}\label{ltest}
An algorithm with polynomial scaling runtime can decide whether there exists a symmetric multi-linear polynomial with a selected subset of non-zero Fourier coefficients, $F=\big\{ \widehat{f}(|S|=i)\ \big|\ i\in T,\ T \subseteq [n] \big\}$, that computes a particular symmetric Boolean function deterministically in the \textnormal{NMQC}$_\oplus$ model.
\end{lemma}
\begin{proof}
The proof is presented in Appendix \ref{test}.
\end{proof}
\noindent Furthermore, this algorithm can be employed to identify the order of magnitude of the sparsity of an arbitrary SBF in polynomial time (For more details, see Appendix \ref{conjecture}). Subsequently, applying these processes to the set of CSF functions led to the following conjecture,

\begin{conjecture}\label{lowerbound}
The number of qubits in a GHZ state necessary for the deterministic evaluation of a complete symmetric function $C^k$, of constant degree $k$, with an $n$ bit input string and a symmetric measurement assignment, scales as $\Omega(n^{\floor*{k/2}-1})$.
\end{conjecture}
\noindent\textit{Support:} This conjecture is supported empirically by a numerical verification process which was applied to all CSF that can be defined with an input size up to 345 bits. The details of this numerical verification process are presented in Appendix \ref{conjecture}. Additionally, the conjecture is consistent with the lower bound for the $C^2$ function and the $\mathsf{AND}$ function. In particular, for the latter, each CSF function is exactly equal to the function for a specific input size, as for all $k\in\mathbb{N}$ and $|x|=k$ , $C^k(x)=\mathsf{AND}(x)$.
\vspace{0,2cm}

\noindent Remember that the set of CSFs behaves as a building block for all SBFs, Equation \eqref{symmetricrep}. Therefore, the correctness of this conjecture implies an equivalent lower bound for families of SBF with a constant degree.

\begin{corollary}
Any family of symmetric Boolean functions defined with a constant degree and evaluated with a symmetric measurement assignment in the \textnormal{NMQC}$_\oplus$ model requires a number of qubits which is asymptotically bounded by $\Omega(n^{\floor*{\mathsf{deg}(f)/2}-1})$.
\end{corollary}
\noindent\textit{Support:} This corollary follows directly from Conjecture \ref{lowerbound} and the definition of SBF's.
\vspace{0.3cm}

From an abstract point of view, the conjecture and the respective numerical support further strengthen the idea that the degree of the Boolean function and the size of the input string determine the number of qubits necessary for deterministic computations in the \textnormal{NMQC}$_\oplus$ model. Moreover note that the symmetry condition was not imposed on the lower bounds obtained previously. However, the optimal measurements assignments for the SBF $\mathsf{AND}$ and $C^2$ are symmetric. It remains to be proven, but it could be possible that the optimal measurement assignments for SBFs concerning the minimum number of qubits are always symmetric. Supposing that this last hypothesis and the conjecture were both proven true, then the multi-linear polynomials presented in Lemma \ref{polinomials} would be optimal and generate the most resource-efficient measurement assignments for the respective SBFs in the \textnormal{NMQC}$_\oplus$ model. 

\subsection{Qubit count for symmetric Boolean functions}\label{asym}

In the previous subsection, we conjectured a lower bound for the number of qubits in a GHZ state required for an \textnormal{NMQC}$_\oplus$ evaluation of symmetric Boolean functions. 
Complementary to this result, comparing this bound with the existing upper bound is particularly interesting. For this, it will be necessary to define 
new upper bounds based on our proposed constructions 
and their corresponding multi-linear polynomials.

Formerly, the upper bound for SBF resulted from the $\mathcal{EF}$ construction presented in \cite{Mori2019}. In particular, this construction generates multi-linear polynomials that had a sparsity equal to $\sum_{i=0}^{\mathsf{deg}(f)} \binom{n}{i}$, which entails an asymptotical scaling of $\mathcal{O}\big(n^{\mathsf{deg}(f)}\big)$. Recently, Frembs \textit{et al}. \cite{Frembs22} presented a sparser description for the multi-linear polynomials representing CSFs that directly reduces the sparsity of the multi-linear polynomials used for SBFs. More precisely, the proposed multi-linear polynomials have a sparsity of $\mathsf{sparsity}(C^k(x))=\sum_{i=0}^{k-1} \binom{n}{i}$, which translates to an asymptotical scaling for SBF equivalent to $\mathcal{O}\big(n^{\mathsf{deg}(f)-1}\big)$. Futhermore, this last upper bound was the best known up to now, to the best of our knowledge. 

In order to find new upper bounds, it is necessary to define the impact of the $\mathcal{CSF}$ construction on the sparsity of the multi-linear polynomials generated for SBFs. For this purpose, we must consider that this construction creates multi-linear polynomials using the subset of CSFs with a degree equal to a power of two exclusively. Therefore, it can be proven that the sparsity of the resulting multi-linear polynomials also depends uniquely on the same set. 

\begin{lemma}\label{sSBF}
The sparsity of any multi-linear polynomial equivalent to a symmetric Boolean function $f$ is bounded as
\begin{equation}
    \mathsf{sparsity}\big(f(x)\big) \leq \sum_{i=0}^n c_i \prod_{r \in R_i} \Big (  \mathsf{sparsity}\big(C^{2^r} (x)\big)+1 \Big )  ,  
\end{equation}
\noindent with $R_i \subseteq [\ceil*{\log(i)}]$, such that $\sum_{r\in R_i} 2^{r}=i$.
\end{lemma}
\begin{proof}
A valid multi-linear polynomial can be generated for any SBF based on Equation \eqref{symmetricrep} and the general decomposition defined in Proposition \ref{decomp}. Therefore, the sparsity of SBF is always upper bounded by the sparsity of the polynomials built from this representation. In particular, the same expression directly defines the relation between the resulting polynomial's sparsity and that of the CSFs used.   
\end{proof}

\noindent Consequently, for families of SBF with a constant degree, the asymptotic upper bound depends only on the subset of CSFs with degrees that are powers-of-two,

\begin{equation}
\mathsf{sparsity}\big(f(x)\big) = \mathcal{O}\bigg(  \prod_{r \in R_d}\mathsf{sparsity}\big(C^{2^r}(x) \big) \bigg)    
\end{equation}
\noindent \noindent with $R_d \subseteq [n]$, such that $\sum_{r\in R_d} 2^{r}=\mathsf{deg}\big(f\big)$.
\noindent The fact that only the CSF function with the largest degree is considered can be explained by the fact that the CSFs polynomials have monotonically increasing sparsity. 

Subsequently, this bound has to be combined with the minimal sparsities known for the CSF with a degree equal to a power of two. In particular, for instances with a degree equal to or smaller than 64, the most compact multi-linear polynomials are the ones identified in Lemma \ref{polinomials}, and, for higher degree functions, are the ones described in \cite{Frembs22},
\begin{equation}
\mathsf{sparsity}\Big( C^{2^r}(x)\Big)= sparsity\Big( C^{k}(x)\Big)=
\begin{cases}
\mathcal{O}\big(n^{2^{r-1}}\big) =\mathcal{O}\big(n^{k/2}\big)  	\Leftarrow & r=2,3,4,5,6 \\
\mathcal{O}\big(n^{2^{r}-1}\big) =\mathcal{O}\big(n^{k-1}\big) 	\Leftarrow    & r>6
\end{cases} ,
\end{equation}

\noindent with $2^r=k$. Therefore when these instances are used to determine the asymptotic sparsities of CSF with arbitrary degrees, we obtain the following bounds, 
\begin{equation}\label{newup}
     \mathsf{sparsity}\Big ( C^k(x) \Big )= \mathcal{O}\big( n^{\rho} \big),\ with  \ \rho=\Bigg ({\sum_{r\in R_k , r\leq 6} 2^{r-1} +\sum_{r\in R_k , r> 6} 2^r-1}\Bigg),
\end{equation}

\noindent with $R_k \subseteq [\ceil*{\log(k)}]$, such that $\sum_{r\in R_k} 2^{r}=k$. Likewise, the same bounds will be valid for for SBFs with a constant degree,
\noindent       
\begin{equation}
    \mathsf{sparsity} \big(f(x)\big )=\mathsf{sparsity}\big ( C^k(x) \big ),\ with \ k=\mathsf{deg}(f(x)) \ .
\end{equation}

Notice that any CSF that resorts uniquely to the multi-linear polynomials from Lemma \ref{polinomials} in their decomposition (Equation \eqref{decomC}) has an upper bound of $\mathcal{O}(n^{k/2})$. This bound holds for all CSF with a degree equal to or smaller than 127. For these instances the upper bound is tight to the conjectured lower bound. This indicates that if the general form for multi-linear polynomials identified in  Lemma \ref{polinomials} is proven correct, then all CSF functions would have an upper bound tight to the conjectured lower bound. Yet, without this proof, the best upper bound is the one described by Equation \eqref{newup}. In Figure \ref{fig:scaling} we represent this newly obtained upper bound, together with the lower bound.

In summary, the CSF construction, in combination with the use of sparse multi-linear polynomials, considerably lowers the number of qubits required to evaluate SBF in the \textnormal{NMQC}$_\oplus$, as well as the respective upper bounds. Furthermore, the sparsities for SBF functions can be described as depending on a subset of CSFs. This suggests that the effort to reduce the sparsities of SBFs should be guided towards the subset of CSF.

\begin{figure}[H]
    \centering
    \includegraphics[scale=0.53]{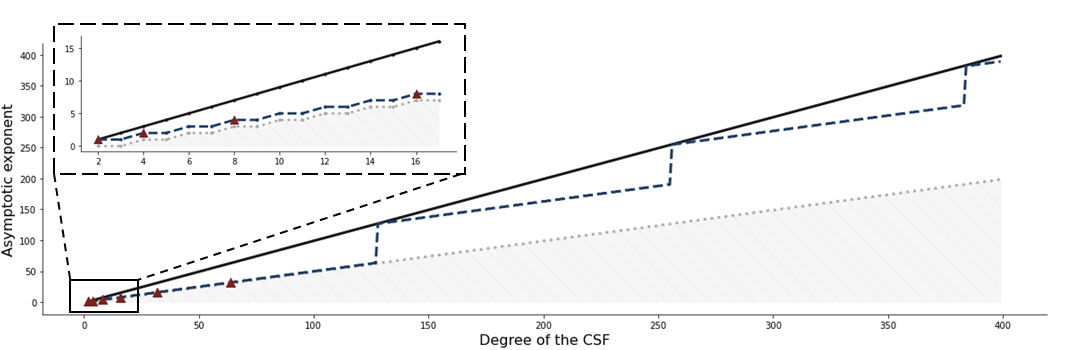}
    \caption{Representation of the relations between the individual CSF (ordered by degree) and the asymptotic growth rate of the GHZ state necessary for deterministic evaluation with the \textnormal{NMQC}$_\oplus$ model. The continuous black line represents the most compact polynomials identified in \cite{Frembs22}. The dashed dark blue line represents the most compact polynomials constructed using the decomposition from Proposition \ref{decomposition} and the multi-linear polynomials from Lemma \ref{polinomials}. Additionally, the red triangles are the CSF computed exactly with the polynomials from Lemma \ref{polinomials}. Finally, the dotted grey line represents our conjectured lower bound for the set of CSF (Conjecture \ref{lowerbound}).  }
    \label{fig:scaling}
\end{figure}

\section{Circuit realization of the \texorpdfstring{$\textnormal{NMQC}_\oplus$}{NQMC} model}\label{circuits}

\subsection{\texorpdfstring{$\textnormal{NMQC}_\oplus$}{NQMC} circuits}\label{NMQCcircuit}

This subsection translates the measurement assignments for the \textnormal{NMQC}$_\oplus$ model discussed in the paper into realizable computations. This analysis is pertinent because the previous works related to this precise model have discussed these computations at an abstract level only. In other words, they do not describe them as circuits composed of elementary operations from a fixed gate set. Here we provide a more complete picture of the computational process, describing the necessary operations to create the GHZ states, characterizing the circuits that compute the auxiliary linear functions, and synthesizing or approximating the required measurement operators, keeping in mind that the latter are rarely part of the fixed gate set. 

This will enable a more direct comparison with fully classical circuits for Boolean function evaluation, in particular with respect to depth, which is associated with the temporal extension of these circuits.

\begin{definition}\textnormal{ \textbf{(Depth).}} The depth of an \textnormal{NMQC}$_\oplus$ circuit is equal to the size of the longest sequence of elementary gates, quantum or classical, in it.
\end{definition}

\noindent The width of the circuits, which describes the maximal computational space occupied during the execution of the circuits, will be equally important. Generally, parallelizable processes can have a smaller depth but at the cost of a larger width.

\begin{definition}\textnormal{ \textbf{(Width).}} 
The width of an \textnormal{NMQC}$_\oplus$ circuit is equal to the largest number of elementary gates, quantum or classical, executed simultaneously in it.
\end{definition}

\noindent The last property that we will account for is the number of gates that the circuits contain. This value relates to previous properties and is also an essential element for general computations as it defines the cost of an evaluation in absolute terms. 

\subsubsection{Linear pre- and post-processing circuits}

The subroutine for the linear pre-processing computes all the linear functions that are part of a specific measurement assignment for a Boolean function. In particular, it contains a number of linear functions equal to the number of necessary qubits 
in the GHZ state. This number is defined by the sparsity of the multi-linear polynomial that generates the measurement assignment. Furthermore, each of these linear functions, $L_i=\oplus_{i\in S}x_i$, can be highly parallelized, creating a circuit composed of XOR gates with a depth of $\log_2(|S|)$, a width of $|S|$, and a number of gates equal to $|S|-1$ (Figure \ref{fig:linearcirc} (a)). Consequently, the total subroutine is the parallel composition of all linear functions (Figure \ref{fig:linearcirc} (b)). Thus, it has a depth of 
$max\{|S|\ \big |\ \widehat{f}(S)\neq 0\}$, a width of
$\sum_{S\subseteq [n]} |S|* \delta\big(\widehat{f}(S)\big)$, and a number of gates equal to $\sum_{S\subseteq [n]} \big(|S|-1\big )* \delta\big(\widehat{f}(S)\big)$, based on the Fourier coefficients $\widehat{f}(S)$ of the multi-linear polynomial $\mathsf{poly}_f(x)$ equivalent to the measurement assignment. 

Likewise, the linear post-processing circuit of an \textnormal{NMQC}$_\oplus$ computation is a very simple subroutine. It solely sums all the measurement outcomes linearly, computing the parity of the bit string. The computation of this linear function can also be parallelized such that the circuit will have a depth of $\log_2(k)$ with $k=\mathsf{sparsity}(\mathsf{poly}_f)$\footnote{ We will label $k$ as the sparsity of the multi-linear polynomial that defines the measurement assignment repeatedly throughout this section. }, a width of $k$, and $k-1$ gates. 
\begin{figure}[H]
\centering
\begin{subfigure}[b]{0.53\textwidth}
         \centering
\includegraphics[scale=0.22]{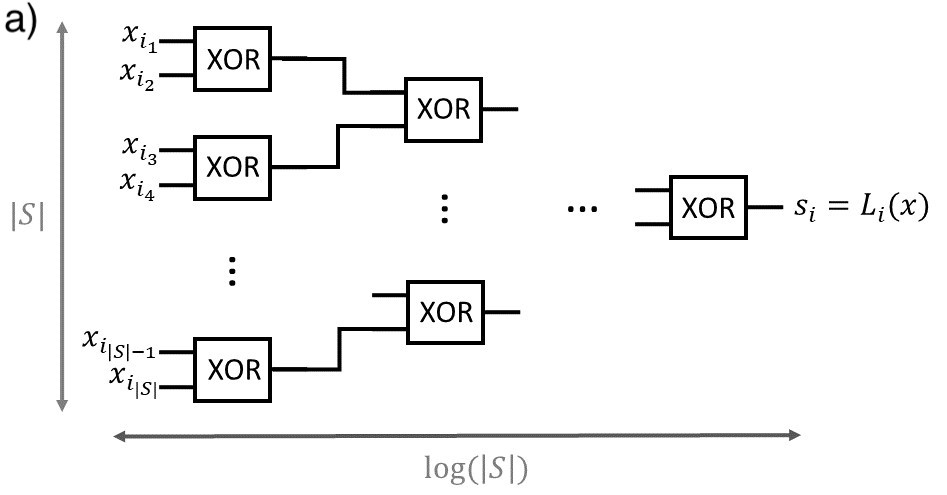}
         \label{fig:linear}
     \end{subfigure}
     \begin{subfigure}[b]{0.46\textwidth}
         \centering
\includegraphics[scale=0.12]{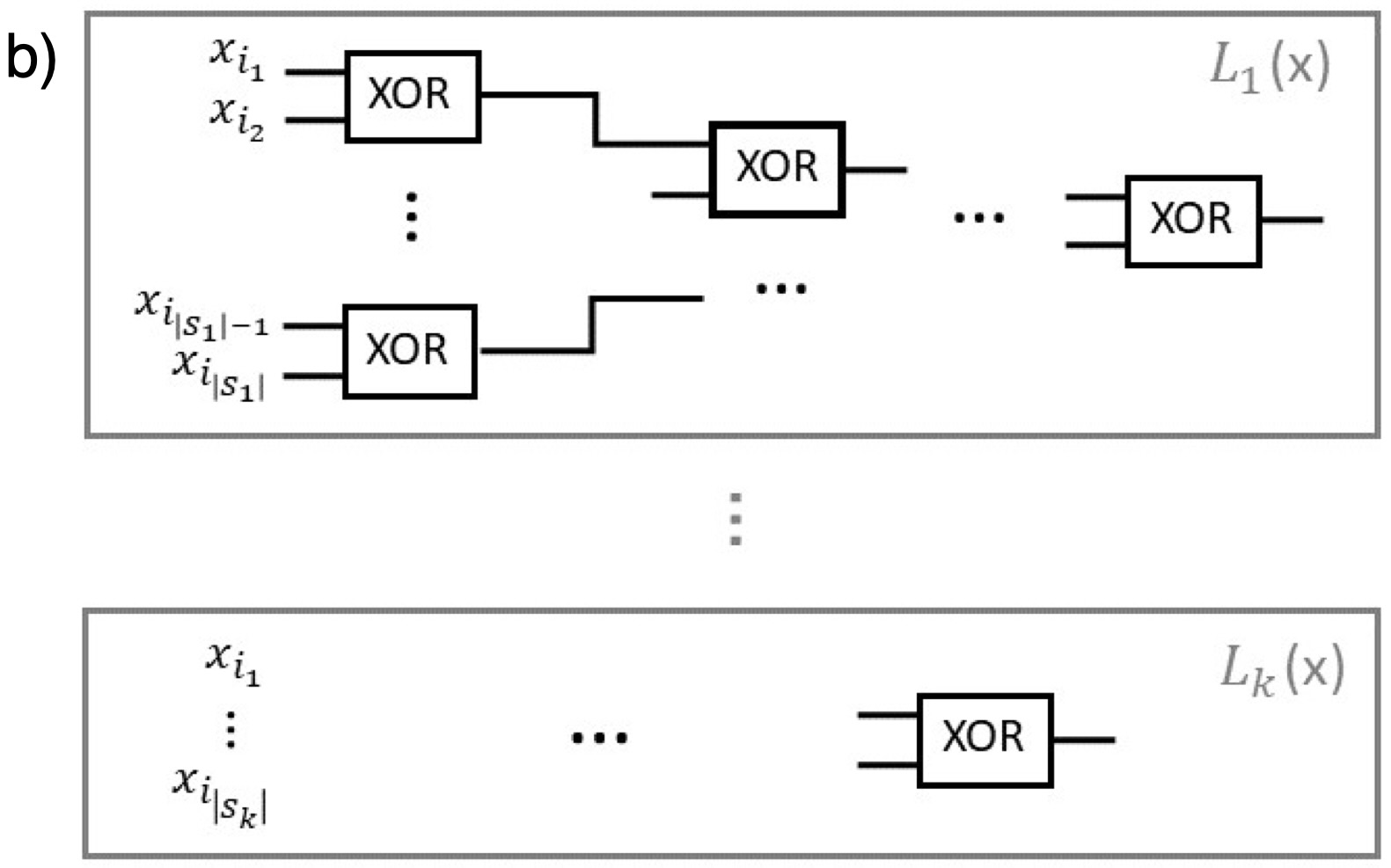}
         \label{fig:paralellinear}
     \end{subfigure}
     \caption{a) The general form and dimensions of a circuit computing a linear function with $|S|$ inputs. b) Parallel composition of the linear functions in the linear pre-processing routine.   }
        \label{fig:linearcirc}
\end{figure}

\subsubsection{GHZ state creation}\label{ghzcreation}

Let us review two different circuit families used to create GHZ states. The first one is more efficient considering the number of gates while having a simpler process. However, the second compensates for the complexity and the additional operations with the capacity to prepare a GHZ state with any number of qubits in constant quantum depth.

\vspace{0,3cm}
\noindent \textbf{(a) Construction in logarithmic quantum depth}
 \vspace{0,3cm}
 
 \noindent The first construction, defined in \cite{Cruz19}, creates a $n$-qubit GHZ state from $n$ qubits initially in the $\ket{0}$ state using only CNOT and H gates. For that, it applies a single H gate to the first qubit and a sequence of CNOT gates in a cascade pattern afterward. In particular, the first CNOT gate takes the first qubit as the control qubit and the qubit in position $n/2$ as the target. Then, in the second layer, two CNOT gates are applied, the first qubits and the qubit in position $n/2$ action as control, and the qubits in positions $n/4$ and $3n/4$ as the target qubits of these gates. This pattern repeats layer after layer, applying twice the number of CNOT gates of the previous layer while using the already entangled qubits as the control qubits and non-entangled qubits as targets. For an illustration of this process, see Figure \ref{fig:ghzcreation} (a). 

It is interesting to characterize this process for a specific \textnormal{NMQC}$_\oplus$computation. Given a GHZ state with a number of qubits equal to the sparsity $k$ of the polynomial is necessary, the corresponding circuit will have a depth of $\log_2(k)$, a width of $k$, and the circuit contains $\mathcal{O}\big(k\big)$ gates.

\vspace{0.5cm} 
\noindent \textbf{(b) Construction in constant quantum depth} \vspace{0,3cm}
 
 \noindent The second technique considered to create GHZ states is particularly interesting because it uses only a constant number of quantum layers to generate a GHZ with an arbitrary number of qubits \cite{Quek22,Watts19}. Nevertheless, these fixed quantum layers have to be interleaved with classical computations of logarithmic size, which results in a total depth - adding quantum and classical processing - similar to the previous one. Also, the method relies on adaptively chosen quantum operations, which depend on previously obtained measurement outcomes. This adaptivity could be a problem as we are considering a non-adaptive measurement-based quantum computation model. However, as this adaptivity does not depend on the input string or the specific computation that will be performed, we consider this a valid alternative. 

The specific protocol presented in \cite{Quek22} creates a $n$-qubit GHZ state from $n$ qubits in the $\ket{0}$ state using six layers of quantum gates. The first quantum layer applies $H$ gates to odd-numbered qubit wires. Then, the second quantum layer applies $n/2$ CNOT gates between neighboring qubit pairs, with odd-numbered control qubit wires and even-numbered target qubit wires. After this process, $n/2$ Bell pairs were created. The third quantum layer applies $n/2-1$ CNOT gates between the qubits of neighbouring Bell pairs, such that one qubit of the Bell pair is a target qubit for the previous Bell pair, and the other is the control qubit for the next Bell pair. Afterward, the next quantum layer (the fourth) measures all the qubits which were target qubits of the CNOT gates applied in the previous process. Then, the results of the previous measurements define an $X$ correction that will be applied in the fifth quantum layer. In particular, such $X$ corrections will be applied to the qubits directly below the measured qubits, and they depend on the parity of all the measurement outcomes previous to that specific qubit. Simultaneously, in the fifth layer, the measured qubits are reset. Finally, the sixth and final quantum layer entangles the qubits recycled to the $\ket{0}$ state after the measurement, with the resulting quantum state. An illustration of the whole process is depicted in Figure \ref{fig:ghzcreation} (b).

In the end, for an \textnormal{NMQC}$_\oplus$ computation, this process produces a quantum state with a number of qubits defined by the sparsity of the respective multi-linear polynomial. This implies that the depth of the respective circuit will be $6+\log_2(k/2-1)$,  with a width of $k$ and a total number of gates bounded above by $\frac{k^2}{8}+\frac{11}{4}k-4$.

\begin{figure}[H]
\centering
\begin{subfigure}[b]{0.50\textwidth}
         \centering
\includegraphics[scale=0.200]{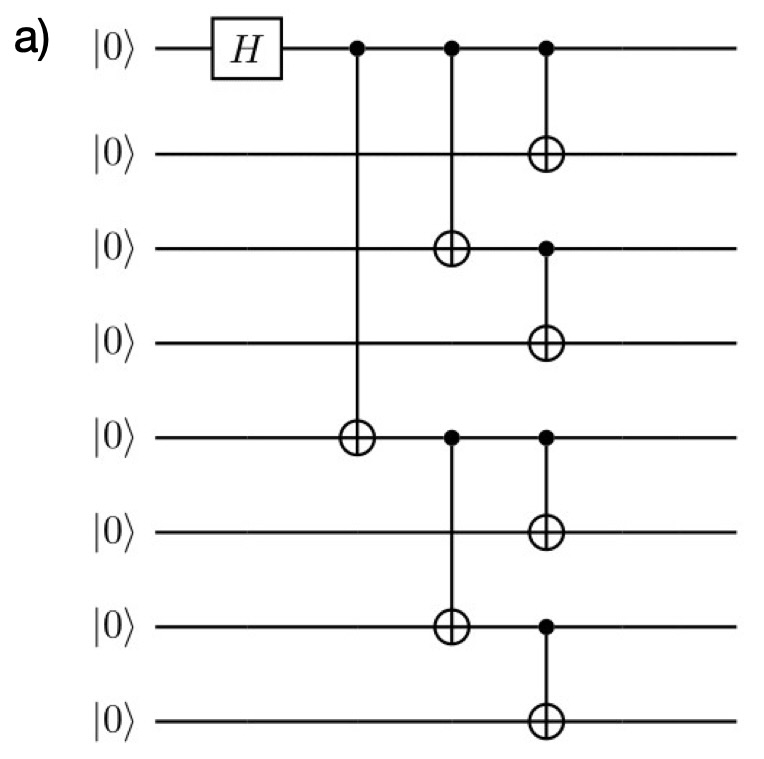}
         \label{fig:logghz}
     \end{subfigure}
     \begin{subfigure}[b]{0.49\textwidth}
         \centering
\includegraphics[scale=0.22]{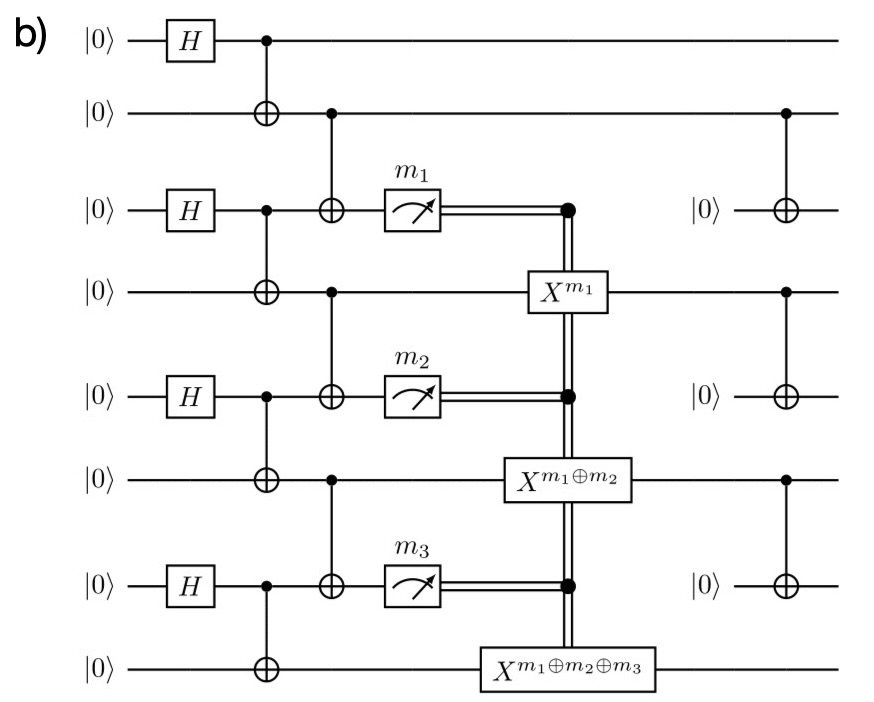}
         \label{fig:constghz}
     \end{subfigure}
     \caption{a) Logarithmic depth circuit for the creation of an 8-qubit GHZ state with a recursive CNOT application pattern. b) Constant quantum depth circuit with extra linear measurement post-processing the creation of an 8-qubit GHZ state.}
        \label{fig:ghzcreation}
\end{figure}

\subsubsection{Measurement operators}

The measurement operators imposed by the measurement assignment are generally not part of the gate set. Therefore, creating these operators from the available gate set will be necessary. As is common in the literature, we will consider the Clifford+$T$ gate-set. In particular, we need to generate single qubit $Z$-rotations because these are always equal to the measurement operators used by the model after applying a Pauli $X$ operator. Therefore, the optimal circuit for the measurement operator has at most one extra elementary gate, followed by the optimal circuit for a particular $Z$-rotation (Figure \ref{fig:measdecomp}). This observation allows us to use several established results on the circuit synthesis of $Z$-rotations \cite{Selinger15,mosca16,Neil15,Gary20}.

\begin{figure}[H]
\begin{center}
\begin{quantikz}
    \meter{\text{$M_i$}}  & = & & \gate{X} & \gate{R_Z\big(\frac{t_i\pi}{2^{n_i}}\big)} & \meter{}\\
     &\approx &  &\gate{X} & \gate{H} & \gate{T} & \qw & \hdots & & \gate{S} & \meter{}
\end{quantikz}
\end{center}
    \caption{Exact and approximated circuit synthesis of the measurement operator. The former is based on arbitrary $R_Z$ rotations, the latter on the Clifford+ $T$ gate set. }
    \label{fig:measdecomp}
\end{figure}
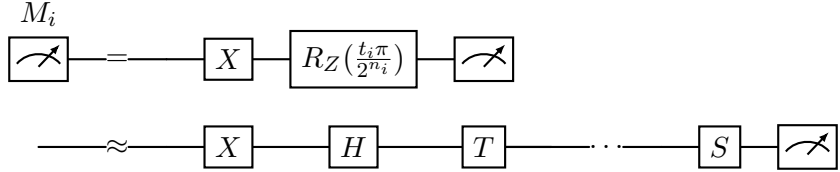

 The first significant result is that no $R_Z(\frac{\pi}{2^n})$ operator for $n \geq 2$ can be exactly synthesized using the Clifford+$T$ gate set \cite{Selinger15,Neil15}. Therefore, using this gate set, Boolean functions with a degree greater than two cannot be computed deterministically in the \textnormal{NMQC}$_\oplus$ model.

\begin{corollary}\label{cliffordlimit}
    No Boolean function $f:\{0,1\}^n\rightarrow \{0,1\}$ with $\mathsf{deg}(f)\geq 3$ can be computed deterministically with the \textnormal{NMQC}$_\oplus$ model realized with the Clifford+$T$ gate set.
\end{corollary}
\begin{proof}
The proof follows from the fact that all the measurement operators necessary for evaluating Boolean functions with a degree greater than two are in third level or higher of the Clifford hierarchy, as proven in \cite{Frembs22}. Therefore, each one of these measurement operators is equivalent to a $Z$-rotation of magnitude $\frac{\pi}{2^n}$ with $n\geq 2$. Thus, it is impossible to synthesize these gates exactly with the Clifford+$T$ gate set \cite{Selinger15,Neil15}, and the whole \textnormal{NMQC}$_\oplus$ computation will have a non-zero error term. 
\end{proof}

\noindent This critical implication imposes that higher degree Boolean functions can only be computed with a certain accuracy that depends on the approximation of the measurement operators. 

For our purpose, we will need to define approximations that guarantee a certain level of accuracy $\epsilon \in [0,1]$ of the \textnormal{NMQC}$_\oplus$ computations. For that, a relationship between the approximations of the $Z$-rotations and the whole computation is required. From \cite{Selinger15}, it is possible to define the approximation of a $Z$-rotation concerning the operator norm for an arbitrary distance $\bigtriangleup$,
\begin{equation}
    \bigtriangleup \geq \Big|\Big|R_{Z}^\ast\Big(\frac{\pi*t}{2^{n}} \Big)\ket{\Psi} -R_{Z}\Big(\frac{\pi*t}{2^{n}}\Big )\ket{\Psi}   \Big|\Big|  \ .
\end{equation}
\noindent with $R_{Z}^\ast$ representing the approximated $Z$-rotation, using a circuit with approximately $K+4\log_2(1/\bigtriangleup)$ \footnote{$K$ is approximately 10, and this value was proven to be lower bounded by -9. Therefore, the algorithm from \cite{Selinger15} produces approximations that are at most an additive constant from the optimal $Z$-approximations. } $T$-gates. Therefore, the approximation of each single measurement operator $M_i$, to be applied to one of the $k$ qubits in GHZ state, is bounded by
\begin{align}
    \bigtriangleup_i \geq\  &\Big|\Big|\bigotimes_{j=0}^{i-1} I_j \otimes M_i \otimes \bigotimes_{l=i+1}^{k} I_l \ket{\Psi} -\bigotimes_{j=0}^{i-1} I_j \otimes M_i^{\ast} \otimes \bigotimes_{l=i+1}^{k} I_l \ket{\Psi}  \Big|\Big |\\
    \geq\  &\Big|\Big| M_i \ket{\Psi} - M_i^\ast \ket{\Psi}   \Big|\Big| \\ \geq\  &\Big|\Big|\sigma_x R_{Z}\Big(\frac{\pi*t_i}{2^{n_i}} \Big)\ket{\Psi} -\sigma_x R_{Z}^\ast\Big(\frac{\pi*t_i}{2^{n_i}}\Big )\ket{\Psi}   \Big|\Big| \\
    \geq\  & \Big|\Big|R_{Z}^\ast\Big(\frac{\pi*t_i}{2^{n_i}} \Big)\ket{\Psi} -R_{Z}\Big(\frac{\pi*t_i}{2^{n_i}}\Big )\ket{\Psi}  \Big|\Big| \ .
\end{align}

\noindent Then, we assume that all the measurement operators have the same error bound $\bigtriangleup$. Also, as all these measurement operators commute, one can interpret that they are applied sequentially, allowing us to bound the error after all the measurements by the sum of all the individual errors \cite{Vazirani93},

\begin{align}
    k*\bigtriangleup =\sum_{i=0}^k \bigtriangleup_i  \geq\   &\Big|\Big|\bigotimes_{i=0}^{k}  M_i  \ket{\Psi} -\bigotimes_{i=0}^{k}  M_i^{\ast} \ket{\Psi}  \Big|\Big |\\    \ .
\end{align}
Finally, we can  use the fact that the total deviation distance between the resulting probability distributions between the ideal quantum state that computes the Boolean function deterministically and the resulting quantum state is at most four times the norm error defined for the quantum state \cite{Vazirani93}, such that, 
\begin{align}
    \epsilon &=(-1)^{f(x)}- \bra{\Psi_{GHZ}^k}\otimes_{i=1}^k M_i(s_i) \ket{\Psi_{GHZ}^k} \\
            & \geq 4*k*\bigtriangleup
\end{align}

After determining the relation between the total error of the computations and the local errors due to the measurement operators, one can define a specific number of $T$ gates necessary for a certain precision $\epsilon$ of the \textnormal{NMQC}$_\oplus$ computation. In particular, the corresponding circuit will have a depth of $4*c*\log_2(\frac{4k}{\epsilon})=4*c*(2+\log_2(k)+\log_2(\frac{1}{\epsilon}))$ with $c$ equal to a constant that accounts for all the additional gates in the circuit, which are generally proportional to the number of $T$ gates \cite{Moska13} \footnote{In reference \cite{Moska13}, this constant appears to be close to 2.5. However, for our analysis, we will leave it undefined}. Additionally, we end up with a width of $k$ and a total number of gates that is bounded by $4*c*k*\log_2(\frac{4k}{\epsilon})$.

\subsubsection{Depth, width and number of gates}

To characterize the \textnormal{NMQC}$_\oplus$ circuits, we will combine the previously described subroutines. In particular, we will combine them to minimize the circuit depth. Therefore, the linear pre-processing and the GHZ state generation will be performed simultaneously, as they have no dependencies. Then, the results of the linear functions define the set of operations that will be realized on the GHZ-state qubits in the measurement process. Then, the linear post-processing will compute the parity of the measurement outcomes to obtain the output of the Boolean function. The whole process is represented in Figure \ref{fig:arquiteture}.

With the architecture of the circuits defined, it is possible to characterize the relevant properties resulting from specific measurement assignments. In particular, the multi-linear polynomials associated with the measurement assignments completely determine the circuits' depth, width, and number of gates.

\begin{lemma}\label{realization}
Any Boolean function $f:\{0,1\}^n\rightarrow\{0,1\}$ with $\mathsf{deg}(f)>2$ can be computed with  accuracy $\epsilon \in [0,1]$ based on a valid measurement assignment using a circuit composed of the Clifford+T gate set with a depth of 
\begin{equation}
2*\log_2(k)  + 4*c*\Big(2+\log_2(k)+\log_2\Big(\frac{1}{\epsilon}\Big) \Big),
\end{equation}
\noindent a width of 
\begin{equation}
    k+ \sum_{S\subseteq [n]} \big(|S|-1\big )* \delta\big(\widehat{f}(S)\big) \ ,
\end{equation}
\noindent and a number of gates bounded by 
\begin{equation}
     \sum_{S\subseteq [n]} \big(|S|-1\big )* \delta\big(\widehat{f}(S)\big)+ k*\Big(2+4*c*\Big(2+\log_2(k)+\log_2\Big(\frac{1}{\epsilon}\Big) \Big)\Big)\ .
\end{equation}
\noindent W.l.o.g $\widehat{f}(S)$ are the Fourier coefficient and $k$ the sparsity of the multi-linear polynomial corresponding to the measurement assignment that will be realized. 
\end{lemma}
\begin{proof}
    Follows from the arguments above.
\end{proof}

\noindent It is essential to mention that we used the GHZ construction of Subsection \ref{ghzcreation} , with logarithmic depth. If constant quantum depth methods were used, a slight difference in the depth and the number of gates will be introduced. Nevertheless, the order of magnitude of the depth would not change. Moreover, the depth necessary to approximate the measurement operators can, in principle, be reduced with other circuit synthesis techniques \cite{Bocharov15,Bocharov2015}. However, the width of the circuits would increase with the use of auxiliary qubits. In addition, we recall that the proposed solution does not guarantee the optimal realization of a given measurement assignment concerning each one of the metrics for the circuits.

Although there is a large variability in the possible realizations of a measurement assignment, we consider that the particular translations presented unravel essential details and questions about the \textnormal{NMQC}$_\oplus$ model. For instance, one could question if there is another gate set that guarantees the exact synthesis of the measurement operators. That existence would allow deterministic computations for Boolean functions with an arbitrary degree. However, in \cite{Gary20}, various gate sets were considered, and none had a particular advantage, nor could they exactly synthesize arbitrary $Z$-rotations. This suggests that computing Boolean functions with an arbitrary degree could be a general limitation of the \textnormal{NMQC}$_\oplus$ model that only adaptivity can bypass.

\begin{figure}[H]
\centering
\includegraphics[scale=0.40]{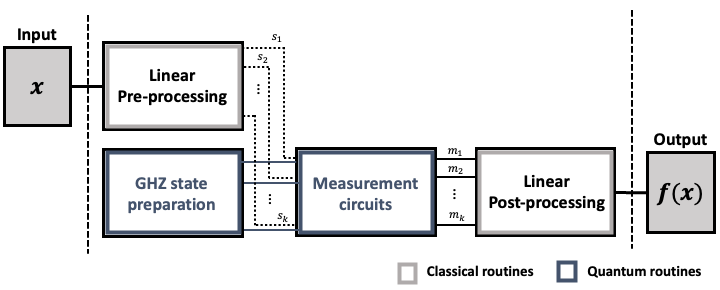}
\caption{A schematic representation of the circuit realization proposed for the \textnormal{NMQC}$_\oplus$ model. Vertically aligned processes are performed in parallel, while horizontally aligned processes are performed sequentially.}
\label{fig:arquiteture}
\end{figure}

\subsection{Quantum vs. classical circuits}

In this last subsection, we compare the known upper and lower bounds on classical circuits for Boolean function evaluation with quantum circuits evaluating the same functions in the \textnormal{NMQC}$_\oplus$ model. In particular, this comparison will focus on symmetric Boolean functions. This choice is justified by the contributions to the case of symmetric Boolean functions discussed in this paper and the many results on symmetric functions that can be found in the classical circuit complexity literature \cite{Wegener87,Vollmer10,Smolensky87,radhakrishnan1991,Stockmeyer76}.

\subsubsection{Degree two symmetric Boolean functions}

The translation of \textnormal{NMQC}$_\oplus$ computations from an abstract description to explicit circuits has shown that degree two Boolean functions are particularly interesting as they evaluate deterministically. Thus, their circuits' depth and number of gates 
do not depend on a prescribed level of approximation accuracy. Furthermore, we were fortunate to observe that some of these functions stand on the edge of the best lower bounds known for their classical circuit size \cite{Fischer82}. In particular, this allows us to notice that there is at least one family of symmetric Boolean functions of degree two with a logarithmic advantage over any conceivable classical circuit.   The advantage exists in terms of width and number of gates, considering that the classical circuits can be formed by Boolean operations from the full basis $\{\mathsf{AND}, \mathsf{OR},\neg, \mathsf{XOR},\mathsf{NAND}, etc. \}$. 

\begin{theorem}\label{separation}
Any classical circuit with unary and binary Boolean operators with single fan out computes the complete symmetric function of degree two $C^2$ with $\Theta(n*\log_2(n))$ gates and circuit width. In contrast, realizations of the \textnormal{NMQC}$_\oplus$ model compute this function with $\Theta(n)$ gates and circuit width.
\end{theorem}

\begin{proof}
It will be necessary to show that the upper and lower bounds have the same scaling rate for both functions in both circuit classes. Considering first the classical circuits, we can relate the $C^2$ function with the counting functions described in \cite{Fischer82}. This is possible because these functions are defined as $\mathsf{Count}(x,k)=1\ $if $\sum_{i=0}^n x_i\ mod\ k=0$. Therefore  
\begin{equation}\label{countdecomp}
    \mathsf{Count}(x,4)= \neg \big(C^2(x)\big )\ \mathsf{AND}\ \neg \big( C^1(x)\big)
\end{equation}
\noindent provided that $C^2(x)$ and $C^1(x)$ represent the two least significant bits of the Hamming weight $|x|$. Then it is possible to relate the length of the Boolean formulas of $\mathsf{Count}(x,4)$ and $C^1(x)$ with the length of the formulas for $C^2(x)$,    
\begin{align}
L(\mathsf{Count}(x,4))&\leq L(C^2(x)) + L(C^1(x)) \\
L(\mathsf{Count}(x,4))- L(C^1(x))&\leq L(C^2(x)) \ .
\end{align}

\noindent From these relations between the length of the formula and the lower bound of $ \varsigma n \log_2(\frac{n}{4})$ with $\varsigma \in [0,1]$ for $L(\mathsf{Count}(x,4))$ and the exact scaling of $n$ for $L(C^1(x))$, we obtain that $L(C^2)\geq \varsigma n \big(\log_2(\frac{n}{4})-1\big)$ \cite{Fischer82}. Thus, $\Omega\big(L(C^2(x)\big)= n*log(n)$.  

The upper bound of $L(C^2(x))$ is simply obtained by considering the following expression 
\begin{equation}
C^2(x)= \bigoplus_{i=1}^{log_2(n)} \Big( \bigwedge_{j=0}^{2^i-1} \Big(\bigoplus_{k=j*\frac{n}{2^i}}^{(j+1)\frac{n}{2^i}} x_{k}\Big) \Big)\ ,
\end{equation}

\noindent from which we obtain that $L(C^2(x)) \leq \mathcal{O}\big(n*\log_2(n)\big )$. Consequently, we conclude that $L(C^2(x)= \Theta\big(n*\log_2(n)\big)$. Finally, if circuits are composed of unary and binary operators with a single fan-out, the number of gates equals the Boolean formula's length \cite{Wegener87}. In addition, the circuit width also scales with the size of the formula and the depth with $\log_2(L(C^2(x))$.

For the \textnormal{NMQC}$_\oplus$ circuit realization of the $C^2(x)$ function, one can consider the multi-linear polynomials from Lemma \ref{polinomials}. These have a sparsity of $n+1$ and have been proven optimal in \cite{Hoban2011a}. Therefore, the process described in section \ref{NMQCcircuit} translates these to quantum circuits with a number of gates and width of $\Theta(n)$. Also, the computations are deterministic, given that the measurement operators can be synthesized exactly for Boolean functions with a degree equal to two. 
\end{proof}

\noindent This separation is based on a single family of symmetric Boolean functions. However, this can be enlarged to all the functions in the $[aC^k]$ class for which the symmetrized function is $C^2$. The explanation is that all these functions will have an \textnormal{NMQC}$_\oplus$ realization with a number of gates scaling with $n$ based on Lemma \ref{acbound}. On the other hand, classical circuits for those functions scale with $n*\log_2(n)$ in width and number of gates.

Moreover, suppose one argues that this circuit separation is still not strong enough because it misses to account for the necessary repetitions of the input strings. For that matter, we can show that the process that generates these repetitions for the \textnormal{NMQC}$_\oplus$ circuits will always be more efficient. In particular, it is sufficient to verify that the classical circuits require more copies of the input string than the \textnormal{NMQC}$_\oplus$ circuits. The classical circuits require a number of copies scaling with $\log_2(n)$, while the quantum solution requires only a single copy of the input string. This single copy is enough because it can be used to compute the parity while the given input string provides all the single bits that control the rotations necessary to evaluate the function (Lemma \ref{expr12}). Therefore, if the cost related to this process was included, the gap between the quantum and the classical circuits would further increase.

Another possible counter-argument relates to the depth of the circuits. We obtained a quantum circuit family for the specific $C^2$ function whose depth scales with $2\log_2(n)$, while the classical circuit scales only with $\log_2(n)+\log_2(\log_2(n))$. Consequently, they scale with the same asymptotical rate of $\mathcal{O}(\log(n))$. Although one could argue that there will be an asymptotical constant advantage equal to two in the classical case, this constant depth advantage is insignificant compared to the logarithmic width and gate advantage, which grows infinitely. Also, one could exclude the GHZ state preparation from the evaluation process and address it as polynomial size advice. This assumption is justified by state preparation is independent of the input. In that case, the $C^2$ function could be evaluated with a depth of $\log_2(n)+7$, producing a small additive advantage of $\log_2(\log_2(n))$ in comparison to the optimal classical circuits for which we do not know any equivalent advice.

\begin{lemma}
    Suppose an $n+1$-qubit GHZ state is provided for an \textnormal{NMQC}$_\oplus$ computation. Then, the complete symmetric function of degree two $C^2$ can be evaluated with a depth equal to $\log_2(n)+7$. 
\end{lemma}
\begin{proof}
Considering the measurement assignment that results from the multi-linear polynomial of Lemma \ref{polinomials}, there are $n$ measurements that depend directly on the input bits. These can be measured and their output linearly summed while the linear side computer determines $\oplus_{i=0}^{n-1} x_i$. When the linear side computer finishes, the last qubit is measured in the determined basis, and the outcome will be included in the linear post-processing. Therefore, the depth of the computations will be equal to or smaller than the depth of operations associated with the measurements, plus the linear post-processing. The measurement process for degree two Boolean function can be executed in 6 layers using only gates of the Clifford+$T$ gate set, given that the control Pauli $X$ and $Y$ gates can be synthesized within that length. Thus, this value is summed to the linear post-processing, which has a depth bounded by $\log_2(n)+1$, obtaining a total depth bounded by $\log_2(n)+7$.
\end{proof}

   We believe this is an interesting quantum-classical separation result because 
   it does not depend on any conjecture. Also, it is proven at the level of circuits \footnote{Boolean formulas are a subclass of circuits because they lack memory. Therefore, they are believed to be weaker than circuits that could reuse sub-computations. However, it is still open if the class of poly-size circuits is larger than the class of poly-size Boolean formulas \cite{Boaz09}. }, which is typically very difficult to do based on the difficulty of proving classical or quantum circuit lower bounds \cite{Scott22}. A related and stronger finding was the breakthrough result of \cite{Bravyi188} that identified a class of search problems with a more significant quantum-classical separation on the circuit level. Other separations at the circuit level were found in cases where the computational space was reduced \cite{ABLAYEV2005,Maslov21,BARRINGTON1989}. These demonstrated a width advantage for quantum circuits at a small and fixed size. In comparison, our result does seem to indicate that this width advantage could scale for computations without limited space. Other separations of larger magnitude have been found \cite{Arkhipov11,Shor1999,Grover98,simons97}. However, most of these are proven under the assumption of black boxes or query complexity \cite{Scott22}. Thus, they are not proven under the more physically realistic assumptions we consider here.

Furthermore, 
we observe that Theorem \ref{separation} implies that there exists a dimension of the input string above which the \textnormal{NMQC}$_\oplus$ evaluation of the $C^2$ functions starts to be more efficient than any classical circuit evaluation based on Boolean formulas. This point is particularly interesting because one can relate this \textnormal{NMQC}$_\oplus$ evaluation to the generalized Svetlichny inequalities \footnote{The connection of \textnormal{NMQC}$_\oplus$ evaluations and Bell inequalities was established in a one-to-one relation by the authors of \cite{Hoban2011a}.}. Therefore, by doing this, we conclude that starting from some specific number of qubits, a maximal violation of the generalized Svetlichny inequalities, with a quantum state and measurement operators that can be prepared with a linear number of operations, simultaneously proves a circuit separation. This observation connects separations related to communication problems to circuit problems. Also, the magnitude of the separations is intriguing as one is exponential, and the other is only logarithmic for the function at hand.

\subsubsection{ Higher degree symmetric Boolean functions}
To analyze higher degree symmetric Boolean functions, we isolate first the functions with degree three. This differentiation occurs because we identified a family of symmetric Boolean functions that presents a logarithmic advantage between the size of their measurement assignments and their optimal Boolean formulas. In particular, this family of Boolean functions refers to the $\mathsf{Count}(x,4)$ function described in the proof of Theorem \ref{separation}. It is possible to find measurement assignments for these functions with sparsity equal to $2n+2$, based on Equation \eqref{countdecomp}, Lemma \ref{polinomials} and Lemma \ref{sSBF}. In addition, their Boolean formulas are lower bounded by $\varsigma n\log_2(n)$, with $\varsigma \in [0,1]$. 

Although we identified a separation at an abstract level, it almost disappears when these measurement assignments are translated to circuits. The depth of the circuit has an additional term for functions with a degree equal to or greater than three (Lemma \ref{realization}). This implies that the depth of the resulting circuits is at least twice the classical depth. Also, the number of gates will scale as $n*\log_2(n)$, similarly to what happens in classical circuits. Therefore, the remaining advantage presents itself with the width of the quantum circuits being logarithmically smaller. However,  the computation has an error term that does not exist in the classical case. A fair comparison would require the consideration of probabilistic classical circuits. We leave the analysis between the probabilistic case and circuits with different depth and width relations for future work.

For higher degree functions, we use the bounds of the Boolean formulas identified for the Counting functions, with the modular division fixed to a power of two, to define upper bounds for the complete symmetric functions.

\begin{lemma}
The Boolean formula of any CSF $L\big(C^k\big )$ is bounded by $\mathcal{O}\big(n\log_2(n)^{ \floor*{\log_2(k)}}\big)$. 
\end{lemma}
\begin{proof}
First, the relation between the Counting functions and the CSFs for powers of two can be established by generalizing Equation \eqref{countdecomp} to $\mathsf{Count}(x,2^p)= \bigwedge_{i=1}^{p-1} \ \neg C^{2^i}$. Then, this allows us to relate Boolean formulas such that 
\begin{align}
L(C^{2^{p-1}})&\leq L\big(Count(x,2^p)\big)+ L\Big(\bigwedge_{i=1}^{p-2} \ \neg C^{2^i} \Big) \\
&\leq L \big(Count(x,2^p)\big)+ \sum_{i=1}^{p-2} L\big(C^{2^i}\big) 
\end{align}
\noindent Then, one can establish the upper bounds recursively for all CSF of the form $C^{2^p}$ by using these relations, starting from the fact that $L(C^1)=n$, and sequentially use the upper bounds of $\mathsf{Count}(n,2^p)$. In particular, the upper bounds of such Counting functions are $L(\mathsf{Count}(n,2^p))=\mathcal{O}\big(n*\log_2(n)^{p-1}\big)$. Therefore, we obtain $L\big(C^{2^p}\big)=\mathcal{O}\big(n*\log_2(n)^p\big )$ as the upper bound for the corresponding Boolean formulas. Finally, using Lemma \ref{decomposition}, one can generalize the bound to arbitrary degrees of the CSFs given that
\begin{align}
L\big(C^k\big )\leq  \sum_{r\in R_k} L\big (C^{2^r} \big ), 
\end{align}
\noindent with $R_k \subseteq [\ceil*{\log(k)}]$, such that $\sum_{r\in R_k} 2^{r}=k$. This relation demonstrates that the largest CSFs present in its decomposition into powers of two will define the upper bound, combining these results with the bounds established for those values, we conclude the proof. 
\end{proof}

 In comparison, our lower bound conjecture indicates that no symmetric measurement assignment in the \textnormal{NMQC}$_\oplus$ model can compute the CSF $C^k$ with less than $n^{\floor*{k/2}-1}$ qubits (Conjecture \ref{conjecture}). Therefore, unless our conjectured lower bound deviates immensely when non-symmetric measurement assignments are considered, the \textnormal{NMQC}$_\oplus$ computations have a polynomial disadvantage compared to the classical circuits. In particular, these would materialize in constant depth, polynomial width, and polynomial gate advantages for the classical circuits.

Moreover, given the lower bounds known from \cite{Hoban2011a, Mori2019}, and the previous analysis, the \textnormal{NMQC}$_\oplus$ model is suspected to be computationally inefficient for higher degree Boolean functions. Actually, it cannot even evaluate those functions deterministically with the Clifford+$T$ gate set (Corollary \ref{cliffordlimit}). These results point to the importance of adaptivity for deterministic computations for an MBQC model. Remember that if an adaptive process is allowed, then a circuit composed of individual \textnormal{NMQC}$_\oplus$ computations can replicate any classical and deterministic circuit \cite{AndersBrown2009}. Similarly, the results of \cite{Danos06} and \cite{Browne_2007} constructively describe the process of reaching determinism via adaptive measurement patterns for universal quantum computations.

\section{Discussion and conclusions} \label{conclusion}

We investigated the problem of evaluating deterministically Boolean functions in the \textnormal{NMQC}$_\oplus$ model, providing new constructions and a more precise understanding of the necessary instructions for each stage of the model. We characterized the complexity of two resources in this model: the number of required qubits in a GHZ state, and the required level of the Clifford hierarchy for single-qubit measurements. Regarding the number of qubits in a GHZ state, we lowered the upper bound for the number of qubits required to evaluate the entire set of symmetric Boolean functions and conjectured a lower bound for this same set, assuming symmetries between the instructions of the process. Regarding the complexity of the measurements required, we proved an upper bound on the level of the Clifford hierarchy. In particular,  this bound is proven to be tight to the boundary identified in \cite{Frembs22}, which limits the degree of the Boolean functions that a set of operators can evaluate. In the end, we translated the examined \textnormal{NMQC}$_\oplus$ evaluations to possible circuit realizations. This translation motivated our main result demonstrating a circuit separation for a family of degree two SBFs. Moreover, we analyzed possible separations for higher-degree functions. Although these functions have not shown any prospect for advantage under the \textnormal{NMQC}$_\oplus$ model, they guide the way to other computational models with potential advantages. Also, they can still solve non-local games for which classical circuit analogues fail. Therefore, their computation using the \textnormal{NMQC}$_\oplus$ model
could still be an exciting element for the study of non-locality and contextuality.

Furthermore, we highlight that proving the conjectured lower bound for the number of qubits necessary to evaluate an SBF on the \textnormal{NMQC}$_\oplus$ model would be an exciting result. It directly affects the understanding of the model, and it would be an interesting result for understanding non-local correlations. We state this because each Boolean function relates to a specific Bell inequality \cite{Hoban2011a}. Thus, results on one front relate to the other, and optimal violations of Bell inequalities are very difficult to define for arbitrary system sizes \cite{BrunnerNicolas2014}. In addition, such a result would be valuable to test the conjecture presented by \cite{Pawlowski09,Popescu94,Barret05}, which attempts to axiomatize quantum mechanics from information theoretical assumptions. For instance, in \cite{Brassard06,VanDam13} the authors define axioms based on communication complexity problems. In particular, their axiom states that correlations that would allow our physical universe to solve these problems trivially are forbidden. Then, if this particular hypothesis is correct, one should retrieve quantum correlations from the maximum allowed by the non-triviality constraint. Interestingly, in \cite{Reznik08}, the $C^2$ was used as a non-local black box for which the authors recovered the maximal allowed efficiency at a constant distance from the corresponding maximal quantum efficiency. Therefore, introducing lower bounds for other SBFs would provide more tests for this concept, given that one can analyze any non-local box and compare the respective maximal quantum correlations to their maximal efficiency.

We showed that a circuit advantage in the \textnormal{NMQC}$_\oplus$ model is possible for evaluating certain Boolean functions with a degree equal to two. However, these computations were proven inefficient for particular higher degree functions compared to the best-known classical circuits. This limitation motivates the extension of the algebraic tools used for the non-adaptive case to temporally-structured computations. In particular, one could potentially connect this computational framework to other circuit synthesis techniques \cite{Mosca19}, classical algebraic techniques \cite{burgisser2013algebraic}, or some existing quantum signal processing techniques \cite{Low17,Haah2019}. The possible results from using these new techniques could be the discovery of compact quantum circuits with potential advantages in both deterministic and probabilistic settings. Also, they could describe possible trade-offs between space and computational time.

For both extreme cases of minimal depth or width, it would be interesting to understand if circuits exhibit some quantum advantage. For instance, for the depth necessary to evaluate the Boolean $\mathsf{OR}$ function defined on any possible size of the input string, one could expect an advantage \cite{scott21,Huang19,Ambainis17}. In particular, in \cite{Hoyer02}, the use of a single adaptive step demonstrated an exponential reduction of the number of qubits for the same function. The same separation was then generalized equally with a single adaptive step for arbitrary symmetric Boolean functions in \cite{Mori2019} and with a constant number of adaptive steps in \cite{Austin22} resorting to the previously mentioned quantum signal processing techniques. Identically, suppose one does try to minimize the computational space. In that case, the quantum circuits could have some width advantage. Interestingly, for a small computational space, a separation was already identified in \cite{ABLAYEV2005,Maslov21,BARRINGTON1989}, and experimentally explored on IBM's quantum devices \cite{Maslov21}. Nevertheless, we question whether the width separations for quantum and classical computations hold for computations with unbounded width and arbitrary input size. Our results seem to indicate that this could indeed be the case.

Another type of analysis would be focused on an intermediate class of quantum circuits. For that, we consider the circuits class composed of quantum layers that can only pass classical information between them. Remember that for MBQC computations, the various quantum layers are entangled. Therefore, there is a quantum and a classical channel between each layer \footnote{Some initial ideas on how to describe a quantum wire that transmits information between layers have been proposed in \cite{Paul22}.}. In addition, the computational capacity of this intermediate class is not entirely understood \cite{Mori2019,anand2022power}. However, if they present any advantage, they are easier to implement in Noisy, intermediate-scale quantum devices (NISQ), being therefore attractive candidates for experimental realizations \cite{Preskill2018}. That reduction in the hardware requirements relates to the fact that these circuits demand less connectivity between the qubits and coherence of the gates. In the end, describing such intermediate complexity classes could be a very appealing task to understand further the importance of a temporal structure in general quantum computations and other particular circuits with advantages \cite{Bravyi188, Courdron18}.

Finally, we hope this work will inspire more experiments like the ones described in \cite{Demirel2021, Austin22, Swain19}, to verify the computational capabilities of quantum correlations with realizations equivalent to the \textnormal{NMQC}$_\oplus$ model. Moreover, as other works have shown, these realizations could work as a device benchmarking tool \cite{Bo22,Baccari20,Wei20,Huang20,Sheffer22}. Independently, the ideas explored in this paper are also connected to some secure multi-party computation protocols \cite{Dunjko2016}, which consider clients with limited computational power. Therefore, we believe that the techniques developed here could allow the design of new and better protocols. These could then potentially motivate exciting experiments, such as the ones reported in  \cite{Barz15,Clementi17}.

\vspace{0.5cm}
\noindent\textbf{Acknowledgments}
We acknowledge helpful feedback from Dan Browne, Pedro Patricio, and Matty Hoban on an early version of these results. This work is financed by National Funds through the FCT - Fundação para a Ciência e a Tecnologia, I.P. (Portuguese Foundation for Science and Technology) within the project IBEX, with reference PTDC/CCI-COM/4280/2021, and via CEECINST/00062/2018 (EFG). This work was supported by the H2020-FETOPEN Grant PHOQUSING (GA no.:899544).

\bibliographystyle{quantum}
\bibliography{library}

\begin{thebibliography}{100}

\bibitem{DevonAaronson21}
Scott Aaronson, DeVon Ingram, and William Kretschmer.
\newblock ``{The Acrobatics of BQP}''.
\newblock In Shachar Lovett, editor, 37th Computational Complexity Conference (CCC 2022).
\newblock \href{https://dx.doi.org/10.4230/LIPIcs.CCC.2022.20}{Volume 234 of Leibniz International Proceedings in Informatics (LIPIcs), pages 20:1--20:17}.
\newblock Dagstuhl, Germany~(2022). Schloss Dagstuhl -- Leibniz-Zentrum f{\"u}r Informatik.

\bibitem{Jozsa03}
Richard Jozsa and Noah Linden.
\newblock ``On the role of entanglement in quantum-computational speed-up''.
\newblock \href{https://dx.doi.org/https://doi.org/10.1098/rspa.2002.1097}{Proceedings of the Royal Society of London. Series A: Mathematical, Physical and Engineering Sciences {\bf 459}, 2011--2032}~(2003).

\bibitem{Howard2014}
Mark Howard, Joel Wallman, Victor Veitch, and Joseph Emerson.
\newblock ``Contextuality supplies the ‘magic' for quantum computation''.
\newblock \href{https://dx.doi.org/https://doi.org/10.1038/nature13460}{Nature {\bf 510}, 351--355}~(2014).

\bibitem{Bermejo17}
Juan Bermejo-Vega, Nicolas Delfosse, Dan~E. Browne, Cihan Okay, and Robert Raussendorf.
\newblock ``Contextuality as a resource for models of quantum computation with qubits''.
\newblock \href{https://dx.doi.org/10.1103/PhysRevLett.119.120505}{Phys. Rev. Lett. {\bf 119}, 120505}~(2017).

\bibitem{Galvao05}
Ernesto~F. Galv\~ao.
\newblock ``Discrete wigner functions and quantum computational speedup''.
\newblock \href{https://dx.doi.org/10.1103/PhysRevA.71.042302}{Phys. Rev. A {\bf 71}, 042302}~(2005).

\bibitem{man12}
A.~Mari and J.~Eisert.
\newblock ``Positive wigner functions render classical simulation of quantum computation efficient''.
\newblock \href{https://dx.doi.org/10.1103/PhysRevLett.109.230503}{Phys. Rev. Lett. {\bf 109}, 230503}~(2012).

\bibitem{Grover98}
Lov~K. Grover.
\newblock ``The advantages of superposition''.
\newblock \href{https://dx.doi.org/10.1126/science.280.5361.228}{Science {\bf 280}, 228--228}~(1998).

\bibitem{Raussendorf01}
Robert Raussendorf and Hans~J. Briegel.
\newblock ``A one-way quantum computer''.
\newblock \href{https://dx.doi.org/10.1103/PhysRevLett.86.5188}{Phys. Rev. Lett. {\bf 86}, 5188--5191}~(2001).

\bibitem{UniversalRes06}
Maarten Van~den Nest, Akimasa Miyake, Wolfgang D\"ur, and Hans~J. Briegel.
\newblock ``Universal resources for measurement-based quantum computation''.
\newblock \href{https://dx.doi.org/10.1103/PhysRevLett.97.150504}{Phys. Rev. Lett. {\bf 97}, 150504}~(2006).

\bibitem{AndersBrown2009}
Janet Anders and Dan~E. Browne.
\newblock ``Computational power of correlations''.
\newblock \href{https://dx.doi.org/10.1103/PhysRevLett.102.050502}{Phys. Rev. Lett. {\bf 102}, 050502}~(2009).

\bibitem{Danos06}
Vincent Danos and Elham Kashefi.
\newblock ``{Determinism in the one-way model}''.
\newblock \href{https://dx.doi.org/10.1103/PhysRevA.74.052310}{Phys. Rev. A {\bf 74}, 052310}~(2006).

\bibitem{Browne_2007}
Daniel~E Browne, Elham Kashefi, Mehdi Mhalla, and Simon Perdrix.
\newblock ``{Generalized flow and determinism in measurement-based quantum computation}''.
\newblock \href{https://dx.doi.org/10.1088/1367-2630/9/8/250}{New Journal of Physics {\bf 9}, 250}~(2007).

\bibitem{Bremner16}
Michael~J Bremner, Ashley Montanaro, and Dan~J Shepherd.
\newblock ``{Average-Case Complexity Versus Approximate Simulation of Commuting Quantum Computations}''.
\newblock \href{https://dx.doi.org/10.1103/PhysRevLett.117.080501}{Phys. Rev. Lett. {\bf 117}, 080501}~(2016).

\bibitem{Matty14}
Matty~J. Hoban, Joel~J. Wallman, Hussain Anwar, Na\"{\i}ri Usher, Robert Raussendorf, and Dan~E. Browne.
\newblock ``Measurement-based classical computation''.
\newblock \href{https://dx.doi.org/10.1103/PhysRevLett.112.140505}{Phys. Rev. Lett. {\bf 112}, 140505}~(2014).

\bibitem{Bremner2016}
Michael~J. Bremner, Ashley Montanaro, and Dan~J. Shepherd.
\newblock ``Achieving quantum supremacy with sparse and noisy commuting quantum computations''.
\newblock \href{https://dx.doi.org/10.22331/q-2017-04-25-8}{{Quantum} {\bf 1}, 8}~(2017).

\bibitem{Novo2021quantumadvantage}
Leonardo Novo, Juani Bermejo-Vega, and Ra{\'{u}}l Garc{\'{i}}a-Patr{\'{o}}n.
\newblock ``{Quantum advantage from energy measurements of many-body quantum systems}''.
\newblock \href{https://dx.doi.org/10.22331/q-2021-06-02-465}{Quantum {\bf 5}, 465}~(2021).

\bibitem{Hayashi_2019}
Masahito Hayashi and Yuki Takeuchi.
\newblock ``{Verifying commuting quantum computations via fidelity estimation of weighted graph states}''.
\newblock \href{https://dx.doi.org/10.1088/1367-2630/ab3d88}{New Journal of Physics {\bf 21}, 93060}~(2019).

\bibitem{Vega18}
Juan Bermejo-Vega, Dominik Hangleiter, Martin Schwarz, Robert Raussendorf, and Jens Eisert.
\newblock ``{Architectures for Quantum Simulation Showing a Quantum Speedup}''.
\newblock \href{https://dx.doi.org/10.1103/PhysRevX.8.021010}{Phys. Rev. X {\bf 8}, 021010}~(2018).

\bibitem{Miller17}
Jacob Miller, Stephen Sanders, and Akimasa Miyake.
\newblock ``Quantum supremacy in constant-time measurement-based computation: A unified architecture for sampling and verification''.
\newblock \href{https://dx.doi.org/10.1103/PhysRevA.96.062320}{Phys. Rev. A {\bf 96}, 062320}~(2017).

\bibitem{Hoban2011a}
Matty~J Hoban, Earl~T Campbell, Klearchos Loukopoulos, and Dan~E Browne.
\newblock ``{Non-adaptive measurement-based quantum computation and multi-party Bell inequalities}''.
\newblock \href{https://dx.doi.org/10.1088/1367-2630/13/2/023014}{New Journal of Physics {\bf 13}, 23014}~(2011).

\bibitem{Mori2019}
Ryuhei Mori.
\newblock ``{Periodic Fourier representation of Boolean functions}''.
\newblock Quantum Info. Comput. {\bf 19}, 392--412~(2019).
\newblock  url:~\url{https://dl.acm.org/doi/abs/10.5555/3370251.3370253}.

\bibitem{Frembs22}
Markus Frembs, Sam Roberts, Earl~T Campbell, and Stephen~D Bartlett.
\newblock ``Hierarchies of resources for measurement-based quantum computation''.
\newblock \href{https://dx.doi.org/10.1088/1367-2630/acaee2}{New Journal of Physics {\bf 25}, 013002}~(2023).

\bibitem{Mackeprang22}
Jelena Mackeprang, Daniel Bhatti, Matty~J Hoban, and Stefanie Barz.
\newblock ``The power of qutrits for non-adaptive measurement-based quantum computing''.
\newblock \href{https://dx.doi.org/10.1088/1367-2630/acdf77}{New Journal of Physics {\bf 25}, 073007}~(2023).

\bibitem{Collins02}
Daniel Collins, Nicolas Gisin, Sandu Popescu, David Roberts, and Valerio Scarani.
\newblock ``{Bell-Type Inequalities to Detect True $\mathit{n}$-Body Nonseparability}''.
\newblock \href{https://dx.doi.org/10.1103/PhysRevLett.88.170405}{Phys. Rev. Lett. {\bf 88}, 170405}~(2002).

\bibitem{BrunnerNicolas2014}
Nicolas Brunner, Daniel Cavalcanti, Stefano Pironio, Valerio Scarani, and Stephanie Wehner.
\newblock ``{Bell nonlocality}''.
\newblock \href{https://dx.doi.org/10.1103/RevModPhys.86.419}{Rev. Mod. Phys. {\bf 86}, 419--478}~(2014).

\bibitem{kravvcenko2013quantum}
Dmitrijs Krav{\v{c}}enko.
\newblock ``{Quantum Games, Quantum States, Their Properties and Applications}''.
\newblock PhD thesis.
\newblock Latvijas Universit{\=a}te.
\newblock ~(2013).

\bibitem{Slofstra11}
William Slofstra.
\newblock ``{Lower bounds on the entanglement needed to play XOR non-local games}''.
\newblock \href{https://dx.doi.org/10.1063/1.3652924}{Journal of Mathematical Physics {\bf 52}, 102202}~(2011).

\bibitem{ambainis2012advantage}
Andris Ambainis, J{\={a}}nis Iraids, Dmitry Kravchenko, and Madars Virza.
\newblock ``Advantage of quantum strategies in random symmetric xor games''.
\newblock In Anton{\'i}n Ku{\v{c}}era, Thomas~A. Henzinger, Jaroslav Ne{\v{s}}et{\v{r}}il, Tom{\'a}{\v{s}} Vojnar, and David Anto{\v{s}}, editors, Mathematical and Engineering Methods in Computer Science.
\newblock \href{https://dx.doi.org/https://doi.org/10.1007/978-3-642-36046-6_7}{Pages 57--68}.
\newblock Berlin, Heidelberg~(2013). Springer Berlin Heidelberg.

\bibitem{ambainis2013provable}
Andris Ambainis and Janis Iraids.
\newblock ``{Provable Advantage for Quantum Strategies in Random Symmetric XOR Games}''.
\newblock In Simone Severini and Fernando Brandao, editors, 8th Conference on the Theory of Quantum Computation, Communication and Cryptography (TQC 2013).
\newblock \href{https://dx.doi.org/10.4230/LIPIcs.TQC.2013.146}{Volume~22 of Leibniz International Proceedings in Informatics (LIPIcs), pages 146--156}.
\newblock Dagstuhl, Germany~(2013). Schloss Dagstuhl -- Leibniz-Zentrum f{\"u}r Informatik.

\bibitem{Reznik08}
Samuel Marcovitch and Benni Reznik.
\newblock ``Implications of communication complexity in multipartite systems''.
\newblock \href{https://dx.doi.org/10.1103/PhysRevA.77.032120}{Phys. Rev. A {\bf 77}, 032120}~(2008).

\bibitem{Pawlowski09}
Marcin Paw{\l}owski, Tomasz Paterek, Dagomir Kaszlikowski, Valerio Scarani, Andreas Winter, and Marek {\.{Z}}ukowski.
\newblock ``{Information causality as a physical principle}''.
\newblock \href{https://dx.doi.org/10.1038/nature08400}{Nature {\bf 461}, 1101--1104}~(2009).

\bibitem{Popescu94}
Sandu Popescu and Daniel Rohrlich.
\newblock ``{Quantum nonlocality as an axiom}''.
\newblock \href{https://dx.doi.org/10.1007/BF02058098}{Foundations of Physics {\bf 24}, 379--385}~(1994).

\bibitem{Barret05}
Jonathan Barrett, Noah Linden, Serge Massar, Stefano Pironio, Sandu Popescu, and David Roberts.
\newblock ``{Nonlocal correlations as an information-theoretic resource}''.
\newblock \href{https://dx.doi.org/10.1103/PhysRevA.71.022101}{Phys. Rev. A {\bf 71}, 022101}~(2005).

\bibitem{razborov2003quantum}
A~A Razborov.
\newblock ``Quantum communication complexity of symmetric predicates''.
\newblock \href{https://dx.doi.org/https://dx.doi.org/10.1070/IM2003v067n01ABEH000422}{Izvestiya: Mathematics {\bf 67}, 145}~(2003).

\bibitem{zhang2009communication}
Zhiqiang Zhang and Yaoyun Shi.
\newblock ``{Communication complexities of symmetric XOR functions}''.
\newblock Quantum Information and Computation {\bf 9}, 255--263~(2009).
\newblock  url:~\url{https://dl.acm.org/doi/abs/10.5555/2011781.2011786}.

\bibitem{botteronnonlocal}
Pierre Botteron.
\newblock ``{NonLocal Boxes and Communication Complexity}''.
\newblock Master's thesis.
\newblock Université Paul Sabatier Toulouse III.
\newblock ~(2022).
\newblock  url:~\url{https://pierre-botteron.github.io/Articles/2022-06-MSc-Thesis.pdf}.

\bibitem{bae2018generalized}
Kwangil Bae and Wonmin Son.
\newblock ``Generalized nonlocality criteria under the correlation symmetry''.
\newblock \href{https://dx.doi.org/10.1103/PhysRevA.98.022116}{Phys. Rev. A {\bf 98}, 022116}~(2018).

\bibitem{Frembs_2018}
Markus Frembs, Sam Roberts, and Stephen~D Bartlett.
\newblock ``{Contextuality as a resource for measurement-based quantum computation beyond qubits}''.
\newblock \href{https://dx.doi.org/10.1088/1367-2630/aae3ad}{New Journal of Physics {\bf 20}, 103011}~(2018).

\bibitem{Bravyi188}
Sergey Bravyi, David Gosset, and Robert K{\"{o}}nig.
\newblock ``{Quantum advantage with shallow circuits}''.
\newblock \href{https://dx.doi.org/10.1126/science.aar3106}{Science {\bf 362}, 308--311}~(2018).

\bibitem{Grier20}
Daniel Grier and Luke Schaeffer.
\newblock ``{Interactive Shallow Clifford Circuits: Quantum Advantage against NC¹ and Beyond}''.
\newblock In Proceedings of the 52nd Annual ACM SIGACT Symposium on Theory of Computing.
\newblock \href{https://dx.doi.org/10.1145/3357713.3384332}{Pages 875--888}.
\newblock STOC 2020New York, NY, USA~(2020). Association for Computing Machinery.

\bibitem{Libor22}
Libor Caha, Xavier Coiteux-Roy, and Robert Koenig.
\newblock ``{Single-qubit gate teleportation provides a quantum advantage}''~(2022).
\newblock  \href{http://arxiv.org/abs/2209.14158}{arXiv:2209.14158}.

\bibitem{Gall19}
Fran{\c{c}}ois~Le Gall.
\newblock ``{Average-Case Quantum Advantage with Shallow Circuits}''.
\newblock In Amir Shpilka, editor, 34th Computational Complexity Conference (CCC 2019).
\newblock \href{https://dx.doi.org/10.4230/LIPIcs.CCC.2019.21}{Volume 137 of Leibniz International Proceedings in Informatics (LIPIcs), pages 21:1----21:20}.
\newblock Dagstuhl, Germany~(2019). Schloss Dagstuhl--Leibniz-Zentrum fuer Informatik.

\bibitem{Courdron18}
Matthew Coudron, Jalex Stark, and Thomas Vidick.
\newblock ``Trading locality for time: certifiable randomness from low-depth circuits''.
\newblock \href{https://dx.doi.org/https://doi.org/10.1007/s00220-021-03963-w}{Communications in mathematical physics {\bf 382}, 49--86}~(2021).

\bibitem{Bravyi2020}
Sergey Bravyi, David Gosset, Robert K{\"{o}}nig, and Marco Tomamichel.
\newblock ``{Quantum advantage with noisy shallow circuits}''.
\newblock \href{https://dx.doi.org/10.1038/s41567-020-0948-z}{Nature Physics {\bf 16}, 1040--1045}~(2020).

\bibitem{Atsuya21}
Atsuya Hasegawa and Fran\c{c}ois Le~Gall.
\newblock ``{Quantum Advantage with Shallow Circuits Under Arbitrary Corruption}''.
\newblock In Hee-Kap Ahn and Kunihiko Sadakane, editors, 32nd International Symposium on Algorithms and Computation (ISAAC 2021).
\newblock \href{https://dx.doi.org/10.4230/LIPIcs.ISAAC.2021.74}{Volume 212 of Leibniz International Proceedings in Informatics (LIPIcs), pages 74:1--74:16}.
\newblock Dagstuhl, Germany~(2021). Schloss Dagstuhl -- Leibniz-Zentrum f{\"u}r Informatik.

\bibitem{Watts19}
Adam~Bene Watts, Robin Kothari, Luke Schaeffer, and Avishay Tal.
\newblock ``{Exponential Separation between Shallow Quantum Circuits and Unbounded Fan-in Shallow Classical Circuits}''.
\newblock In Proceedings of the 51st Annual ACM SIGACT Symposium on Theory of Computing.
\newblock \href{https://dx.doi.org/10.1145/3313276.3316404}{Pages 515--526}.
\newblock STOC 2019New York, NY, USA~(2019). Association for Computing Machinery.

\bibitem{parham2022power}
Natalie Parham.
\newblock ``{On the Power and Limitations of Shallow Quantum Circuits}''.
\newblock Master's thesis.
\newblock University of Waterloo.
\newblock ~(2022).
\newblock  url:~\url{https://uwspace.uwaterloo.ca/handle/10012/18702}.

\bibitem{Maslov21}
Dmitri Maslov, Jin-Sung Kim, Sergey Bravyi, Theodore~J Yoder, and Sarah Sheldon.
\newblock ``{Quantum advantage for computations with limited space}''.
\newblock \href{https://dx.doi.org/10.1038/s41567-021-01271-7}{Nature Physics {\bf 17}, 894--897}~(2021).

\bibitem{ABLAYEV2005}
Farid Ablayev, Aida Gainutdinova, Marek Karpinski, Cristopher Moore, and Christopher Pollett.
\newblock ``On the computational power of probabilistic and quantum branching program''.
\newblock \href{https://dx.doi.org/https://doi.org/10.1016/j.ic.2005.04.003}{Information and Computation {\bf 203}, 145--162}~(2005).

\bibitem{Bremner09}
D~Shepherd and M.~J. Bremner.
\newblock ``{Temporally unstructured quantum computation}''.
\newblock \href{https://dx.doi.org/10.1098/rspa.2008.0443}{Proceedings of the Royal Society of London Series A {\bf 465}, 1413--1439}~(2009).

\bibitem{GHZ1989}
Daniel~M Greenberger, Michael~A Horne, and Anton Zeilinger.
\newblock ``{Going Beyond Bell's Theorem}''.
\newblock In Menas Kafatos, editor, Bell's Theorem, Quantum Theory and Conceptions of the Universe.
\newblock \href{https://dx.doi.org/10.1007/978-94-017-0849-4_10}{Pages 69--72}.
\newblock Dordrecht~(1989). Springer Netherlands.

\bibitem{Cruz19}
Diogo Cruz, Romain Fournier, Fabien Gremion, Alix Jeannerot, Kenichi Komagata, Tara Tosic, Jarla Thiesbrummel, Chun~Lam Chan, Nicolas Macris, Marc-Andr{\'{e}} Dupertuis, and Cl{\'{e}}ment Javerzac-Galy.
\newblock ``{Efficient Quantum Algorithms for GHZ and W States, and Implementation on the IBM Quantum Computer}''.
\newblock \href{https://dx.doi.org/https://doi.org/10.1002/qute.201900015}{Advanced Quantum Technologies {\bf 2}, 1900015}~(2019).

\bibitem{Wolf2001}
R.~F. Werner and M.~M. Wolf.
\newblock ``All-multipartite bell-correlation inequalities for two dichotomic observables per site''.
\newblock \href{https://dx.doi.org/10.1103/PhysRevA.64.032112}{Phys. Rev. A {\bf 64}, 032112}~(2001).

\bibitem{Odonnel2021}
Ryan O’Donnell.
\newblock ``Analysis of boolean functions''.
\newblock Cambridge University Press. ~(2014).
\newblock  url:~\url{http://www.cs.cmu.edu/~./odonnell/papers/Analysis-of-Boolean-Functions-by-Ryan-ODonnell.pdf}.

\bibitem{chistopolskaya2018parity}
Anastasiya {Chistopolskaya} and Vladimir~V. {Podolskii}.
\newblock ``{Parity Decision Tree Complexity is Greater Than Granularity}''~(2018).
\newblock  \href{http://arxiv.org/abs/1810.08668}{arXiv:1810.08668}.

\bibitem{Canteaut05}
A~Canteaut and M~Videau.
\newblock ``{Symmetric Boolean functions}''.
\newblock \href{https://dx.doi.org/10.1109/TIT.2005.851743}{IEEE Transactions on Information Theory {\bf 51}, 2791--2811}~(2005).

\bibitem{Stockmeyer76}
Larry~J Stockmeyer.
\newblock ``{On the combinational complexity of certain symmetric Boolean functions}''.
\newblock \href{https://dx.doi.org/10.1007/BF01683282}{Mathematical systems theory {\bf 10}, 323--336}~(1976).

\bibitem{Arnold63}
R~F Arnold and M~A Harrison.
\newblock ``{Algebraic Properties of Symmetric and Partially Symmetric Boolean Functions}''.
\newblock \href{https://dx.doi.org/10.1109/PGEC.1963.263535}{IEEE Transactions on Electronic Computers {\bf EC-12}, 244--251}~(1963).

\bibitem{Braeken05}
An~Braeken and Bart Preneel.
\newblock ``On the algebraic immunity of symmetric boolean functions''.
\newblock In Subhamoy Maitra, C.~E. Veni~Madhavan, and Ramarathnam Venkatesan, editors, Progress in Cryptology - INDOCRYPT 2005.
\newblock \href{https://dx.doi.org/https://doi.org/10.1007/11596219_4}{Volume 3797 of Lecture Notes in Computer Science, pages 35--48}.
\newblock Berlin, Heidelberg~(2005). Springer Berlin Heidelberg.

\bibitem{BUHRMAN200221}
Harry Buhrman and Ronald de~Wolf.
\newblock ``{Complexity measures and decision tree complexity: a survey}''.
\newblock \href{https://dx.doi.org/https://doi.org/10.1016/S0304-3975(01)00144-X}{Theoretical Computer Science {\bf 288}, 21--43}~(2002).

\bibitem{Matthew13}
Matthew Amy, Dmitri Maslov, Michele Mosca, and Martin Roetteler.
\newblock ``{A Meet-in-the-Middle Algorithm for Fast Synthesis of Depth-Optimal Quantum Circuits}''.
\newblock \href{https://dx.doi.org/10.1109/TCAD.2013.2244643}{IEEE Transactions on Computer-Aided Design of Integrated Circuits and Systems {\bf 32}, 818--830}~(2013).

\bibitem{Bullock06}
V~V Shende, S~S Bullock, and I~L Markov.
\newblock ``{Synthesis of quantum-logic circuits}''.
\newblock \href{https://dx.doi.org/10.1109/TCAD.2005.855930}{IEEE Transactions on Computer-Aided Design of Integrated Circuits and Systems {\bf 25}, 1000--1010}~(2006).

\bibitem{Juha04}
Juha~J Vartiainen, Mikko M{\"{o}}tt{\"{o}}nen, and Martti~M Salomaa.
\newblock ``{Efficient Decomposition of Quantum Gates}''.
\newblock \href{https://dx.doi.org/10.1103/PhysRevLett.92.177902}{Phys. Rev. Lett. {\bf 92}, 177902}~(2004).

\bibitem{Zeng08}
Bei Zeng, Xie Chen, and Isaac~L Chuang.
\newblock ``{Semi-Clifford operations, structure of $\mathcal{C}_{k}$ hierarchy, and gate complexity for fault-tolerant quantum computation}''.
\newblock \href{https://dx.doi.org/10.1103/PhysRevA.77.042313}{Phys. Rev. A {\bf 77}, 042313}~(2008).

\bibitem{Gary20}
Gary~J Mooney, Charles~D Hill, and Lloyd C~L Hollenberg.
\newblock ``{Cost-optimal single-qubit gate synthesis in the {C}lifford hierarchy}''.
\newblock \href{https://dx.doi.org/10.22331/q-2021-02-15-396}{Quantum {\bf 5}, 396}~(2021).

\bibitem{Nadish20}
Nadish de~Silva.
\newblock ``Efficient quantum gate teleportation in higher dimensions''.
\newblock \href{https://dx.doi.org/https://doi.org/10.1098/rspa.2020.0865}{Proceedings of the Royal Society A {\bf 477}, 20200865}~(2021).

\bibitem{GottesChuang99}
Daniel Gottesman and Isaac~L Chuang.
\newblock ``Demonstrating the viability of universal quantum computation using teleportation and single-qubit operations''.
\newblock \href{https://dx.doi.org/https://doi.org/10.1038/46503}{Nature {\bf 402}, 390--393}~(1999).

\bibitem{Gottesman98}
Daniel {Gottesman}.
\newblock ``{The Heisenberg Representation of Quantum Computers}''~(1998).
\newblock  \href{http://arxiv.org/abs/quant-ph/9807006}{arXiv:quant-ph/9807006}.

\bibitem{Moska13}
Vadym Kliuchnikov, Dmitri Maslov, and Michele Mosca.
\newblock ``Fast and efficient exact synthesis of single-qubit unitaries generated by clifford and t gates''.
\newblock Quantum Info. Comput. {\bf 13}, 607–630~(2013).
\newblock  url:~\url{https://dl.acm.org/doi/abs/10.5555/2535649.2535653}.

\bibitem{Brunner12}
Nicolas Brunner, James Sharam, and Tam\'as V\'ertesi.
\newblock ``Testing the structure of multipartite entanglement with bell inequalities''.
\newblock \href{https://dx.doi.org/10.1103/PhysRevLett.108.110501}{Phys. Rev. Lett. {\bf 108}, 110501}~(2012).

\bibitem{Kempe11}
Julia Kempe, Hirotada Kobayashi, Keiji Matsumoto, Ben Toner, and Thomas Vidick.
\newblock ``{Entangled Games Are Hard to Approximate}''.
\newblock \href{https://dx.doi.org/10.1137/090751293}{SIAM Journal on Computing {\bf 40}, 848--877}~(2011).

\bibitem{Quek22}
Yihui Quek, Eneet Kaur, and Mark~M. Wilde.
\newblock ``Multivariate trace estimation in constant quantum depth''.
\newblock \href{https://dx.doi.org/10.22331/q-2024-01-10-1220}{{Quantum} {\bf 8}, 1220}~(2024).

\bibitem{Selinger15}
Peter Selinger.
\newblock ``{Efficient Clifford+T Approximation of Single-Qubit Operators}''.
\newblock \href{https://dx.doi.org/https://dl.acm.org/doi/abs/10.5555/2685188.2685198}{Quantum Info. Comput. {\bf 15}, 159--180}~(2015).

\bibitem{mosca16}
Vadym Kliuchnikov, Dmitri Maslov, and Michele Mosca.
\newblock ``{Practical Approximation of Single-Qubit Unitaries by Single-Qubit Quantum Clifford and T Circuits}''.
\newblock \href{https://dx.doi.org/10.1109/TC.2015.2409842}{IEEE Transactions on Computers {\bf 65}, 161--172}~(2016).

\bibitem{Neil15}
Neil~J Ross.
\newblock ``{Optimal Ancilla-Free CLIFFORD+V Approximation of Z-Rotations}''.
\newblock Quantum Info. Comput. {\bf 15}, 932--950~(2015).
\newblock  url:~\url{https://dl.acm.org/doi/abs/10.5555/2871350.2871354}.

\bibitem{Vazirani93}
Ethan Bernstein and Umesh Vazirani.
\newblock ``{Quantum Complexity Theory}''.
\newblock In Proceedings of the Twenty-Fifth Annual ACM Symposium on Theory of Computing.
\newblock \href{https://dx.doi.org/10.1145/167088.167097}{Pages 11--20}.
\newblock STOC '93New York, NY, USA~(1993). Association for Computing Machinery.

\bibitem{Bocharov15}
Alex Bocharov, Martin Roetteler, and Krysta~M Svore.
\newblock ``{Efficient synthesis of probabilistic quantum circuits with fallback}''.
\newblock \href{https://dx.doi.org/10.1103/PhysRevA.91.052317}{Phys. Rev. A {\bf 91}, 052317}~(2015).

\bibitem{Bocharov2015}
Alex Bocharov, Martin Roetteler, and Krysta~M Svore.
\newblock ``{Efficient Synthesis of Universal Repeat-Until-Success Quantum Circuits}''.
\newblock \href{https://dx.doi.org/10.1103/PhysRevLett.114.080502}{Phys. Rev. Lett. {\bf 114}, 080502}~(2015).

\bibitem{Wegener87}
Ingo Wegener.
\newblock ``{The Complexity of Boolean Functions}''.
\newblock \href{https://dx.doi.org/https://dl.acm.org/doi/abs/10.5555/35517}{John Wiley $\&$ Sons, Inc.} USA~(1987).

\bibitem{Vollmer10}
Heribert Vollmer.
\newblock ``{Introduction to Circuit Complexity: A Uniform Approach}''.
\newblock Springer Publishing Company, Incorporated. ~(2010).
\newblock 1st edition.
\newblock  url:~\url{https://link.springer.com/book/10.1007/978-3-662-03927-4}.

\bibitem{Smolensky87}
R~Smolensky.
\newblock ``{Algebraic Methods in the Theory of Lower Bounds for Boolean Circuit Complexity}''.
\newblock In Proceedings of the Nineteenth Annual ACM Symposium on Theory of Computing.
\newblock \href{https://dx.doi.org/10.1145/28395.28404}{Pages 77--82}.
\newblock STOC '87New York, NY, USA~(1987). Association for Computing Machinery.

\bibitem{radhakrishnan1991}
Jaikumar Radhakrishnan.
\newblock ``{Better bounds for threshold formulas}''.
\newblock In [1991] Proceedings 32nd Annual Symposium of Foundations of Computer Science.
\newblock \href{https://dx.doi.org/10.1109/SFCS.1991.185384}{Pages 314--323}.
\newblock IEEE Computer Society~(1991).

\bibitem{Fischer82}
Michael~J Fischer, Albert~R Meyer, and Michael~S Paterson.
\newblock ``{$\Omega(N\log n)$ Lower Bounds on Length of Boolean Formulas}''.
\newblock \href{https://dx.doi.org/10.1137/0211033}{SIAM J. Comput. {\bf 11}, 416--427}~(1982).

\bibitem{Boaz09}
Sanjeev Arora and Boaz Barak.
\newblock ``{Computational Complexity: A Modern Approach}''.
\newblock Cambridge University Press. USA~(2009).
\newblock 1st edition.
\newblock  url:~\url{https://dl.acm.org/doi/abs/10.5555/1540612}.

\bibitem{Scott22}
Scott Aaronson.
\newblock ``{How Much Structure Is Needed for Huge Quantum Speedups?}''~(2022).
\newblock  \href{http://arxiv.org/abs/2209.06930}{arXiv:2209.06930}.

\bibitem{BARRINGTON1989}
David~A Barrington.
\newblock ``{Bounded-width polynomial-size branching programs recognize exactly those languages in NC1}''.
\newblock \href{https://dx.doi.org/https://doi.org/10.1016/0022-0000(89)90037-8}{Journal of Computer and System Sciences {\bf 38}, 150--164}~(1989).

\bibitem{Arkhipov11}
Scott Aaronson and Alex Arkhipov.
\newblock ``{The Computational Complexity of Linear Optics}''.
\newblock In Proceedings of the Forty-Third Annual ACM Symposium on Theory of Computing.
\newblock \href{https://dx.doi.org/10.1145/1993636.1993682}{Pages 333--342}.
\newblock STOC '11New York, NY, USA~(2011). Association for Computing Machinery.

\bibitem{Shor1999}
Peter~W Shor.
\newblock ``{Polynomial-Time Algorithms for Prime Factorization and Discrete Logarithms on a Quantum Computer}''.
\newblock \href{https://dx.doi.org/10.1137/S0036144598347011}{SIAM Review {\bf 41}, 303--332}~(1999).

\bibitem{simons97}
Daniel~R Simon.
\newblock ``{On the Power of Quantum Computation}''.
\newblock \href{https://dx.doi.org/10.1137/S0097539796298637}{SIAM Journal on Computing {\bf 26}, 1474--1483}~(1997).

\bibitem{Brassard06}
Gilles Brassard, Harry Buhrman, Noah Linden, Andr{\'{e}}~Allan M{\'{e}}thot, Alain Tapp, and Unger Falk.
\newblock ``{Limit on Nonlocality in Any World in Which Communication Complexity Is Not Trivial}''.
\newblock \href{https://dx.doi.org/10.1103/PhysRevLett.96.250401}{Phys. Rev. Lett. {\bf 96}, 250401}~(2006).

\bibitem{VanDam13}
Wim van Dam.
\newblock ``{Implausible consequences of superstrong nonlocality}''.
\newblock \href{https://dx.doi.org/10.1007/s11047-012-9353-6}{Natural Computing {\bf 12}, 9--12}~(2013).

\bibitem{Mosca19}
Matthew Amy and Michele Mosca.
\newblock ``{T-Count Optimization and Reed–Muller Codes}''.
\newblock \href{https://dx.doi.org/10.1109/TIT.2019.2906374}{IEEE Transactions on Information Theory {\bf 65}, 4771--4784}~(2019).

\bibitem{burgisser2013algebraic}
Peter B{\"{u}}rgisser, Michael Clausen, and Mohammad~A Shokrollahi.
\newblock ``{Algebraic complexity theory}''.
\newblock Volume 315.
\newblock Springer Science \& Business Media. ~(2013).
\newblock  url:~\url{https://dl.acm.org/doi/abs/10.5555/1965416}.

\bibitem{Low17}
Guang~Hao Low and Isaac~L. Chuang.
\newblock ``Optimal hamiltonian simulation by quantum signal processing''.
\newblock \href{https://dx.doi.org/10.1103/PhysRevLett.118.010501}{Phys. Rev. Lett. {\bf 118}, 010501}~(2017).

\bibitem{Haah2019}
Jeongwan Haah.
\newblock ``{Product {D}ecomposition of {P}eriodic {F}unctions in {Q}uantum {S}ignal {P}rocessing}''.
\newblock \href{https://dx.doi.org/10.22331/q-2019-10-07-190}{Quantum {\bf 3}, 190}~(2019).

\bibitem{scott21}
Scott Aaronson, Shalev Ben-David, Robin Kothari, Shravas Rao, and Avishay Tal.
\newblock ``{Degree vs. Approximate Degree and Quantum Implications of Huang's Sensitivity Theorem}''.
\newblock In Proceedings of the 53rd Annual ACM SIGACT Symposium on Theory of Computing.
\newblock \href{https://dx.doi.org/10.1145/3406325.3451047}{Pages 1330--1342}.
\newblock STOC 2021New York, NY, USA~(2021). Association for Computing Machinery.

\bibitem{Huang19}
Hao Huang.
\newblock ``{Induced subgraphs of hypercubes and a proof of the Sensitivity Conjecture}''.
\newblock \href{https://dx.doi.org/10.4007/annals.2019.190.3.6}{Annals of Mathematics {\bf 190}, 949--955}~(2019).

\bibitem{Ambainis17}
Andris Ambainis, Kaspars Balodis, Aleksandrs Belovs, Troy Lee, Miklos Santha, and Juris Smotrovs.
\newblock ``{Separations in Query Complexity Based on Pointer Functions}''.
\newblock \href{https://dx.doi.org/10.1145/3106234}{J. ACM{\bf 64}}~(2017).

\bibitem{Hoyer02}
Peter H{\o}yer and Robert {\v{S}}palek.
\newblock ``Quantum circuits with unbounded fan-out''.
\newblock In Helmut Alt and Michel Habib, editors, STACS 2003.
\newblock \href{https://dx.doi.org/https://doi.org/10.1007/3-540-36494-3_22}{Pages 234--246}.
\newblock Berlin, Heidelberg~(2003). Springer Berlin Heidelberg.

\bibitem{Austin22}
Austin~K Daniel, Yingyue Zhu, C~Huerta Alderete, Vikas Buchemmavari, Alaina~M Green, Nhung~H Nguyen, Tyler~G Thurtell, Andrew Zhao, Norbert~M Linke, and Akimasa Miyake.
\newblock ``{Quantum computational advantage attested by nonlocal games with the cyclic cluster state}''.
\newblock \href{https://dx.doi.org/10.1103/PhysRevResearch.4.033068}{Phys. Rev. Research {\bf 4}, 033068}~(2022).

\bibitem{Paul22}
Paul Herringer and Robert Raussendorf.
\newblock ``Classification of measurement-based quantum wire in stabilizer {PEPS}''.
\newblock \href{https://dx.doi.org/10.22331/q-2023-06-12-1041}{{Quantum} {\bf 7}, 1041}~(2023).

\bibitem{anand2022power}
Abhishek Anand.
\newblock ``{On the power of interleaved low-depth quantum and classical circuits}''.
\newblock Master's thesis.
\newblock University of Waterloo.
\newblock ~(2022).
\newblock  url:~\url{https://uwspace.uwaterloo.ca/handle/10012/18805}.

\bibitem{Preskill2018}
John Preskill.
\newblock ``{Quantum {C}omputing in the {NISQ} era and beyond}''.
\newblock \href{https://dx.doi.org/10.22331/q-2018-08-06-79}{Quantum {\bf 2}, 79}~(2018).

\bibitem{Demirel2021}
B{\"u}lent Demirel, Weikai Weng, Christopher Thalacker, Matty Hoban, and Stefanie Barz.
\newblock ``Correlations for computation and computation for correlations''.
\newblock \href{https://dx.doi.org/https://doi.org/10.1038/s41534-020-00354-2}{npj Quantum Information {\bf 7}, 1--8}~(2021).

\bibitem{Swain19}
Manoranjan Swain, Amit Rai, Bikash~K Behera, and Prasanta~K Panigrahi.
\newblock ``{Experimental demonstration of the violations of Mermin's and Svetlichny's inequalities for W and GHZ states}''.
\newblock \href{https://dx.doi.org/10.1007/s11128-019-2331-5}{Quantum Information Processing {\bf 18}, 218}~(2019).

\bibitem{Bo22}
Bo~Yang, Rudy Raymond, Hiroshi Imai, Hyungseok Chang, and Hidefumi Hiraishi.
\newblock ``Testing scalable bell inequalities for quantum graph states on ibm quantum devices''.
\newblock \href{https://dx.doi.org/10.1109/JETCAS.2022.3201730}{IEEE Journal on Emerging and Selected Topics in Circuits and Systems {\bf 12}, 638--647}~(2022).

\bibitem{Baccari20}
F.~Baccari, R.~Augusiak, I.~\ifmmode \check{S}\else \v{S}\fi{}upi\ifmmode~\acute{c}\else \'{c}\fi{}, J.~Tura, and A.~Ac\'{\i}n.
\newblock ``Scalable bell inequalities for qubit graph states and robust self-testing''.
\newblock \href{https://dx.doi.org/10.1103/PhysRevLett.124.020402}{Phys. Rev. Lett. {\bf 124}, 020402}~(2020).

\bibitem{Wei20}
Ken~X Wei, Isaac Lauer, Srikanth Srinivasan, Neereja Sundaresan, Douglas~T McClure, David Toyli, David~C McKay, Jay~M Gambetta, and Sarah Sheldon.
\newblock ``{Verifying multipartite entangled Greenberger-Horne-Zeilinger states via multiple quantum coherences}''.
\newblock \href{https://dx.doi.org/10.1103/PhysRevA.101.032343}{Phys. Rev. A {\bf 101}, 032343}~(2020).

\bibitem{Huang20}
Wei-Jia Huang, Wei-Chen Chien, Chien-Hung Cho, Che-Chun Huang, Tsung-Wei Huang, and Ching-Ray Chang.
\newblock ``{Mermin's inequalities of multiple qubits with orthogonal measurements on IBM Q 53-qubit system}''.
\newblock \href{https://dx.doi.org/https://doi.org/10.1002/que2.45}{Quantum Engineering {\bf 2}, e45}~(2020).

\bibitem{Sheffer22}
Meron Sheffer, Daniel Azses, and Emanuele~G {Dalla Torre}.
\newblock ``{Playing Quantum Nonlocal Games with Six Noisy Qubits on the Cloud}''.
\newblock \href{https://dx.doi.org/https://doi.org/10.1002/qute.202100081}{Advanced Quantum Technologies {\bf 5}, 2100081}~(2022).

\bibitem{Dunjko2016}
Vedran Dunjko, Theodoros Kapourniotis, and Elham Kashefi.
\newblock ``{Quantum-Enhanced Secure Delegated Classical Computing}''.
\newblock \href{https://dx.doi.org/https://dl.acm.org/doi/abs/10.5555/3179320.3179325}{Quantum Info. Comput. {\bf 16}, 61--86}~(2016).

\bibitem{Barz15}
Stefanie Barz, Vedran Dunjko, Florian Schlederer, Merritt Moore, Elham Kashefi, and Ian~A. Walmsley.
\newblock ``Enhanced delegated computing using coherence''.
\newblock \href{https://dx.doi.org/10.1103/PhysRevA.93.032339}{Phys. Rev. A {\bf 93}, 032339}~(2016).

\bibitem{Clementi17}
Marco Clementi, Anna Pappa, Andreas Eckstein, Ian~A Walmsley, Elham Kashefi, and Stefanie Barz.
\newblock ``Classical multiparty computation using quantum resources''.
\newblock \href{https://dx.doi.org/10.1103/PhysRevA.96.062317}{Physical Review A {\bf 96}, 062317}~(2017).

\bibitem{ahmed1975walsh}
Nasir Ahmed and Kamisetty~Ramamohan Rao.
\newblock ``{Walsh-hadamard transform}''.
\newblock In Orthogonal transforms for digital signal processing.
\newblock Pages 99--152.
\newblock Springer~(1975).

\bibitem{nielsen_chuang_2010}
Michael~A Nielsen and Isaac~L Chuang.
\newblock ``{Quantum Computation and Quantum Information: 10th Anniversary Edition}''.
\newblock \href{https://dx.doi.org/10.1017/CBO9780511976667}{Cambridge University Press}. ~(2010).

\bibitem{feinsilver2005krawtchouk}
Philip Feinsilver and Jerzy Kocik.
\newblock ``Krawtchouk polynomials and krawtchouk matrices''.
\newblock \href{https://dx.doi.org/10.1007/0-387-23394-6_5}{Pages 115--141}.
\newblock Recent Advances in Applied Probability. Springer US. Boston, MA~(2005).

\bibitem{feinsilver2018krawtchouk}
Philip Feinsilver and Rene Schott.
\newblock ``{Krawtchouk transforms and convolutions}''.
\newblock \href{https://dx.doi.org/https://doi.org/10.1007/s13373-018-0132-2}{Bulletin of Mathematical SciencesPages 1--19}~(2018).

\bibitem{stobinska2019quantum}
M.~Stobińska, A.~Buraczewski, M.~Moore, W.~R. Clements, J.~J. Renema, S.~W. Nam, T.~Gerrits, A.~Lita, W.~S. Kolthammer, A.~Eckstein, and I.~A. Walmsley.
\newblock ``Quantum interference enables constant-time quantum information processing''.
\newblock \href{https://dx.doi.org/10.1126/sciadv.aau9674}{Science Advances {\bf 5}, eaau9674}~(2019).

\bibitem{Kannan79}
Ravindran Kannan and Achim Bachem.
\newblock ``{Polynomial Algorithms for Computing the Smith and Hermite Normal Forms of an Integer Matrix}''.
\newblock \href{https://dx.doi.org/10.1137/0208040}{SIAM Journal on Computing {\bf 8}, 499--507}~(1979).

\bibitem{Alman21}
Josh Alman and Virginia~Vassilevska Williams.
\newblock ``A refined laser method and faster matrix multiplication''.
\newblock In Proceedings of the Thirty-Second Annual ACM-SIAM Symposium on Discrete Algorithms.
\newblock \href{https://dx.doi.org/10.1137/1.9781611976465.32}{Page 522–539}.
\newblock SODA '21USA~(2021). Society for Industrial and Applied Mathematics.

\end{thebibliography}

\onecolumn \newpage
\appendix

\section{Mathematical transformations}

\subsection{Walsh-Hadamard transform}\label{fourier}
The Walsh-Hadamard transform \cite{ahmed1975walsh,nielsen_chuang_2010} $\mathcal{F}$ is a Fourier transform that has their characters $X_{S}(x)$ defined as
\begin{equation}
  X_{S}(x) =
\left\{
\begin{array}{ll}
     \prod_{i\in S} x_i ,  & S \subseteq ([n]\setminus \emptyset) \\
    1,    & S = \emptyset
\end{array} 
\right. \ ,
\end{equation}
\noindent with $x\in \{1,-1\}^n$ and $S\subseteq [n]$. These characters are the real parity functions that are used to define the multi-linear polynomials (Definition \ref{multipoly}). Moreover, the matrix that defines the transformation is composed of all possible characters $X_{S}(x)$ that can be computed based on the size of the input string $S\subseteq [n]$ evaluated over all possible input values  $x\in \{1,-1\}^n$,
\begin{equation}
\mathcal{F}=\begin{bmatrix}
    X_{S_1}(x_1) &  X_{S_2}(x_1) &  X_{S_3}(x_1) & ... &  X_{S_{2^n}}(x_1) \\
   X_{S_1}(x_2) &  X_{S_2}(x_2) &  X_{S_3}(x_2) & ... &  X_{S_{2^n}}(x_2) \\
    X_{S_1}(x_3) &  X_{S_2}(x_3) &  X_{S_3}(x_3) & ... &  X_{S_{2^n}}(x_3) \\
    \vdots  & \vdots & \vdots & \hdots  & \vdots \\
     X_{S_1}(x_{2^n}) &  X_{S_2}(x_{2^n}) &  X_{S_2}(x_{2^n}) & ... &  X_{S_{2^n}}(x_{2^n}) 
\end{bmatrix}   \ .
\end{equation}

 Furthermore, the transposed matrix of the  Walsh-Hadamard transform applied to the column vector containing the truth table of the Boolean function generates the Fourier coefficients, 
\begin{equation}
\frac{1}{2^n}\begin{bmatrix}
    X_{S_1}(x_1) &  X_{S_1}(x_2) &  X_{S_1}(x_2) & ... &  X_{S_1}(x_{2^n}) \\
   X_{S_2}(x_1) &  X_{S_2}(x_2) &  X_{S_2}(x_2) & ... &  X_{S_2}(x_{2^n}) \\
    X_{S_3}(x_1) &  X_{S_3}(x_2) &  X_{S_3}(x_2) & ... &  X_{S_3}(x_{2^n}) \\
    \vdots  & \vdots & \vdots & \hdots  & \vdots \\
     X_{S_{2^n}}(x_1) &  X_{2^n}(x_2) &  X_{2^n}(x_2) & ... &  X_{2^n}(x_{2^n})
\end{bmatrix}   *
\begin{bmatrix}
    f(x_1) \\
    f(x_2)\\
    f(x_3)\\
    \vdots \\
    f(x_{2^n})
\end{bmatrix} \\ = 
\begin{bmatrix}
    \widehat{f}(S_1) \\
    \widehat{f}(S_2) \\
    \widehat{f}(S_3) \\
    \vdots \\
    \widehat{f}(S_{2^n}) \\
\end{bmatrix} \ .
\end{equation}

\noindent Given that, every element of the resulting column vector is derived from the equation 
\begin{equation}
    \widehat{f}(S)=\sum_{x_i\in \{0,1\}^n} f(x_i)X_S(g^{-1}(x_i)) \ ,
\end{equation}
\noindent which computes the Fourier coefficient of the Boolean function \cite{Odonnel2021}.

\subsection{Krawtchouk transform}\label{Krawtchouk}

The Krawtchouk transform \cite{feinsilver2005krawtchouk,feinsilver2018krawtchouk,stobinska2019quantum},  $\mathcal{K}$ is defined by a matrix with entries specified by the following expressions,

\begin{equation}
 K_{i,j}=
\begin{cases}
    \sum_{k=0}^{i} (-1)^{k}\binom{j}{k} \binom{n-j}{i-k},  &  i<j\ and\ i+j<=n\\
    \sum_{k=0}^{j} (-1)^{k}\binom{j}{k} \binom{n-j}{i-k},    &  i>=j \ and \ i+j<n \\
      \sum_{k=j+i-n}^{j} (-1)^{k}\binom{j}{k} \binom{n-j}{i-k},  &  i>=j\  and\ i+j>=n \\
       \sum_{k=0}^{n-j} (-1)^{i-k}\binom{j}{i-k} \binom{n-j}{k},  &  i<j\  and\ i+j>n \\
\end{cases} 
\ .
\end{equation}

\noindent The matrix has, therefore, the following form,
\begin{equation}
  \mathcal{K}^{[n]}=  \begin{bmatrix}
   \binom{n}{0} &  \binom{n}{1}& ... & \binom{n}{i} & ...   & \binom{n}{n} \\
   \binom{n-1}{0} &  -\binom{n-1}{0}+\binom{n-1}{1}& ... &-\binom{n-1}{1}+\binom{n-1}{2} & ...    & -\binom{n-1}{n-1} \\
   \binom{n-2}{0} &  -2\binom{n-2}{0}+\binom{n-2}{1}& ... &\binom{n-2}{0}-2\binom{n-2}{1}+\binom{n-2}{2} & ...  & \binom{n-2}{n-2} \\
   \vdots & \vdots & \vdots & \vdots  & \vdots & \vdots \\
    \binom{n-j}{0}  & \hdots & \hdots & rf_{i,j}=\sum_{k=0}^{j} (-1)^{j-k}\binom{j}{k} \binom{n-j}{i-k}  & \hdots & (-1)^{j}\binom{n-j}{n-j} \\
   \vdots & \vdots & \vdots & \vdots  & \vdots & \vdots \\
    (-1)^0\binom{n}{n} &  (-1)^1\binom{n}{n-1} &\hdots  & (-1)^i\binom{n}{n-i} &  \hdots & (-1)^{i=j}\binom{n}{0} \\
\end{bmatrix} 
\end{equation}
\noindent with dimension $n+1\times n+1$.

\section{Proof of Lemma \ref{decomposition} }\label{lemma2proof}

This section proves that any complete symmetric function $C^k$ can be decomposed into a conjunction formed uniquely with CSF with a degree equal to a power of two $C^{2^r}$. We obtain the following lemma for single conjunction.

\begin{lemma}\label{lemma4}
 The composition of two complete symmetric functions with degrees $p$ and $q$ equals the complete symmetric function with degree $p+q$, 
\begin{equation}
    C^p \wedge C^q = C^{p+q}
\end{equation}
\noindent if $q$ is a power of two and $p$ is a number composed of powers two larger than $q$, i.e. $p=\sum_{i\in N} 2^{n_i},\  q= 2^{m},\forall i \in N \  (n_i>m)$
\end{lemma}

\begin{proof}
We proceed by strong induction. 
\vspace{0.2cm}

\noindent \textit{Base case:} 
\begin{equation}
    C^p \wedge C^1 = C^{p+1}, with\ p=\sum_{i\in N} 2^{n_i},\forall_{i \in N} \  (n_i>0)
\end{equation}
\noindent To prove this equality, we resort to its ANFs,
 \begin{equation}\label{basecase}
  \underbrace{\Big ( \underset{1\leq i_1 < i_2 < ... \leq k}{\bigoplus} x_{i_1} \wedge x_{i_2} \wedge ... \wedge x_{i_p} \Big)}_{Exp_1}\ \wedge\  \underbrace{  \overset{k}{\underset{i=0}{\bigoplus} }\ x_i }_{Exp_2} = \underset{1\leq i_1 < i_2 < ... \leq k}{\bigoplus} x_{i_1} \wedge x_{i_2} \wedge ... \wedge x_{i_{p+1}}  
 \end{equation} 
\noindent  and apply the AND operator to verify that LHS equals the RHS. However, the expression will contain terms that are not only simple linear sums of AND operators composed of single bits. The following rule will reduce the terms with linear sums inside the AND operators possible

\begin{equation}\label{rule}
    (a \oplus b) \wedge c = (a \wedge c) \oplus (b \wedge c)  \tag{R1}
\end{equation}

\noindent By applying repeatedly Rule \ref{rule} to the LHS of Equation \eqref{basecase} we obtain an expression with a term for each element of $Exp_2$ composed with $Exp_1$

\begin{align}
\begin{split}
   & C^p \wedge C^1=  \\ & \Big (\Big (\underbrace{\underset{1\leq i_1 < i_2 < ... \leq k}{\bigoplus} x_{i_1} \wedge x_{i_2} \wedge ... \wedge x_{i_p} \Big )}_{Exp1} \wedge x_1 \Big ) \oplus \Big (\Big(\underbrace{  \underset{1\leq i_1 < i_2 < ... \leq k}{\bigoplus} x_{i_1} \wedge x_{i_2} \wedge ... \wedge x_{i_p} \Big ) }_{Exp1} \wedge x_2 \Big) \\ & \oplus \ \ \ \ \ \ \ \ \ \ \ \ \  \ \ \ \ \ \ \   \cdot \cdot \cdot \ \ \ \ \ \ \ \ \ \ \ \ \ \ \ \ \ \ \ \ \ \ \ \ \  \oplus \Big( \Big ( \underbrace{ \underset{1\leq i_1 < i_2 < ... \leq k}{\bigoplus} x_{i_1} \wedge x_{i_2} \wedge ... \wedge x_{i_p} \Big ) }_{Exp1} \wedge x_k  \Big)
\end{split}
\end{align}
\noindent The same rule is applied to each of the terms in $Exp_1$ with the single term of $Exp_2$. At the end, each term of $Exp_1$ is composed with all elements contained in $Exp_2$
\begin{equation}
\begin{split}
&C^p(x) \wedge C^1(x) =\\ & \ \ \ \mathbf{x_{1}} \wedge x_{2} \wedge \hdots \wedge x_{p} \wedge \mathbf{x_1} \  \oplus\ \ \  x_{1} \wedge \mathbf{x_{2}} \wedge \hdots \wedge x_{p} \wedge \mathbf{x_2} \ \ \ \oplus ... \oplus   x_{1} \wedge x_{2} \wedge \hdots \wedge x_{p} \wedge \mathbf{x_k} \oplus \\ & \mathbf{x_{1}}  \wedge ... \wedge x_{p-1} \wedge x_{p+1} \wedge \mathbf{x_1}\oplus x_{1}  \wedge \hdots \wedge x_{p-1} \wedge x_{p+1} \wedge \mathbf{x_2}\oplus  ... \oplus x_{1}  \wedge ... \wedge x_{p+1} \wedge \mathbf{x_k}\oplus \\ & \hspace{7cm} \cdots \\ &  x_{k-p}  \wedge ... \wedge x_{k-1} \wedge x_{k} \wedge \mathbf{x_1} \oplus x_{k-p}  \wedge ... \wedge x_{k-1} \wedge x_{k} \wedge \mathbf{x_2} \oplus ... \oplus  x_{k-p}  \wedge ... \wedge x_{k-1} \wedge \mathbf{x_{k}} \wedge \mathbf{x_k}
\end{split}
\end{equation}
\noindent In order to reduce this expression to a unique ANF, it will be necessary  first to remove repetitions of bits in each term. For instance, the first term of $Exp_1$ composed with the first term of $Exp_2$ repeats the $x_1$ bit.

Afterward, it will be necessary to count the number of times each term does repeat itself. If a term occurs an odd number of times, it remains and has an exclusive existence in the ANF; if it occurs an even number of times, it disappears. In particular, there will be only terms of size $p$, for which the unique element introduced from $Exp_2$ was already part of the term, and the terms of size $p+1$ for which this term was new. The terms of size $p$ occur once each time that the new bit is equal to an existing one
\begin{equation}
    \#\big(\big|x_i\wedge x_j\wedge ... \wedge x_n\big|=p\big)= \binom{p}{1} =p
\end{equation}
\noindent and as  $p$ is by definition an even number. This implies that no term of size $p$ exists in its final ANF, which is consistent with the RHS of Equation \eqref{basecase}.  

Finally, to verify the equality, it has to be guaranteed that each element of RHD of the equation exists and appears an odd number of times in the LHS of \eqref{basecase}. The first condition is verified because all terms of size $p$ exist in $Exp_1$, and the product with $Exp_2$ produces all possible combinations of terms with size $p+1$. It remains to prove that the terms with $p+1$ arguments appear an odd number of times. For that, let us consider a unique term,
\begin{equation}
     x_{1} \wedge x_{2} \wedge ... \wedge x_{p} \wedge \mathbf{x_i},\ i \notin \{1,2,...,p\}
\end{equation}
\noindent  with $p+1$ different indexes can occur repeatedly due to the composition of another product term that initially shared $p$ indexes and was composed with the missing bit. Therefore each term with $p+1$ different bits is created as
\begin{equation}
     \#\big(\big|x_i\wedge x_j\wedge ... \wedge x_n\big|=p+1\big)= \binom{p+1}{p} = p+1
\end{equation}
\noindent which is an odd number by definition, proving the base case. 

\vspace{0.2cm}
\noindent \textit{Induction step:}
\begin{equation}
 \forall_{y<m+1}\  C^p  \wedge C^{2^y}  = C^{p+2^y}   \implies  C^p \wedge C^{2^{m+1}}  =  C^{p+2^{m+1}}
  \ ,\ with\ p=\sum_{i\in N} 2^{n_i},\ \forall_{i \in N}  \ m+1>n_i
\end{equation}

To prove that the RHS of the consequent equality is equal to its LHS, we rewrite all these functions in the ANF. The procedure is the one used in the base case resorting to Rule \ref{rule} repeatedly. Subsequently, it is necessary to count the number of times this process creates each term to verify its inclusion in the ANF. However, this process is more complex for these cases than for the base case because there are significantly more repetitions, and the decomposition generates product terms with a mixture of lengths. Once more, we resort to the symmetry argument to analyze only one product term for each size. Therefore, we obtain potential terms from the size $2^{m+1}$ to $2^{m+1}+p$, given that they could repeat from $p$ to $0$ times in the expression. Furthermore, the number of times each of these terms repeats is equal to
\begin{align}
    &\#\big ( \big |x_{i_1} \wedge x_{i_2} \wedge ... \wedge x_{i_p}\big | = 2^{m+1} \big )= \binom{2^{m+1}}{p}\ repetitions \\
     &\#\big ( \big |x_{i_1} \wedge x_{i_2} \wedge ... \wedge x_{i_p} \wedge \mathbf{x_{j_1}} \big | = 2^{m+1}+1 \big ) = \binom{2^{m+1}}{p-1}*\binom{1}{1}\ repetitions \\
    &\hspace{5cm} \hdots \nonumber
    \\
       &\#\big ( \big |x_{i_1} \wedge x_{i_2} \wedge ... \wedge x_{i_p} \wedge \mathbf{x_{j_1} \wedge ... \wedge x_{j_{q-1}}}\big |= 2^{m+1}+p-1 \big ) =\binom{2^{m+1}}{1}*\binom{p-1}{p-1}\ repetitions \\
     & \#\big ( \big |x_{i_1} \wedge x_{i_2} \wedge ... \wedge x_{i_p} \wedge \mathbf{x_{j_1} \wedge ... \wedge x_{j_{q}}} \big | = 2^{m+1}+p\big )= \binom{2^{m+1}}{0}*\binom{p}{0}\ repetitions . \\
\end{align} 
It is easy to verify that terms of size $2^{m+1}+p$ exist only once. However, to verify the equality, it is necessary to prove that all other terms of smaller dimensions exist an even number of times. In order to achieve that, it helps to write the modular division by two of these number of combinations with two additional description presented in \cite{Mori2019},
\begin{equation}
  C^{2^r}(x)=  \binom{|x|}{2^r}\ mod\ 2  =Bit^{r}\Big(|x|\Big)
\end{equation}
\noindent showing that the CSF with a degree equal to a base of two has two equivalent representations. One as a modular division by two of a combinatorial term involving the Hamming weight of the input string, and another as a specific bit $r$ of the binary representation of the Hamming weight $|x|$ which will be designated $Bit^r(|x|)$. With these descriptions and the hypothesis for the induction step, we obtain that
\begin{equation}\label{even}
    \binom{2^{m+1}}{p}\ mod\ 2= \bigwedge_{i\in N} Bit^{n_i}(2^{m+1}) =0
\end{equation}
\noindent with $p=\sum_{i\in N} 2^{n_i}$ and as the binary representation of $2^{m+1}$ has only the $(m+1)$th bit equal to $1$, and all other bits and combinations of bits of the binary representation, smaller than $m+1$, are all equal to $0$. This proves that the expression in Equation \eqref{even} is always zero by the definition of the number $p$, and therefore the terms of length $2^{m+1}$ vanish. The same argument can be used for all terms of length $2^{m+1}-1$ to $2^{m+1}+p-1$, as the the modular division by two of their number of repetitions will be always an even number. This statement holds because
\begin{equation}\label{oddd}
    \binom{2^{m+1}}{\{1,2,...,p\}}\ mod\ 2= \bigwedge_{i\in S\subseteq N} Bit^{n_i}(2^{m+1}) =0
\end{equation}
\noindent which concludes the proof that all terms of length smaller then $2^{m+1}+p$ vanish, verifying the induction step.
\end{proof}

Finally to prove Lemma \ref{decomposition}, it is only necessary to obtain the binary representation of $k$ and, with Lemma \ref{lemma4} compose all the base of two components of $k$ sequentially in increasing size order. 
\begin{equation}
   \bigg(\Big(\big( C^{2^{i_1}} \wedge C^{2^{i_2}}\big) \wedge C^{2^{i_3}}\Big) ... \wedge  C^{2^{i_n}} \bigg)=  C^{2^{i_1}+2^{i_2}+2^{i_3}+...+2^{i_n}} = C^{k},\ with\  i_1<i_2<...<i_n
\end{equation}

\noindent obtaining exactly the decomposition stated by Lemma \ref{decomposition}.

\section{Proof of Lemma \ref{polinomials}}\label{polynomials}

In order to verify that the first instances of the general form presented in Proposition \ref{decomp} compute the corresponding CSF for any possible input, it helps to manipulate the general expression given by Equation \eqref{expr12}. The goal is to rewrite this expression such that it can be expressed in function of the Hamming weight $\sum_{i=1}^n x_i = |x|$, which is sufficient to determine the result of the function, and the size of the input string $n$ that is necessary to determine the specific polynomial. In particular, the polynomials are defined, for all ${x\in \{0,1\}^n}$, as
 \begin{align}
 \begin{split}
     \mathsf{poly}_{C^{k}}(x)=\frac{\pi}{2^{k-1}}\Bigg (\sum_{j=1}^{k/2} \binom{n-k/2-j}{k/2-j}  (-1)^j \bigg( \sum_{S_i\subseteq [n],|S_i|=j} \bigoplus_{i \in S_i} x_i    - \sum_{S_i \subseteq [n] , |S_i|=n-j+1}  \bigoplus_{i \in S_i} x_i  \ \bigg)  \Bigg)\ .
\end{split}
 \end{align}
\noindent For this transformation, the coefficients remain equal, but the transformed parity bases ($\bigoplus_{i \in S} x_i$) change to a different expression that evaluates to the same values based on the newly selected parameters,
\begin{equation}
    \sum_{S_i\subseteq [n]\wedge |S_i|=j} \bigoplus_{i \in S_i} x_i = \sum_{k=0}^{\frac{2j+(-1)^j+1}{2}} \binom{|x|}{2k+1}*\binom{n-|x|}{j-2k-1}\ .
\end{equation}
\noindent  It will also be useful to distinguish the sum of the transformed parity bases described with sets of size $n-j+1$, which are defined as 
 \begin{align}
    \sum_{S_i \subseteq [n] \wedge |S_i|=n-j+1}  \bigoplus_{i \in S_i} x_i  = \left\{
\begin{array}{ll}
     \sum_{l=0}^{\frac{2j+(-1)^j+1}{2}} \binom{|x|}{2l+1}*\binom{n-|x|}{j-2l-1},  & |x|\ mod\ 2 =0 \\
    -\sum_{l=0}^{\frac{2j+(-1)^j+1}{2}} \binom{|x|}{2l+1}*\binom{n-|x|}{j-2l-2},    & |x|\ mod\ 2 =1
\end{array} 
\right. \ .
\end{align}

\noindent Using this transformation and considering the case of $|x|\ mod\ 2 =0$, the general expression in Equation \eqref{expr12} becomes, for all  $|x|\in \{0,2,4,...,n\}$ and $n\in \mathbb{N}$,  

\begin{align}\label{expr3}
 \begin{split}
    \mathsf{poly}_{C^{k}}(n,|x|)=\frac{\pi}{2^{k-1}}\Bigg (\sum_{j=0}^{k/2} \binom{n-k/2-j}{k/2-j}& (-1)^j \Bigg( \sum_{l=0}^{\frac{2j+(-1)^j+1}{2}} \binom{|x|}{2l+1}*\binom{n-|x|}{j-2l-1} \\ &-\sum_{l=0}^{\frac{2j+(-1)^{j-1}-1}{2}} \binom{|x|}{2l+1}*\binom{n-|x|}{j-2l-2}  \Bigg )\Bigg)\ .
\end{split}
\end{align}

\noindent For $|x|\ mod\ 2 =1$ it can be rewritten for all $|x|\in \{1,3,5,...,n\}$ and $n\in \mathbb{N}$, as
\begin{align}\label{expr4}
 \begin{split}
   \mathsf{poly}_{C^{k}}(n,|x|)=\frac{\pi}{2^{k-1}}\Bigg (\sum_{j=0}^{k/2} \binom{n-k/2-j}{k/2-j} & (-1)^j \Bigg (\sum_{l=0}^{\frac{2j+(-1)^j+1}{2}} \binom{|x|}{2l+1}*\binom{n-|x|}{j-2l-1}\\  &+\sum_{l=0}^{\frac{2j+(-1)^{j-1}-1}{2}} \binom{|x|}{2l+1}*\binom{n-|x|}{j-2l-2} \Bigg )\Bigg)\ .
\end{split}
\end{align}

\noindent Fortunately, with the use of Wolfram Mathematica code we could reduce the expressions. In particular, it was possible to show that the dependency on $n$ disappears, and the resulting expression is equivalent to 
\begin{equation}\label{equal}
     \mathsf{poly}_{C^{k}}(n,|x|)=
     \left\{
\begin{array}{ll}
     \frac{\pi}{2^{\frac{k}{2}}*\frac{k}{2}!} \prod_{i=0}^{\frac{k}{2}} (|x|-2i)= \pi*\binom{\frac{|x|}{2}}{\frac{k}{2}},  & |x|\ mod\ 2 =0 \\
   \frac{\pi}{2^{\frac{k}{2}}*\frac{k}{2}!} \prod_{i=0}^{\frac{k}{2}} (|x|-1-2i)= \pi*\binom{\frac{|x|-1}{2}}{\frac{k}{2}},    & |x|\ mod\ 2 =1 \ ,
\end{array} 
\right.
\end{equation}
\noindent with $|x|\in \mathbb{N}$, $n\in \mathbb{N}$, $k=2,4,8,16,32$ and $64$. Now it remains to show that this simpler expression computes the intended functions, and for that, we resort again the $Bit^l$ notation used in Appendix \ref{lemma2proof}

\begin{equation}
    \mathsf{poly}_{C^k}(n,|x|)=
      \left\{
\begin{array}{ll}
     \binom{\frac{|x|}{2}}{\frac{k}{2}}\ mod\ 2=Bit^{log(k/2)}\big(\frac{|x|}{2}\big),  & |x|\ mod\ 2 =0 \\
   \binom{\frac{|x|-1}{2}}{\frac{k}{2}}\ mod\ 2=Bit^{log(k/2)}\big(\frac{|x|-1}{2}\big),    & |x|\ mod\ 2 =1
\end{array}
\right  .
\end{equation}
\noindent Knowing that
\begin{equation}
    Bit^l(x)=Bit^{l-1}(x/2)
\end{equation}
\noindent we obtain
\begin{equation}
    \mathsf{poly}_{C^k}(n,h)=
      \left\{
\begin{array}{ll}
     Bit^{k}\big(|x|\big),  & |x|\ mod\ 2 =0 \\
   Bit^{k}\big(|x|-1\big),    & |x|\ mod\ 2 =1
\end{array} 
\right.
\end{equation}
\noindent Moreover, for any $k>1$ and $|x|\ mod\ 2=1$
\begin{equation}
    Bit^k(|x|-1)=Bit^k(|x|)
\end{equation}
\noindent such that
\begin{align}
     \mathsf{poly}_{C^k}(n,|x|)&= Bit^k(|x|) = \binom{|x|}{k}mod\ 2= C^k(n,|x|)
     \\ &\equiv C^k(x)
\end{align}
\noindent  as intended.

\section{Proof of Lemma \ref{rfbound} } \label{FouriertoKrawt}

In order to prove the lemma, we will simplify the Fourier transform sequentially with conversions that do not impact on the computation of the Fourier coefficients of symmetric semi-Boolean functions until we obtain the Krawtchouk transform. 

First, we will reduce the number of lines in the transform by exploiting a symmetry between the Fourier coefficients, as stated below.

\begin{lemma}\label{equalcoef}
All symmetric semi-Boolean functions, $f_{\ast}:\{0,1\}^n\rightarrow \mathbb{R}$, have equal Fourier coefficients $\widehat{f}(S)$ for all the parity bases ($\bigoplus_{l\subseteq S} x_l$) defined by sets with the same size
\begin{equation}
   S_i,S_j \in [n] ,\ |S_i|=|S_j| \implies \widehat{f}(S_i)=\widehat{f}(S_j)
\end{equation}
\end{lemma}

\begin{proof}
The equation used to compute the Fourier coefficient of symmetric semi-Boolean function can be rewritten using the simplified value vector of the function,
\begin{align}
    \widehat{f}(S)&=\sum_{x_i\in \{0,1\}^n} f(x_i)X_S(x_i) \\ &= \sum_{i=0}^n v_f(i) \sum_{x\in\{0,1\}^n \wedge |x|=i } X_S(x) \ .
\end{align}
\noindent From this equation, it can be verified that the values of the $v_f$ term are always equal for the same function. Therefore, it is sufficient to show that for all parity bases defined by sets $S_i,S_j \subseteq [n]$, of the same size $|S_i|=|S_j|$, the following equality holds
\begin{equation}\label{eqequal}
   \sum_{x\in\{0,1\}^n \wedge |x|=i } X_{S_i}(x) =\sum_{x\in\{0,1\}^n \wedge |x|=i } X_{S_j}(x)\ .
\end{equation}
\noindent This equality is very simple to verify as all the input strings that evaluate to $1$, one of the parity basis will evaluate to $1$ the other for a specific permutation ($p$). This permutation has only to exchange the bits of the input string selected by one of the sets with the bits selected by the other set. Then for all $x\in\{0,1\}^n$ and $|x|=i$, 
\begin{equation}
  X_{S_i}(x)=X_{S_j}(p(x))\ .
\end{equation}

\noindent Furthermore, as no permutation $p$ applied to the input strings changes the number of strings that are evaluated to $1$ by the parity basis, the sum of before and after the permutation is exactly equal, 
\begin{equation}
   \sum_{x\in\{0,1\}^n \wedge |x|=i } X_{S_j}(p(x)) =\sum_{x\in\{0,1\}^n \wedge |x|=i } X_{S_j}(x)\ .
\end{equation}

\noindent Finally, this proves that equality of Equation $\eqref{eqequal}$ provided one of the permutations described previously.  
\end{proof}

\noindent This property shows that at most $n+1$ different Fourier coefficients exist, one for each possible size of the parity bases. Thus, it is sufficient to use $n+1$ lines of the Fourier transform to determine exactly those values, 
\begin{equation}
    \mathcal{F}'=\begin{bmatrix}
     X_{|S_i|=0}(x_1) &  X_{|S_i|=0}(x_2) & X_{|S_i|=0}(x_3) & ... & X_{|S_i|=0}(x_{2^n}) \\
   X_{|S_i|=1}(x_1) &  X_{|S_i|=1}(x_2) & X_{|S_i|=1}(x_3) & ... & X_{|S_i|=1}(x_{2^n}) \\
    X_{|S_i|=2}(x_1) &  X_{|S_i|=2}(x_2) & X_{|S_i|=2}(x_3) & ... & X_{|S_i|=2}(x_{2^n}) \\
    \vdots  & \vdots & \vdots & \hdots  & \vdots \\
     X_{|S_i|=n}(x_1) &  X_{|S_i|=n}(x_2) & X_{|S_i|=n}(x_3) & ... & X_{|S_i|=n}(x_{2^n}) 
\end{bmatrix}\ .
\end{equation}

For the second reduction of the transformation matrix, it is necessary to prove that the Fourier coefficients can be determined based on the simplified value vector of the function $v_f$. Hence, we start with the equation computing Fourier coefficients using a general line of the Fourier transform and apply the necessary transformation with some additional simplifications,  
\begin{align}
  \widehat{f}(S)&=\sum_{x_i\in \{0,1\}^n} f(x_i)X_S(x_i)\\
  &= \sum_{i\in \{0,1\}^n}  v_f(|x_i|) *X_S(x_i)
  \\&= \sum_{x' \in \{0,1\}^{|S|}} \sum_{y'\in \{0,1\}^{n-|S|}} v_f(|x'|+|y'|)* X_S(x')\\
  &= \sum_{x' \in \{0,1\}^{|S|}} \sum_{y'\in \{0,1\}^{n-|S|}} v_f(|x'|+|y'|)* (-1)^{|x'|}*\binom{|S|}{|x'|} \\
  &= \sum_{i=0}^{|S|} \sum_{k=0}^{n-|S|} v_f(i+k) * \binom{n-|S|}{k}*(-1)^i*\binom{|S|}{i}\ .
\end{align}
\noindent The resulting equation is already written based on the simplified value vector and the size of the Fourier basis. However, it is necessary to isolate the term corresponding to the simplified value vector of the function so that one can define the elements of a matrix that computes the Fourier coefficients if applied to that object. For that we resort to the following change of variables $j=i+k$, yielding
\begin{align}
   \widehat{f}(S)&= \sum_{i=0}^{|S|} \sum_{j=0}^{n} v_f(j) * (-1)^i*\binom{|S|}{i}*\binom{n-|S|}{j-i}\\&= \sum_{j=0}^{n} v_f(j) *\sum_{i=0}^{|S|}(-1)^i*\binom{|S|}{i}* \binom{n-|S|}{j-i} \ .
\end{align}
\noindent The last equation allows us to define the lines of a matrix that computes the Fourier coefficients when applied to the simplified value vector, 

\begin{equation}
\begin{bmatrix}
    K_{1,1} &  K_{1,2} &   K_{1,3} & ... &   K_{1,n+1} \\
   K_{2,1} &  K_{2,2} &   K_{2,3} & ... &   K_{2,n+1} \\
    K_{3,1} &  K_{3,2} &   K_{3,3} & ... &   K_{3,n+1} \\
    \vdots  & \vdots & \vdots & \hdots  & \vdots \\
      K_{n+1,1} &  K_{n+1,2} &   K_{n+1,3} & ... &   K_{n+1,n+1}
\end{bmatrix}   *
\begin{bmatrix}
    v_f(0) \\
    v_f(1)\\
    v_f(2)\\
    \vdots \\
    v_f(n)
\end{bmatrix} \\ = 
\begin{bmatrix}
    \widehat{f}(|S|=0) \\
    \widehat{f}(|S|=1) \\
    \widehat{f}(|S|=2) \\
    \vdots \\
    \widehat{f}(|S|=n) \\
\end{bmatrix}\ ,
\end{equation}
\noindent with the following rule defining the entries of the matrix for the resulting transformation, 

\begin{equation}
 K_{i,j}= \sum_{k=0}^{j} (-1)^{k}\binom{j}{k} \binom{n-j}{i-k}\ .
\end{equation}
\noindent  The rule obtained for each matrix element is precisely the rule used to build the Krawtchouk transformation, proving that one can compute the Fourier coefficient with this transformation.

In the end, to prove the bound on the number of operations, it is necessary to determine the number of operations necessary to build the matrix. In particular, each element results from the sum of combinatorial terms, which compute factorial terms. These can be exponentially bigger than $n$. Nevertheless, they have at most $n^2$ digits, which are numbers that are simple to operate with for efficient number multiplication algorithms. The latter are bounded by $\mathcal{O}\big(n\ log(n)\big)$. Therefore, the operations necessary to compute each element of the Krawtchok transform are bounded by  $\mathcal{O}(n^3)$. Then, these are computed for all matrix elements, which pushes the upper bound to $\mathcal{O}(n^5)$.

 \section{\texorpdfstring{$\textnormal{NMQC}_\oplus$}{NQMC} measurement assignment of a CSF: An Example } \label{examplecsf}

The complete symmetric function $C^5:\{0,1\}^n\rightarrow\{0,1\}$ is defined, for all $x\in \{0,1\}^n$, as

\begin{equation}
    C^5(x)=\bigoplus_{i_1=1}^{n-4} x_{i_1} \bigg(   \bigoplus_{i_2=i_1+1}^{n-3} x_{i_2}  \bigg( \bigoplus_{i_3=i_2+1}^{n-2} x_{i_3} \bigg (  \bigoplus_{i_4=i_3+1}^{n-1} x_{i_4} \bigg (  \bigoplus_{i_5=i_4+1}^{n} x_{i_4}\bigg )\bigg ) \bigg ) \bigg) \ .
\end{equation}

\noindent This can be decomposed as the conjunction of two CSF with a degree equal to a power of two as:
\begin{equation}
    C^5(x)=C^4(x)\wedge C^1(x) = \bigoplus_{i_1=1}^{n-3} x_{i_1} \bigg(   \bigoplus_{i_2=i_1+1}^{n-2} x_{i_2}  \bigg( \bigoplus_{i_3=i_2+1}^{n-1} x_{i_3} \bigg (  \bigoplus_{i_4=i_3+1}^{n} x_{i_4}\bigg ) \bigg ) \bigg) \bigg ( \bigoplus_{j=0}^n x_j \bigg) \ .
\end{equation}

\noindent Now resorting to the multi-linear polynomials, with domain $\{-1,1\}^n$, from Lemma \ref{polinomials} and Proposition \ref{decomp}, one can represent this family of Boolean functions as the following family of multi-linear polynomials,
\begin{equation}
\begin{split}
    \mathsf{poly}_{C^5}(x)=&\frac{\pi}{32}  \bigg(\frac{(n-4)(n-3)}{2} \Big(1-\prod_{j=1}^n x_j\Big) - (n-4)\Big( \sum_{i=1}^n x_i + \sum_{i=1}^n \Big(\big(\prod_{j=1}^n x_j\big)*x_i\Big)\Big)\\ &+ \Big( \sum_{i=1}^n \sum_{j=i}^n x_i*x_j\Big)  \Big) -\Big( \sum_{i=1}^n \sum_{j=i}^n \big(\prod_{j=1}^n x_j\big)*x_i*x_j\Big)  \bigg) \ .
\end{split}
\end{equation}

Let us consider a specific instance of this family with a fixed size for the input strings ($|x|=6$) such that the following polynomial represents the function, 
\begin{equation}
\begin{split}
    \mathsf{poly}_{C^5}(x)= &\frac{\pi}{32} \Big(3 - 3*\prod_{j=1}^6 x_j - 2\sum_{i=1}^6 x_i +2 \sum_{i=1}^6 \Big(\big(\prod_{j=1}^6 x_j\big)*x_i\Big)  \\ & + \Big( \sum_{i=1}^6 \sum_{j=i}^6 x_i*x_j\Big)  \Big) -\Big( \sum_{i=1}^6 \sum_{j=i}^6 \big(\prod_{j=1}^6 x_j\big)*x_i*x_j\Big)  \Big) \ .
\end{split}
\end{equation}

\noindent This multi-linear polynomial translates to a measurement assignment using Equations \ref{linears}, \ref{angles} and \ref{fmeasure}. We can then act on a 43-qubit GHZ state, $\ket{\Psi_{GHZ}^{43}}= 1/\sqrt{2} \big(\ket{0}^{\otimes 43} + \ket{1}^{\otimes 43} \big)$, with the following linear functions and measurement operators 
\begin{align}
    L_{i\in \{1,..,6\}}&\equiv x_j,\ \ \ \ \ \ \ \ \ \ \ \ \ \ \ \ \ \ \ with \ M_i(s_i)= \cos(\frac{\pi}{8}*s_i) \sigma_x + \sin(\frac{\pi}{8}*s_i) \sigma_y\\ 
    L_{i\in \{7,..,12\}}&\equiv x_j \oplus \bigoplus_{k=1}^6 x_k,\ \ \ \ \ \ \  with \ M_i(s_i)= \cos(-\frac{\pi}{8}*s_i) \sigma_x + \sin(-\frac{\pi}{8}*s_i) \sigma_y \\
    L_{i\in \{13,..,27\}}&\equiv x_j \oplus x_l ,\ \ \ \ \ \ \ \ \ \ \ \ \ with \ M_i(s_i)= \cos(-\frac{\pi}{16}*s_i) \sigma_x + \sin(-\frac{\pi}{16}*s_i) \sigma_y \\
    L_{i\in \{28,..,42\}}& \equiv x_j \oplus x_l \oplus \bigoplus_{k=1}^6 x_k,\ with \ M_i(s_i)= \cos(\frac{\pi}{16}*s_i) \sigma_x + \sin(\frac{\pi}{16}*s_i) \sigma_y\\
    L_{43}&\equiv \bigoplus_{k=1}^6 x_k ,\ \ \ \ \ \ \ \ \ \ \ \ \ \  with \ M_i(s_i)= \cos(\frac{3\pi}{16}*s_i) \sigma_x + \sin(\frac{3\pi}{16}*s_i) \sigma_y\ ,
\end{align}
\noindent for which the linear functions are defined as a general instance, and the specific linear functions are all the combinations of this expression built with the $6$ inputs bits of the input string. The second best solution regarding resource requirements is the solution presented in \cite{Frembs22}, which requires $63$ qubits. In particular, this solution would require all the linear functions of the form $x_i \oplus x_j \oplus x_k$, which become unnecessary in the presented solution.

\section{Proof of Lemma \ref{ltest}}\label{test}
We will prove the lemma introducing a more detailed description of the problem and the algorithm to solve it. In the end, the algorithm's running time is shown to be polynomially bounded. 

The algorithm has to determine whether there exits a symmetric multi-linear polynomial with a particular set of non-zero Fourier coefficients, $F=\big\{ \widehat{f}(|S|=i)\ \big|\ i\in T,\ T \subseteq [n] \big\}$, that computes deterministically the SBF. Therefore, if this polynomial exists it will be an element of the set of  multi-linear polynomials equivalent to the SBF in the \textnormal{NMQC}$_\oplus$ model. Thus, it will be helpful to redefine this set in in a similar way to what we have done previously for the minimal sparsity in Definition \ref{minspar}. First, consider that any semi-Boolean function $g:\{0,1\}^n\rightarrow \mathbb{Z}$ that verifies the following equality, 
\begin{equation}\label{equivfunc}
    g(x)=f(x)+2h(x)\ ,
\end{equation}
\noindent for any other semi-Boolean function $h(x):\{0,1\}^n\rightarrow \mathbb{Z}$,  has a multi-linear polynomial that computes the Boolean function $f(x):\{0,1\}^n\rightarrow\{0,1\}$ deterministically in the \textnormal{NMQC}$_\oplus$ model. This allows us to use the semi-Boolean function $h(x)$ as the degree of freedom that determines the elements of the equivalence set. Furthermore this degree of freedom can be restricted for our specific problem. First, only symmetric multi-linear polynomials are considered. Therefore if $f(x)$ is symmetric, then also $h(x)$ has to be symmetric to generate symmetric semi-Boolean functions equivalent to $f(x)$ through Equation \eqref{equivfunc}. In addition, the same equality on the functions will be true for the corresponding multi-linear polynomials. Therefore, the elements of the set of symmetric multi-linear polynomials $\mathsf{poly}_{g(x)}$ equivalent to the symmetric Boolean function $f(x)$ in the \textnormal{NMQC}$_\oplus$ model, must verify the following equality, 
\begin{equation}
   \sum_{i=0}^n\widehat{g}\big(|S|=i\big)= \sum_{i=0}^n\widehat{f}\big(|S|=i\big)+ 2 \sum_{i=0}^n\widehat{h}\big(|S|=i\big)\ .
\end{equation}

A second restriction imposed by the problem is that only a specific subset of Fourier coefficients $T \subseteq [n]$ can be non-zero. Therefore all other Fourier coefficients of the symmetric multi-linear polynomial have to be equal to zero. This implies that all the valid symmetric multi-linear polynomials correspond respectively to instances of $h(x)$ that verify simultaneously the following equality $\widehat{f}(|S|=i)=-2\widehat{h}(|S|=i) $ for all $i\in E$, with $E=[n]\setminus T$. This guarantees that   $\widehat{g}(|S|=i)$ is equal to zero for all $i\in E$. Finally, a solution to the problem exists if there is a semi-Boolean function $h(x)$ for which the following equation is verified, 
\begin{equation}
    \sum_{x} X_{|S|=i}(x)f(x)=-\sum_{x} X_{|S|=i}(x)2h(x)
\end{equation}
\noindent for all $i \in E$. This problem defines a system of linear Diophantine equations. Therefore, determining if a symmetric measurement assignment exists reduces to determining if this system of linear Diophantine equations has a solution. Thus, a solution to the decision problem will always exist. However, we should define if it can be determined in polynomial time with respect to the size of the function. Unfortunately, the equations are of exponential size in relation to the input string because they represent specific lines of the Fourier transform. Therefore, any algorithm that receive these equations as input will need to manipulate an exponential size object. This means that solving the problem would also take exponential time. 

Fortunately, the problem can be simplified using the Krawtchouk transform which computes the Fourier coefficient using equations of polynomial size. All the Fourier coefficients will be determined using specific lines of the Krawtchouk transform applied to the simplified value vectors of the functions $f$ and $h$, designated $v_{f}$ and $v_{h}$. Consequently, the problem reduces to finding a solution to the following equation,

\begin{equation}
    \sum_{j=0}^n K_{i,j}^{[n]}*v_{f}(j)=-\sum_{j=0}^n K_{i,j}^{[n]}*2*v_{h}(j)
\end{equation}
 for all $i\in E$. The combination of all these equations defines a system of Diophantine equations that can be represented as,

\begin{equation}
   \begin{bmatrix}
   K_{i_1,1} & K_{i_1,2} & K_{i_1,3} & ... & K_{i_1,n+1}\\
   K_{i_2,1} & K_{i_2,2} & K_{i_2,3} & ... & K_{i_2,n+1} \\
   K_{i_3,1} & K_{i_3,2} & K_{i_3,3} & ... & K_{i_3,n+1} \\
   \vdots & \vdots & \vdots & \vdots  & \vdots & \vdots \\
    K_{i_k,1} & K_{i_k,2} & K_{i_k,3} & ... & K_{i_k,n+1} \\
\end{bmatrix} *2
\begin{bmatrix}
    v_h(0) \\
    v_h(1)\\
    v_h(2)\\
    \vdots \\
    v_h(n)
\end{bmatrix} \\ = 
\begin{bmatrix}
   \widehat{ f}(|S|=i_1) \\
   \widehat{ f}(|S|=i_2) \\
    \widehat{f}(|S|=i_3) \\
    \vdots \\
    \widehat{f}(|S|=i_k) \\
\end{bmatrix}  \ .
\end{equation}

In order to verify if a system of Diophantine equations has a solution, one can use the Smith Normal Form of the matrix defining the system. In particular, the matrix defining the system of Diophantine equations required to be solved will be the submatrix of the Krawtchouk transform composed of the lines selected by the set $E$, $\mathcal{K}^{[n+1]}[E;1,2,...,n+1]$. Therefore, the Smith Normal Form of this matrix creates three unique matrices that verify the following equality,  
\begin{equation}
    U^{-1}D V^{-1}= \mathcal{K}^{[n]}[E;1,2,...,n+1] \ .
\end{equation}
\noindent Then the system of equations to be solved can be decomposed as,
\begin{align}
    U^{-1}DV^{-1} * v_h[E]&= [\widehat{f}(E)]\\
    U^{-1}DV^{-1}*2*v_h'[E]&= [\widehat{f}(E)]\\
    V^{-1}*v_h'[E]=\frac{D^{-1}U[\widehat{f}(E)]}{2}
\end{align}
\noindent with $[f\widehat{(E)}]$ being the set of Fourier coefficients selected by $E$, resulting from $T$. Consequently, a solution for the system exists if for all $i\in E$,
\begin{equation}\label{tests}
  \frac{(U[\widehat{f}(|S|=i)])}{2D[i,i]} \in \mathbb{Z} \ .
\end{equation}

Finally, to prove that this process can be performed in polynomial time, it will be necessary to show that each step is polynomially bound. First, the Smith decomposition can be performed with algorithms that scale polynomially concerning the size of the matrix, as proven in \cite{Kannan79}. This solution demonstrates that determining the Smith Normal Form of the Krawtchouk transform can be done in polynomial time and space, given that the matrix scales polynomially with the size of the input string. Then, all the matrix and vector multiplications are polynomially bounded, provided these objects are of polynomial size \cite{Alman21}. The last process tests if all elements of the final vector, composed of the elements described by Equation \eqref{tests}, are integer numbers. The previously mentioned vector is smaller or equal to $n$, which also makes this subprocess polynomially bounded. Therefore, the entire process is polynomially bounded, proving the lemma.

\section{Support for Conjecture \ref{lowerbound}}\label{conjecture}
The conjecture is supported by the numerical verification process described below. First, the polynomially scaling algorithm described in the proof of Lemma \ref{ltest}  was implemented, in a way such that one could test the order of magnitude upon whether a solution exists or not. This was necessary because the number of combinations of Fourier coefficients that can be selected to be tested grows exponentially. However, applying the algorithm to a subset of these combinations suffices to test the scaling rate. The last statement is supported by the number of parity bases corresponding to each non-zero Fourier coefficient. Moreover, a non-zero Fourier coefficient relative to parity bases of size $t=|S|$ imposes the use of $\binom{n}{|S|}$ qubits and measurements.  Therefore, each Fourier coefficient is related to a certain asymptotic growth rate, 

\begin{equation}
   \overbrace{ \underbrace{\binom{n}{0}}_{1}}^{ \widehat{f}({|S|=0})},\  \overbrace{\underbrace{\binom{n}{1}}_{n}}^{\widehat{f}({|S|=1})},\ \overbrace{\underbrace{\binom{n}{2}}_{\mathcal{O}\big(n^2\big)}}^{\widehat{f}({|S|=2})},\ \overbrace{\underbrace{\binom{n}{3}}_{\mathcal{O}\big(n^3\big)}}^{\widehat{f}({|S|=3})},\ \ \ \ \hdots \ \ \ \ ,  \overbrace{\underbrace{\binom{n}{n-2}}_{\mathcal{O}\big(n^2\big)}}^{\widehat{f}({|S|=n-2})},\ \overbrace{ \underbrace{\binom{n}{n-1}}_{n}}^{\widehat{f}({|S|=n-1})},\ \overbrace{\underbrace{\binom{n}{n}}_{1}}^{\widehat{f}({|S|=n})} \ .
\end{equation}

In order to test if there is a solution that has at most a certain growth rate $\mathcal{O}(n^t)$, it is sufficient to test if all the Fourier coefficients with larger growth rates could be made equal to zero. Therefore, this method reduced the number of necessary tests necessary to determine the order of the growth rate to a maximum of $n/2$ for each symmetric function. In particular, for the conjecture it was enough to test if there are solutions with a growth rate of $\mathcal{O}(n^{\mathsf{deg}(f)/2}-1)$ and $\mathcal{O}(n^{\mathsf{deg}(f)/2})$. Therefore, all Fourier coefficients with larger scaling rates than these values must be zero for a positive answer. This statement implies that the following equation  
\begin{equation}
    \sum_{j=0}^n K_{i,j}^{[n]}*v_{f}(j)=-\sum_{j=0}^n K_{i,j}^{[n]}*2*v_{h}(j)
\end{equation}
\noindent with $[n]$ defining the size and $i,j$ the element of the Krawtchouk transform, is verified for all $\{i \big|\  i>\mathsf{deg}(f),\  i>n-\mathsf{deg}(f)/2 \}$ if the function has a scaling rate of $\mathcal{O}(n^{\mathsf{deg(f)/2-1}})$. Equivalently, for $\mathcal{O}(n^{\mathsf{deg(f)/2}})$ the same equation has to be true for all $\{i \big|\  i>\mathsf{deg}(f)+1,\  i>n-\mathsf{deg}(f)/2-1 \}$. 

The numerical tests realized verified that all complete symmetric functions that can be defined with an input size up to 345 bits have solutions with the respective sparsities scaling with $\mathcal{O}(n^{\mathsf{deg(f)/2}})$ and failed to have a solution with $\mathcal{O}(n^{\mathsf{deg(f)/2-1}})$. Moreover, the main limitation of the numerical process is that for each specific CSF, only instances of a particular size of the input string are tested each time. One can not use this test to verify the sparsities of all possible input strings of a particular CSF simultaneously. 

In summary, a reasonable number of functions have verified the conjectured scaling relating to their minimum sparsity. Furthermore, the same conjecture and numerical test are consistent with the lower bound found for two particular symmetric Boolean functions. More numerical tests could be conducted. We believe that actually proving the conjecture will be possible resorting to a general Smith decomposition of the Krawtchouk transform.

\end{document}